\def\arxiv{1}

\documentclass[11pt,letterpaper]{article}
\usepackage[margin=1in]{geometry}
\usepackage[utf8]{inputenc}
\usepackage[colorlinks=true,allcolors=blue]{hyperref}
\usepackage{amsmath}
\usepackage{amsfonts}
\usepackage{amsthm}
\usepackage{braket}
\usepackage{authblk}
\usepackage{physics}
\usepackage{comment}
\usepackage[colorlinks=true,allcolors=blue]{hyperref}
\usepackage{breakcites}
\usepackage{cleveref}
\usepackage[sanserif,basic]{complexity}
\usepackage{graphicx}
\usepackage{subcaption}

\newtheorem{theorem}{Theorem}[section]

\newtheorem{lemma}[theorem]{Lemma}
\newtheorem{claim}[theorem]{Claim}

\newtheorem{corollary}[theorem]{Corollary}

\newtheorem{remark}[theorem]{Remark}
\newtheorem{problem}{Problem}

\AddToHook{env/remark/begin}{\crefalias{theorem}{remark}}
\AddToHook{env/lemma/begin}{\crefalias{theorem}{lemma}}
\AddToHook{env/proposition/begin}{\crefalias{theorem}{proposition}}
\AddToHook{env/claim/begin}{\crefalias{theorem}{claim}}
\AddToHook{env/corollary/begin}{\crefalias{theorem}{corollary}}
\DeclareMathOperator*{\Ex}{\mathbb{E}}

\crefname{claim}{Claim}{Claims}
\crefname{theorem}{Theorem}{Theorems}
\crefname{lemma}{Lemma}{Lemmas}
\crefname{proposition}{Proposition}{Propositions}
\crefname{corollary}{Corollary}{Corollaries}
\crefname{appendix}{appendix}{appendices}
\Crefname{appendix}{Appendix}{Appendices}

\AddToHook{cmd/appendix/before}{%
  \crefalias{section}{appendix}%
  \crefalias{subsection}{appendix}%
}

\title{Worst-Case Quantum Algorithm for Optimal Polynomial Intersection Beyond  Decoded Quantum Interferometry}
\ifnum\arxiv=1
\author[1]{Shuji Horinaga}
\author[2]{Takashi Yamakawa}
\affil[1]{NTT Institute for Fundamental Mathematics}
\affil[2]{NTT Social Informatics Laboratories}
\else
\author{}
\fi

\begin{document}
\maketitle
\begin{abstract}
The Optimal Polynomial Intersection (OPI) problem asks us to find a low-degree polynomial over a finite field whose values lie in prescribed subsets on as many given inputs as possible. Decoded quantum interferometry (DQI) gives a quantum algorithm for OPI in parameter regimes beyond those achieved by the best known classical heuristics. Follow-up works improve the parameter regimes, but their analyses are limited to average-case settings. Recently, Sun and Wootters showed that, even in the worst case, OPI has a solution in a larger parameter regime than the one covered by DQI. However, they left open whether one can design a quantum algorithm that solves OPI in the worst-case beyond the DQI regime.

We give such a quantum algorithm. As a byproduct, we also improve the existential bound of Sun and Wootters in certain parameter regimes. In particular, when each subset contains roughly half of the field elements, our algorithm finds a solution with satisfaction rate $s=1$ whenever the rate satisfies $R>0.75$. This matches the previous average-case bound, whereas DQI cannot achieve $s=1$ unless $R=1$. Our existential bound guarantees the existence of a solution when $R> 0.7158$, improving over the previous threshold $R>0.7495$. More generally, our existential results extend to the Max-LINSAT problem with respect to arbitrary maximum distance separable (MDS) codes. The corresponding algorithmic results apply only to MDS codes whose dual admits an efficient list decoder. Our results are obtained through a novel application of a Brascamp--Lieb-type inequality in the MDS setting, which may have further applications.

\end{abstract}
\thispagestyle{empty}
\clearpage 

\newcommand{\takashi}[1]{{\color{red}(\textbf{Takashi}: #1)}}
\newcommand{\revise}[1]{{\color{red}#1}}
\newcommand{\Shuji}[1]{{\color{blue}(\textbf{Shuji}: #1)}}

\newcommand{\one}{\mathbf{1}}
\def \sample { \overset{\hspace{0.1em}\mathsf{\scriptscriptstyle\$}}{\leftarrow} }
\newcommand{\ra}{\rightarrow}
\newcommand{\la}{\leftarrow}
\newcommand{\negl}{\mathsf{negl}}
\newcommand{\secpar}{\lambda}
\newcommand{\A}{\mathcal{A}}
\newcommand{\B}{\mathcal{B}}
\newcommand{\siml}{\mathcal{S}}
\newcommand{\ora}{\mathcal{O}}
\newcommand{\bit}{\{0,1\}}
\newcommand{\td}{\mathsf{td}}
\newcommand{\defeq}{:=}
\newcommand{\oracle}{\mathcal{O}}
\newcommand{\enc}{\mathsf{Enc}}
\newcommand{\dec}{\mathsf{Dec}}
\newcommand{\calI}{\mathcal{I}}
\newcommand{\calS}{\mathcal{S}}
\newcommand{\calX}{\mathcal{X}}
\newcommand{\calY}{\mathcal{Y}}
\newcommand{\calZ}{\mathcal{Z}}
\newcommand{\calR}{\mathcal{R}}
\newcommand{\vecr}{{\mathbf{r}}}
\newcommand{\vecx}{\mathbf{x}}
\newcommand{\vecy}{\mathbf{y}}
\newcommand{\vecz}{\mathbf{z}}
\newcommand{\matA}{{\mathbf{A}}}
\newcommand{\matB}{{\mathbf{B}}}
\newcommand{\matD}{{\mathbf{D}}}
\newcommand{\reprogram}{\mathsf{Reprogram}}
\newcommand{\Func}{\mathsf{Func}}
\newcommand{\TD}{\mathsf{TD}}
\newcommand{\Samp}{\mathsf{Samp}}
\newcommand{\experiment}{\mathsf{Exp}}
\newcommand{\semiconst}{\mathsf{SC}}
\newcommand{\perm}{\mathsf{Perm}}
\newcommand{\FF}{\mathbb{F}}
\newcommand{\vecv}{\mathbf{v}}
\renewcommand{\vu}{\mathbf{u}}
\newcommand{\vx}{\mathbf{x}}
\newcommand{\vy}{\mathbf{y}}
\newcommand{\vz}{\mathbf{z}}
\newcommand{\ve}{\mathbf{e}}
\newcommand{\vc}{\mathbf{c}}
\newcommand{\vm}{\mathbf{m}}
\newcommand{\mA}{\mathbf{A}}
\newcommand{\mG}{\mathbf{G}}
\newcommand{\mH}{\mathbf{H}}
\newcommand{\valpha}{\boldsymbol \alpha}
\newcommand{\GRS}{\mathrm{GRS}}
\newcommand{\GRSDecode}{\mathrm{GRSDecode}}
\newcommand{\wt}{\mathsf{wt}}
\newcommand{\GRSDecodeError}{\mathrm{GRSDecodeError}}
\newcommand{\GRSListDecode}{\mathrm{GRSListDecode}}
\newcommand{\on}{\mathsf{on}}
\newcommand{\off}{\mathsf{off}}
\newcommand{\hw}{\mathsf{hw}}
\newcommand{\HW}{\mathcal{H}\mathcal{W}}
\newcommand{\HWl}{\HW_{\leq \alpha n }}
\newcommand{\HWg}{\HW_{> \alpha n }}
\newcommand{\etag}{\eta_{1}}
\newcommand{\etale}{\eta_{0}}
\newcommand{\bad}{\mathsf{BAD}}
\newcommand{\good}{\mathsf{GOOD}}
\newcommand{\rand}{\mathsf{rand}}
\newcommand{\MaxLINSAT}{\mathrm{MaxLINSAT}}
\newcommand{\OPI}{\mathrm{OPI}}
\newcommand{\mfol}{(m)}
\newcommand{\decode}{\mathsf{Decode}}
\newcommand{\encode}{\mathsf{Encode}}
\newcommand{\RS}{\mathrm{RS}}
\newcommand{\vzero}{\mathbf{0}}
\newcommand{\keylength}{{\ell_\mathsf{key}}}
\newcommand{\inlength}{{\ell_\mathsf{in}}}
\newcommand{\outlength}{{\ell_\mathsf{out}}}
\newcommand{\pilength}{{\ell_{\pi}}}
\newcommand{\accept}{\mathsf{acc}}
\newcommand{\poq}{\mathsf{poq}}
\newcommand{\concat}{||}
\newcommand{\Approx}{\mathsf{Approx}}
\newcommand{\gooderrors}{\mathcal{G}}
\newcommand{\baderrors}{\mathcal{B}}
\newcommand{\dist}{\mathcal{D}}
\newcommand{\bardist}{\widetilde{\mathcal{D}}}
\newcommand{\bardistprime}{\widetilde{\mathcal{D}'}}
\newcommand{\hashset}{\widetilde{\mathcal{H}}}
\newcommand{\Col}{\mathrm{Col}}

\newcommand{\poqsetup}{\mathsf{PoQRO}.\mathsf{Setup}}
\newcommand{\poqprove}{\mathsf{PoQRO}.\mathsf{Prove}}
\newcommand{\poqverify}{\mathsf{PoQRO}.\mathsf{Verify}}
\newcommand{\pk}{\mathsf{pk}}
\newcommand{\sk}{\mathsf{sk}}
\newcommand{\gen}{\mathsf{Gen}}
\newcommand{\prove}{\mathsf{Prove}}
\newcommand{\verify}{\mathsf{Verify}}
\newcommand{\prob}{\mathsf{prob}}
\newcommand{\sol}{\mathsf{sol}}
\newcommand{\QFT}{\mathsf{QFT}}
\newcommand{\learner}{\mathcal{L}}
\newcommand{\efficientlearner}{\widetilde{\mathcal{L}}}
\newcommand{\Lout}{L_{\mathsf{out}}}

\newcommand{\ecrhsetup}{\mathsf{ECRH}.\mathsf{Setup}}
\newcommand{\ecrhgen}{\mathsf{ECRH}.\mathsf{Gen}}
\newcommand{\ecrheval}{\mathsf{ECRH}.\mathsf{Eval}}
\newcommand{\ecrhequiv}{\mathsf{ECRH}.\mathsf{Equiv}}
\newcommand{\crs}{\mathsf{crs}}

\newcommand{\sigkeygen}{\mathsf{Sig}.\mathsf{KeyGen}}
\newcommand{\sigsign}{\mathsf{Sig}.\mathsf{Sign}}
\newcommand{\sigverify}{\mathsf{Sig}.\mathsf{Verify}}
\newcommand{\vk}{\mathsf{vk}}
\newcommand{\sigk}{\mathsf{sigk}}

\newcommand{\pkekeygen}{\mathsf{PKE}.\mathsf{KeyGen}}
\newcommand{\pkeenc}{\mathsf{PKE}.\mathsf{Enc}}
\newcommand{\pkedec}{\mathsf{PKE}.\mathsf{
Dec}}
\newcommand{\ek}{\mathsf{ek}}
\newcommand{\dk}{\mathsf{dk}}
\newcommand{\ct}{\mathsf{ct}}
\newcommand{\st}{\mathsf{st}}

\newcommand{\win}{\mathsf{win}}
\newcommand{\chal}{\mathcal{C}}
\newcommand{\inp}{\mathsf{inp}}

\tableofcontents
\thispagestyle{empty}
\clearpage
\pagenumbering{arabic}

\section{Introduction}
Quantum algorithms are believed to solve certain problems significantly faster than classical algorithms, a phenomenon often referred to as quantum advantage. Notable examples include integer factorization and the discrete logarithm problem, both of which can be solved in quantum polynomial time by Shor's algorithm~\cite{FOCS:Shor94}, while no polynomial-time classical algorithms are known for them. 

Recently, a new family of quantum algorithms~\cite{EC:CheLiuZha22,JACM:YamZha24,TQC:ChaTil24,DQI,STOC:ChaTil25,chailloux2025opixsoftdecoders,SODA:BDGJLS26,C:ChaHer26} has emerged, originating from Regev's quantum reduction~\cite{JACM:Regev09} between lattice problems. Among these, the algorithm known as decoded quantum interferometry (DQI)~\cite{DQI} is believed to achieve quantum advantage for a problem called the optimal polynomial intersection (OPI) problem. The OPI problem is defined as follows. 
\begin{problem}[Optimal Polynomial Intersection Problem]\label{prob:OPI}
Let $q$ be a prime power, let $k,n$ be integers with $0 \le k \le n \le q$, 
let $\valpha=(\alpha_1,\alpha_2,\ldots,\alpha_n) \in \mathbb{F}_q^n$ be distinct, 
let $S_i \subseteq \mathbb{F}_q$ for each $i \in [n]$, and let $s\in [0,1]$.

The optimal polynomial intersection problem
$\OPI(k,q,\valpha,\{S_i\}_{i\in[n]},s)$ is the following problem:
\begin{description}
    \item[Given:] $k$, $q$, $\valpha$, $s$, and quantum oracle access to the membership oracle of $S_i$ for each $i\in[n]$.\footnote{Existing works appear to treat the subsets $\{S_i\}_{i\in [n]}$ as being given explicitly, with the regime $q\le \poly(n)$ in mind. In contrast, in our formulation we provide only quantum oracle access to the membership oracle of each $S_i$. This choice is motivated by the desire to capture settings in which $q$ is superpolynomial, so that an explicit description of the sets may be too large. We note that existing algorithms can also be implemented using oracle access rather than an explicit description of the sets.}
    \item[Find:] A polynomial $P \in \mathbb{F}_q[X]$ such that $\deg(P)<k$ and
    \[
        \left|\{i\in[n] : P(\alpha_i)\in S_i\}\right| \ge s n .
    \]
\end{description}
\end{problem}
 
  For any $\valpha=(\alpha_1,\alpha_2,\ldots,\alpha_n) \in \mathbb{F}_q^n$ consisting of distinct elements and any subsets $S_i \subseteq \mathbb{F}_q$ satisfying $|S_i|=\rho q$ for every $i \in [n]$, DQI solves
$\OPI(k,q,\valpha,\{S_i\}_{i\in[n]},s)$ if the following is satisfied:
\[
R>2\left(1-\left(\sqrt{s \rho}+\sqrt{(1-s)(1-\rho)}\right)^2\right)
\]
where $R=k/n$.

It has been argued that this performance surpasses that of the best known heuristic classical algorithm, namely Prange's algorithm~\cite{Prange62}, thereby providing a candidate for quantum advantage.

While this is remarkable, it is natural to ask whether these parameters are optimal. For example, in the typical case where $\rho= 1/2$, the above bound does not allow one to achieve $s= 1$ unless $R= 1$, which is  meaningless.

In this direction, Chailloux~\cite{chailloux2025opixsoftdecoders} improved the bound in the average-case setting as follows:
\[
R> 1-\left(\sqrt{s \rho}+\sqrt{(1-s)(1-\rho)}\right)^4.
\]
This is indeed an improvement since $2(1-x^2)-(1-x^4)=(1-x^2)^2\ge 0$ for any $x\in \mathbb{R}$. 
This improvement is significant: it allows one to achieve $s=1$ even when $\rho= 1/2$, provided that $R> 0.75$.

However, an important caveat of this result is that the success guarantee holds only over a random choice of the sets $\{S_i\}_{i\in[n]}$. More precisely, what is shown is that, in the parameter regime described above, there exists a quantum polynomial-time algorithm that solves
$\OPI(k,q,\valpha,\{S_i+c_i\}_{i\in[n]},s)$ 
with overwhelming probability over the independent choices of shifts $c_i\gets \mathbb{F}_q$ for each $i\in[n]$, where $S_i+c_i=\{x+c_i: x\in S_i\}$. 
Thus, this analysis does not apply to worst-case choices of $\{S_i\}_{i\in[n]}$. To date, DQI remains the only known quantum algorithm with a worst-case success guarantee for OPI.

\paragraph{Existential bound for OPI.}
Putting quantum algorithms aside, it is also natural to ask for which parameters OPI admits a solution at all. In this direction, Sun and Wootters~\cite{SunWootters26} studied conditions for the existence of solutions to OPI. They proved that, in the worst case, OPI admits solutions beyond the parameter regime in which DQI is guaranteed to work.

In particular, they formulate their results in terms of the following two quantities:\footnote{They use the notation $\mu_1(\rho)$ and $\mu_0(\rho)$ to denote one half of $R_1(\rho)$ and $R_0(\rho)$ in our notation, respectively.}
\begin{description}
\item[Saturation bound:]
For $\rho\in(0,1)$, the saturation bound $R_1(\rho)$ is the smallest value of $R$ for which OPI admits an almost perfect solution, that is, a solution with $s\approx 1$.
\item[Improvement bound:]
For $\rho\in (0,1)$, the improvement bound $R_0(\rho)$ is the smallest value of $R$ for which the existential bound improves upon the algorithmic bound achieved by DQI.
\end{description}

For example, they prove that $R_1(1/2)<0.7495$. 
This is notable because even the average-case algorithmic result \cite{chailloux2025opixsoftdecoders} requires $R> 0.75$ to achieve $s\approx 1$ when $\rho= 1/2$. They also show that $R_0(1/2)<0.6225$.

However, their result is purely existential: it guarantees only the existence of solutions. They leave open the problem of turning this existential improvement into an algorithmic improvement over DQI in the worst-case setting.

\subsection{Our Result} 
We give the first worst-case quantum algorithm for OPI that improves upon DQI in certain parameter regimes.
 Since the full condition for our algorithm is somewhat complicated, we state here a simplified version and spell out the condition explicitly in the balanced case $\rho= 1/2$, which is the main case of interest.\footnote{Strictly speaking, when $q$ is odd, we cannot have $|S|=\frac{1}{2}q$. Thus, whenever we write $\rho=1/2$ in the introduction, it should be understood as meaning
$|S|=\left(\frac{1}{2}\pm o(1)\right)q.$} See \Cref{thm:main} for the full statement, and see \Cref{rem:case_of_rho_half} for how \Cref{thm:main_intro} is derived from \Cref{thm:main}. 

\begin{theorem}[Informal; balanced case]\label{thm:main_intro}
Let $q=p^e$ be a prime power such that $p=\omega(1)$. Let $n,k$ be positive integers such that
$n/2\le k\le n\le q$, and let $R=k/n$. 

\begin{itemize}
\item If $R>0.75$, then there exists a quantum polynomial-time algorithm that solves
$\OPI(k,q,\valpha,\{S_i\}_{i\in[n]},1)$ for any distinct $\alpha_1,\ldots,\alpha_n\in\FF_q$ and any subsets $S_i\subseteq\FF_q$ such that $|S_i|\approx q/2$ for all $i\in[n]$. 

\item For $R\le 0.75$ and $s\in (1/2,1)$, 
suppose that the following are satisfied with a constant margin: 
\begin{align}
R> 1-\left(\sqrt{\frac{s}{2}}+\sqrt{\frac{1-s}{2}}\right)^4 \label{eq:theorem_cond_one_intro}
\end{align}
and
\begin{align}
H(R)+R\log \left(3\left(\frac{2s-1}{\sqrt{3}}\right)^{\frac{1}{1-R}}\right)<0, \label{eq:theorem_cond_two_intro}
\end{align}
where
$H(R)=-R\log R-(1-R)\log(1-R)$ 
is the binary entropy function. 

Then there exists a quantum polynomial-time algorithm that solves
$\OPI(k,q,\valpha,\{S_i\}_{i\in[n]},s)$ for any distinct $\alpha_1,\ldots,\alpha_n\in\FF_q$ and any subsets $S_i\subseteq\FF_q$ such that $|S_i|\approx q/2$ for all $i\in[n]$. 
\end{itemize}
\end{theorem}

\begin{figure}[t]
    \centering
    \includegraphics[width=0.8\linewidth]{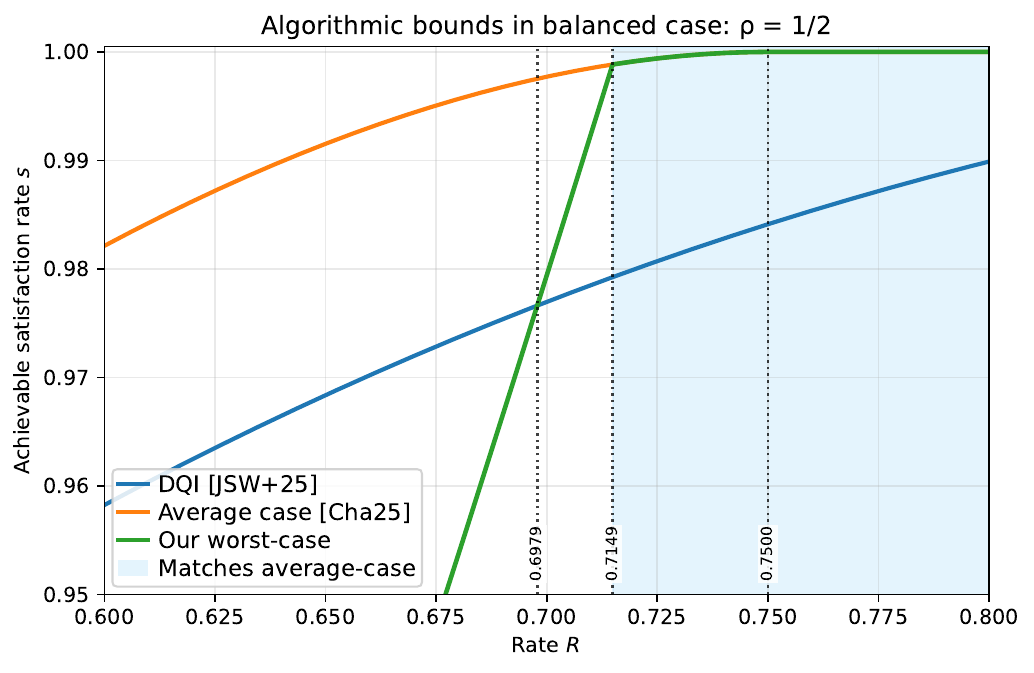}
    \caption{Comparison of achievable algorithmic satisfaction rates in the balanced case $\rho=1/2$.}
    \label{fig:balanced-case-comparison}
\end{figure}
\begin{figure}[t]
    \centering
    \includegraphics[width=0.8\linewidth]{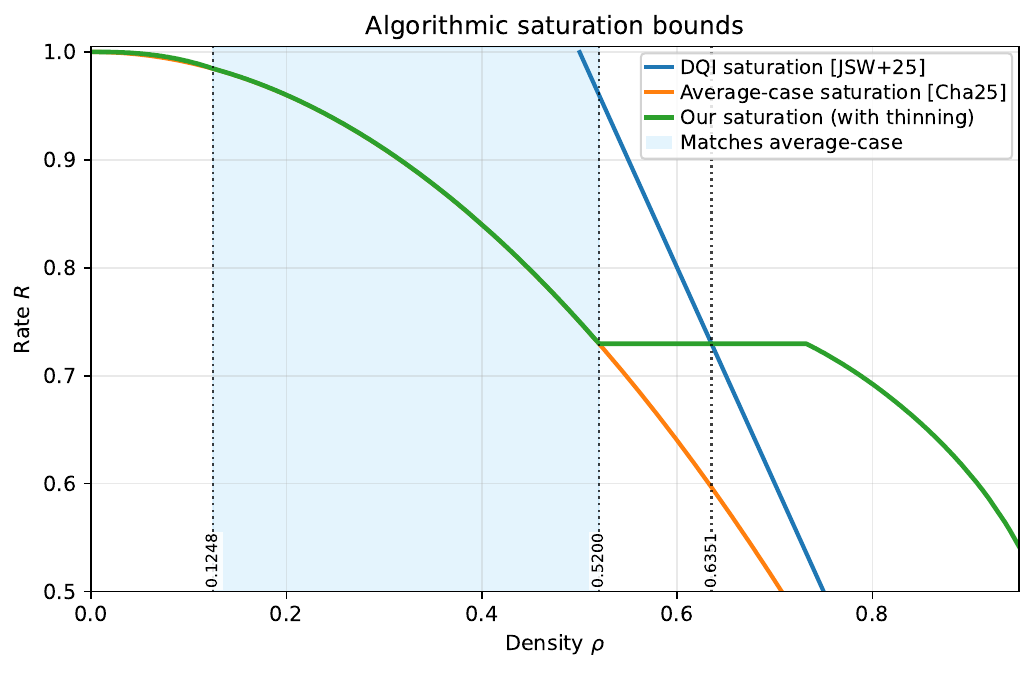}
    \caption{Comparison of algorithmic saturation bounds as functions of $\rho$. The direct application of \Cref{thm:main} yields a non-monotone curve, with a bump for $\rho>0.5200$. This bump is only an artifact of the proof. Indeed, by a simple thinning argument (see \Cref{lem:thinning}), we may replace the sets by smaller subsets, and hence replace $\rho$ by a smaller value, without reducing the target satisfaction rate. Thus, for $\rho>0.5200$, we use the horizontal bound obtained by applying \Cref{thm:main} at $\rho\approx 0.5200$, wherever this improves upon the bound obtained from a direct application of \Cref{thm:main} at the actual value of $\rho$.
    }
    \label{fig:saturation-bounds}
\end{figure}

See \Cref{fig:balanced-case-comparison} for a comparison with DQI and the average-case algorithm of \cite{chailloux2025opixsoftdecoders} in the balanced regime $\rho= 1/2$.

To interpret our bounds, we introduce algorithmic variants of the saturation and improvement bounds, denoted by $R^{\rm alg}_{1}(\rho)$ and $R^{\rm alg}_{0}(\rho)$, respectively. Here, $R^{\rm alg}_{1}(\rho)$ is the minimum rate for which a quantum polynomial-time algorithm finds an almost perfect solution, i.e., one with $s\approx 1$.\footnote{In fact, our result ensures that our algorithm finds an exact perfect solution, i.e., one with $s=1$, when $R> R_1^{\rm alg}(\rho)$, whereas DQI and the previous existence result of \cite{SunWootters26} only ensure $s=1-o(1)$.}
Similarly, $R^{\rm alg}_{0}(\rho)$ is the minimum rate for which a quantum polynomial-time algorithm improves upon DQI for a given $\rho\in(0,1)$.

In the balanced case $\rho = 1/2$, \Cref{thm:main_intro} gives
$R^{\rm alg}_{1}(1/2) \le 0.75$.
Equivalently, when $\rho = 1/2$, our algorithm finds a perfect solution for every $R > 0.75$.
This matches the condition achieved by the average-case analysis of \cite{chailloux2025opixsoftdecoders}.

More generally, for $R > 0.7149$, the satisfaction rate achieved by our worst-case algorithm matches that achieved by the average-case algorithm of \cite{chailloux2025opixsoftdecoders}.
This is because the first condition in \Cref{thm:main_intro}, namely \Cref{eq:theorem_cond_one_intro}, is identical to the corresponding condition in the average-case analysis of \cite{chailloux2025opixsoftdecoders}.
Although \Cref{thm:main_intro} also requires the additional condition \Cref{eq:theorem_cond_two_intro}, this condition is not the bottleneck when $s$ is sufficiently close to $1$; in this near-saturation regime, the bottleneck is instead \Cref{eq:theorem_cond_one_intro}.
This explains why our worst-case bound matches the previous average-case bound in this regime.  

Moreover, \Cref{thm:main_intro} gives
$R^{\rm alg}_{0}(1/2) < 0.6979$. 
Thus, when $\rho = 1/2$, our algorithm improves upon DQI for every $R > 0.6979$.

In \Cref{fig:saturation-bounds}, we compare our algorithmic saturation bound for various values of $\rho$ with those of DQI and the average-case algorithm of \cite{chailloux2025opixsoftdecoders}, based on the general bound in \Cref{thm:main}. 
Our algorithmic saturation bound improves upon that of DQI for
$0<\rho<0.6351$. In particular, our result gives the first worst-case quantum algorithm for OPI with a meaningful saturation bound for $\rho \le 1/2$.
Moreover, our algorithmic saturation bound matches the average-case bound of \cite{chailloux2025opixsoftdecoders} in the range
$0.1248<\rho<0.5200$.

\paragraph{Existential bound.}
While our main focus is on the algorithmic improvement, our analysis of the algorithm naturally extends to existential bounds, and this yields improvements over the bounds of \cite{SunWootters26} in certain parameter regimes.

Roughly speaking, our existential bound is similar to the algorithmic bound, except that it only requires \Cref{eq:theorem_cond_two_intro} and does not require \Cref{eq:theorem_cond_one_intro}. See \Cref{thm:existence} for the full statement, and see \Cref{rem:comparison_alg_existence} for a more detailed comparison with the algorithmic bound.

This gives a better bound than our algorithmic result when \Cref{eq:theorem_cond_one_intro} is the bottleneck. In particular, when $\rho= 1/2$, we obtain the improved saturation bound $R_1(1/2)<0.7158$, compared with the bound $R_1(1/2)<0.7495$ shown in \cite{SunWootters26}. More generally, our bound improves the achieved satisfaction rate $s$ for $0.7067<R<0.7494$. 
\ifnum\arxiv=1 See \Cref{fig:existential_balanced-case-comparison} for a comparison for the case $\rho= 1/2$.
\fi

For general $\rho$, our existential bound improves the saturation bounds of \cite{SunWootters26} and DQI~\cite{DQI} for $0<\rho<0.7251$. \ifnum\arxiv=1
See \Cref{fig:existential_saturation_bounds} for a comparison of the saturation bound for various values of $\rho$.
\fi 

\ifnum\arxiv=1 
\begin{figure}[!htbp]
    \centering
    \includegraphics[width=0.8\linewidth]{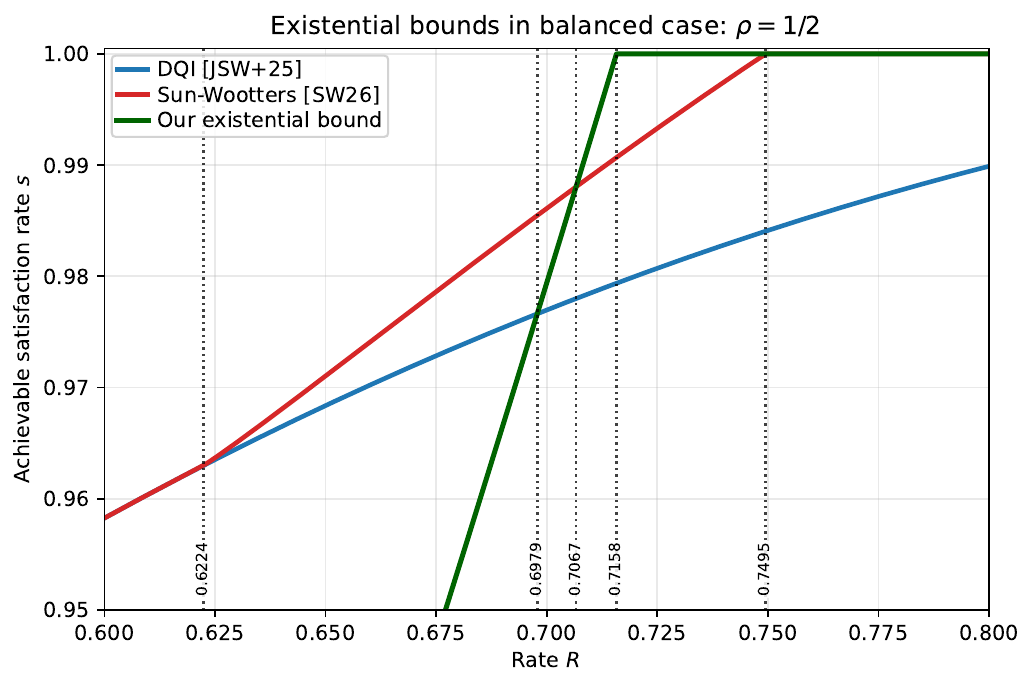}
    \caption{Comparison of existential satisfaction rates in the balanced case $\rho=1/2$.}
    \label{fig:existential_balanced-case-comparison}
\end{figure}
\begin{figure}[!htbp]
    \centering
    \includegraphics[width=0.8\linewidth]{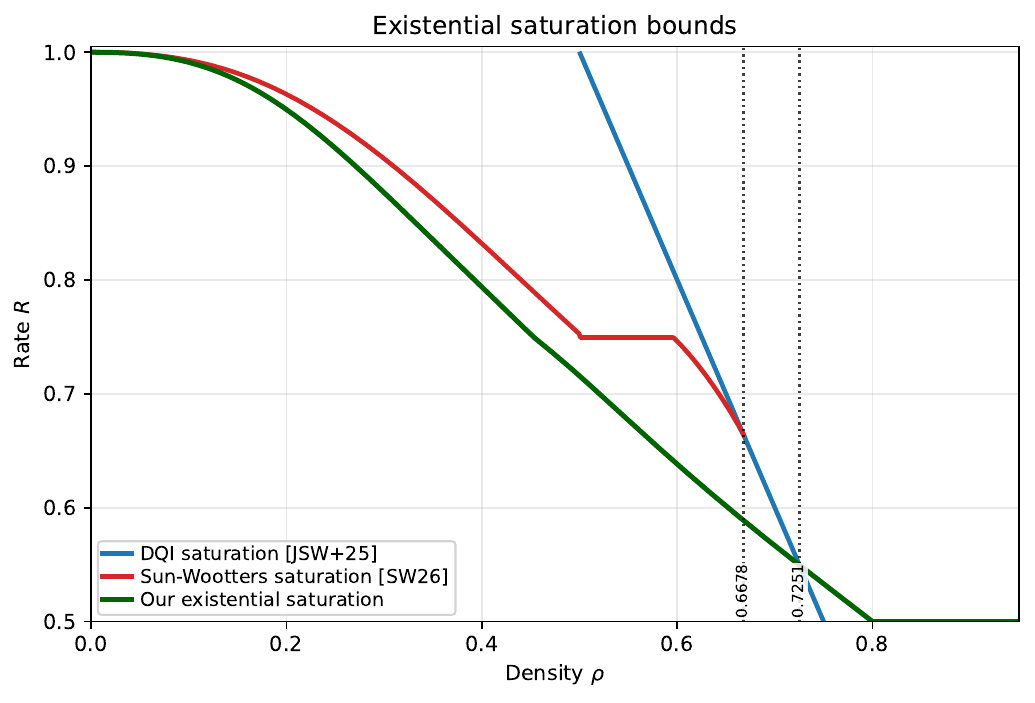}
    \caption{Comparison of existential saturation bounds as functions of $\rho$. The Sun--Wootters curve is based on \cite[Theorems 1.6 and 1.9]{SunWootters26}. While \cite[Theorem 1.9]{SunWootters26} gives a bound for general $\rho$, it has a bump for $\rho>1/2$. On the other hand, \cite[Theorem 1.6]{SunWootters26} gives a better bound in the balanced case $\rho\approx 1/2$. By a thinning argument similar to the one discussed in \Cref{sec:thinning}, one can increase the density $\rho$ without sacrificing the satisfaction rate. Thus, in the regime $\rho\ge 1/2$, whenever the bound from \cite[Theorem 1.9]{SunWootters26} is worse, we plot the thinned balanced-case bound, which gives $R\approx 0.7495$.
    }
    \label{fig:existential_saturation_bounds}
\end{figure}
\fi

\ifnum\arxiv=0
In \Cref{sec:comparison_existential}, we compare our existential bound with that of \cite{SunWootters26}. \Cref{fig:existential_balanced-case-comparison} shows the comparison for the case $\rho= 1/2$, and \Cref{fig:existential_saturation_bounds} gives a comparison of the saturation and improvement bounds for various values of $\rho$.
\fi

\paragraph{Generalization to MDS MaxLINSAT}
Our existential bound extends directly to the MaxLINSAT problem (Problem \ref{prob:MaxLINSAT}) with respect to any maximum distance separable (MDS) codes, 
which generalizes the OPI problem, similarly to the result of \cite{SunWootters26}. 

In contrast, our algorithmic bound does not extend to arbitrary MDS MaxLINSAT instances, since our algorithm relies on the list-decodability of Reed--Solomon codes~\cite{GS99,CHK2026}. 
Nevertheless, the algorithmic bound extends to MDS MaxLINSAT instances with respect to a code whose dual admits an efficient list-decoding algorithm.

\paragraph{On growing characteristic.}
Both our algorithmic and existential results assume that
the characteristic $p$ grows with $n$, i.e., $p=\omega(1)$.
We show that this assumption cannot be removed in full
generality, even for the existence of perfect solutions.

Consider characteristic $p=2$ and density $\rho=1/2$.
Let $C\subseteq\mathbb{F}_q^n$ be any linear code of
rate $R<1$, where $q$ is a power of $2$.
Choose a nonzero vector
$\mathbf{b}=(b_1,\ldots,b_n)\in C^\perp$ and let
$J:=\{i\in[n]:b_i\neq 0\}$.
For each $i\in J$, define
\[
S_i:=\{x\in\mathbb{F}_q:
\operatorname{Tr}(b_i x)=c_i\},
\]
where $c_i\in\mathbb{F}_2$ are chosen so that
$\sum_{i\in J}c_i=1$.
For $i\notin J$, let $S_i$ be any subset of
$\mathbb{F}_q$ of size $q/2$.
Then $|S_i|=q/2$ for every $i$, but
\[
C\cap(S_1\times\cdots\times S_n)=\emptyset.
\]
Indeed, if $\mathbf{x}=(x_1,\ldots,x_n)$ belonged to
this intersection, we would have
\[
0=\operatorname{Tr}
(\mathbf{b}\cdot\mathbf{x})
=\sum_{i\in J}
\operatorname{Tr}(b_i x_i)
=\sum_{i\in J}c_i=1,
\]
a contradiction.
Thus, even for rates arbitrarily close to $1$,
a perfect solution need not exist.
The same construction applies to any fixed
characteristic $p$ at density $\rho=1/p$, using
the field trace to $\mathbb{F}_p$.

\subsection{Technical Overview}
Our algorithm is very similar to those of \cite{STOC:ChaTil25,chailloux2025opixsoftdecoders}: we use Regev's reduction together with a decoder that decodes correctly with probability $1/\poly(n)$.  
The main challenge is to establish a worst-case correctness guarantee, whereas the previous works only proved an average-case guarantee.

For simplicity, we first focus on the case $s=1$. In this case, our algorithm can be viewed as a variant of the algorithm of Yamakawa and Zhandry \cite{JACM:YamZha24}. We therefore begin by reviewing their algorithm.

\paragraph{Review of Yamakawa--Zhandry algorithm.} 
Let $C$ be the Reed--Solomon code
\[
\RS_{\FF_q,\valpha,k-1}
\defeq
\{(f(\alpha_1),f(\alpha_2),\ldots,f(\alpha_n)) : f\in \FF_q[X]_{\deg\le k-1}\}.
\]
Then $\OPI(k,q,\valpha,\{S_i\}_{i\in[n]},1)$ is the problem of finding an element of $C\cap S$, where
$S=S_1\times S_2\times \cdots \times S_n.$ 
We assume that $|S_i|/q=\rho$ is a constant, with $\rho= 1/2$ being the main case of interest. 
Yamakawa and Zhandry \cite{JACM:YamZha24} gave the following quantum algorithm for this problem. 

Let $V$ and $W$ be the normalized indicator functions of $C$ and $S$, respectively. Starting from the state
\[
\sum_{\vx,\ve}V(\vx)W(\ve)\ket{\vx}\ket{\ve},
\]
the algorithm applies the quantum Fourier transform to both registers and then adds the first register to the second, obtaining
\[
\sum_{\vx,\ve}\hat V(\vx)\hat W(\ve)\ket{\vx}\ket{\vx+\ve}.
\]
Since $\hat V$ is the normalized indicator function of the dual code $C^\perp$, the first register is supported on $\vx\in C^\perp$. The algorithm then tries to recover $\vx$ from $\vx+\ve$ using a decoder for $C^\perp$ and subtracts the decoded value from the first register. If the decoder succeeds, the first register becomes $0$. If the decoder always succeeds, applying the inverse Fourier transform to the second register then yields a state proportional to
\[
\sum_{\vy\in\FF_q^n}V(\vy)W(\vy)\ket{\vy},
\]
whose support is contained in $C\cap S$.

Thus, the main issue is to bound the contribution of errors on which the decoder fails. In the formulation of \cite[Lemma 5.1]{JACM:YamZha24}, it suffices to bound the following two quantities:
\begin{align}
     &\sum_{\ve\in \bad}|\hat{W}(\ve)|^2, \label{eq:overview_cond_one}\\
     &q^{-(n-k)}\sum_{\vy \in \FF_q^n}\left|\sum_{\vx\in C^\perp}\one_{\bad}(\vy-\vx)\hat{W}(\vy-\vx)\right|^2. \label{eq:overview_cond_two}
\end{align}
Here, $\bad\subseteq \FF_q^n$ denotes a set of errors on which the decoder may fail:\footnote{In \cite[Lemma 5.1]{JACM:YamZha24}, $\bad$ is defined as a subset of pairs $(\vx,\ve)$. Here, we consider the special case where membership in $\bad$ depends only on $\ve$, which holds for most natural decoders, including the one considered in \cite{JACM:YamZha24}.
} for every $\ve\notin\bad$, we have $\mathsf{Decode}(\vx+\ve)=\vx$ for all $\vx\in C^\perp$.\footnote{Thus, $\bad$ may also contain errors that are in fact decoded correctly.} We denote by $\one_{\bad}$ the indicator function of $\bad$.

They bound these quantities with high probability over the random choices of $S_1,\ldots,S_n$ as follows.

Since $W$ is the product of the normalized indicator functions $W_i$ of $S_i$, we have
$\hat W(\ve)=\prod_{i=1}^n \hat W_i(e_i).$
Moreover,
$|\hat W_i(0)|^2=\frac{|S_i|}{q}=\rho,$ 
and over the random choice of $S_i$, the nonzero Fourier coefficients $\hat W_i(a)$ behave symmetrically over all $a\ne 0$.

Thus, on average, \Cref{eq:overview_cond_one} can be viewed as the probability that $\ve$ lies in $\bad$, where each coordinate is independently $0$ with probability $\rho$ and is otherwise uniformly distributed over $\FF_q\setminus\{0\}$. This probability is exponentially small when the decoder succeeds with high probability against such random errors.

The second quantity, \Cref{eq:overview_cond_two}, is controlled using the random phases of $\hat W$. These phases make the cross terms cancel, so that, with high probability,
\[
\sum_{\vy\in\FF_q^n}
\left|
\sum_{\vx\in C^\perp}
\one_{\bad}(\vy-\vx)\hat W(\vy-\vx)
\right|^2
\approx
\sum_{\vy\in\FF_q^n}
\sum_{\vx\in C^\perp}
\one_{\bad}(\vy-\vx)|\hat W(\vy-\vx)|^2.
\]
The RHS is
\[
|C^\perp|\sum_{\ve\in\bad}|\hat W(\ve)|^2.
\]
Since $|C^\perp|=q^{n-k}$, the prefactor $q^{-(n-k)}$ in \Cref{eq:overview_cond_two} cancels this factor. Hence \Cref{eq:overview_cond_two} reduces to \Cref{eq:overview_cond_one} in the average-case analysis.

\paragraph{Replacing $\bad$ with a Hamming tail via list-decoding.}
While the above analysis is relatively easy when $S_1,\ldots,S_n$ are chosen at random, the worst-case analysis seems much more difficult. Indeed, the work of \cite{JACM:YamZha24}, as well as follow-up works \cite{STOC:ChaTil25,chailloux2025opixsoftdecoders,SODA:BDGJLS26}, relies on the randomness in the choice of the subsets to analyze similar algorithms.

To analyze the quantities in \Cref{eq:overview_cond_one,eq:overview_cond_two} in the worst-case, we need to understand the structure of the set $\bad$. One possible approach is to let $\bad$ consist of all errors outside the unique decoding radius. However, this is too weak, as it only yields the same guarantee as DQI. We would therefore like to use a stronger decoder. For example, \cite{JACM:YamZha24} considered a decoder that first runs a list decoder and then outputs the decoded word if the resulting list contains a unique element. The difficulty with such a decoder, however, is that the corresponding set $\bad$ appears hard to analyze.

Thus, our idea, similar to \cite{STOC:ChaTil25,chailloux2025opixsoftdecoders}, is to 
use a decoder that runs list-decoder and then outputs a random element of the list. 
If we define $\bad$ to be the set of errors outside the list-decoding radius, then 
$
\bad=T_{rn}=\{\vx\in \FF_q^n: \wt(\vx)>rn\}$ 
where $r=1-\sqrt{1-R}-\Omega(1).$ 
Indeed, Reed--Solomon codes of rate $R'$ are known to be list-decodable from errors of Hamming weight at most
$(1-\sqrt{R'}-\Omega(1))n$~\cite{GS99,CHK2026}.\footnote{
Here we apply this fact to $C^\perp$, whose rate is $1-R$. Although $C^\perp$ is not necessarily a Reed--Solomon code, it is a generalized Reed--Solomon code, which has the same list-decoding radius as Reed--Solomon codes.}

The issue with setting $\bad=T_{rn}$ is that the decoder may no longer always correct errors outside $\bad$. Thus, we need to carefully handle the decoding step of the algorithm. 
Recall that we consider a decoder that runs a list-decoder and then outputs a uniformly random element from the list. This may not always result in correct decoding, but since the list size is $\poly(n)$ within the list-decoding radius, the decoder succeeds with probability at least $1/\poly(n)$ whenever $\ve\notin\bad$.

This is not enough to ensure that the algorithm generates a state that is negligibly close to the target state. However, it is sufficient to ensure that the generated state has inverse-polynomial overlap with the desired state. Since measuring the desired state yields a correct solution with probability $1$, measuring such a state yields a correct solution with probability $1/\poly(n)$, which is sufficient for our purpose.

Although this argument may sound simple, the actual proof requires some care. For example, our analysis crucially relies on the specific choice of decoder: it selects a uniformly random element from the decoded list. This is significantly different from the average-case analyses in \cite{STOC:ChaTil25,chailloux2025opixsoftdecoders}, where any decoder that correctly decodes with probability $1/\poly(n)$ suffices.

In addition, to make the argument go through, we need an additional condition that controls the norm of the target state. This condition, which corresponds to \Cref{eq:cond_three} in \Cref{thm:Regev_like_new}, does not appear in the average-case analysis. Fortunately, this additional condition is easier to satisfy than the other two conditions, so we ignore it in the rest of this overview. 
See \Cref{thm:Regev_like_new} and its proof for details. 

Thus, the problem is now to bound \Cref{eq:overview_cond_one,eq:overview_cond_two} for $\bad=T_{rn}$ in the worst case.

When $\bad=T_{rn}$, bounding \Cref{eq:overview_cond_one} in the worst case is straightforward, assuming only the sizes of the sets $S_i$. As discussed above, if $W_i$ is the normalized indicator function of $S_i$, then $|\hat{W}_i(0)|^2=\frac{|S_i|}{q}=\rho.$  
Thus, the total squared amplitude of $\hat{W}_i$ on nonzero frequencies is $1-\rho$. Therefore, \Cref{eq:overview_cond_one} can be interpreted as the upper-tail probability that a sum of independent Bernoulli random variables with mean $(1-\rho)n$ exceeds $rn$. In particular, if
$r>1-\rho$ by a constant margin, then the Chernoff bound implies that \Cref{eq:overview_cond_one} is exponentially small for $\bad=T_{rn}$. 

On the other hand, bounding \Cref{eq:overview_cond_two} remains highly nontrivial even for $\bad=T_{rn}$. This is the main technical challenge that we solve in this work. We explain how to do this below.

\paragraph{Decomposition into shift correlations.}
The first step is to express the value of \Cref{eq:overview_cond_two} in the following form:
\[
    \sum_{\vu\in C^\perp}\Gamma(\vu),
\]
where
\[
\Gamma(\vu)=\sum_{\ve\in \FF_q^n}
\one_{T_{rn}}(\ve)\one_{T_{rn}}(\ve+\vu)\hat W(\ve)\overline{\hat W(\ve+\vu)}.
\]
This follows from a straightforward calculation. See \Cref{lem:split_to_sum_of_gamma}.

Moreover, the value of \Cref{eq:overview_cond_one} is equal to $\Gamma(0)$. Since we already know that this is negligible, it suffices to show that
\begin{align*} \sum_{\vu\in C^\perp\setminus\{0\}}\Gamma(\vu) \end{align*}
is negligible.

\paragraph{Bounding $\Gamma$ via generating polynomials.}
Next, we bound $\Gamma(\vu)$ by introducing a bivariate generating polynomial whose coefficients keep track of the two weights $\wt(\ve)$ and $\wt(\ve+\vu)$.

For $i\in [n]$ and $b\in \FF_q$, define the polynomial
\[
  K_i^{(b)}(X,Y)=\sum_{a\in \FF_q}
  \hat{W}_i(a)\overline{\hat{W}_i(a+b)}
  X^{\one_{a\ne 0}}Y^{\one_{a+b\ne 0}}.
\]
Here, $\one_{a\ne 0}$ is $1$ if $a\ne 0$ and $0$ otherwise, and $\one_{a+b\ne 0}$ is defined similarly.

Then it is easy to see that
\[
   \Gamma(\vu)=\sum_{a,b> rn}[X^aY^b]\prod_{i=1}^{n}K_i^{(u_i)}(X,Y),
\]
where $[X^aY^b]F(X,Y)$ denotes the coefficient of $X^aY^b$ in $F(X,Y)$. See the proof of \Cref{lem:upper_bound_gamma_by_K}.

By straightforward calculations, we have
\[
K_i^{(0)}(X,Y)=\rho+(1-\rho)XY,
\]
and, for $b\ne 0$,
\[
K_i^{(b)}(X,Y)=\mu_i(b)(\rho X+\rho Y+(1-2\rho)XY),
\]
where
\[
\mu_i(b)=\frac{1}{|S_i|}\sum_{x\in S_i}\omega_p^{-\Tr(bx)}.
\]
See the proofs of \Cref{lem:K_zero,lem:K_non_zero} with $\tau=1$.

For a polynomial $F(X,Y)\in \mathbb{C}[X,Y]$, let $\|F(X,Y)\|_1$ denote its $\ell_1$ norm, namely the sum of the absolute values of all coefficients. Then
\[
\|K_i^{(0)}(X,Y)\|_1=1
\]
and, for all $b\in \FF_q\setminus\{0\}$,
\[
\|K_i^{(b)}(X,Y)\|_1=
\begin{cases}
|\mu_i(b)| & \rho\le 1/2,\\
(4\rho-1)|\mu_i(b)| & \rho>1/2.
\end{cases}
\]
Below, we assume $\rho\le 1/2$ for simplicity, in which case $\|K_i^{(b)}(X,Y)\|_1=|\mu_i(b)|$ for every $b\ne 0$.\footnote{The case $\rho>1/2$ can be analyzed similarly, with an additional multiplicative loss of $(4\rho-1)^{\wt(\vu)}$.} 

Combining this with the expression above for $\Gamma(\vu)$, we obtain
\[
|\Gamma(\vu)|
\le
\prod_{i=1}^{n}\|K_i^{(u_i)}(X,Y)\|_1
=
\prod_{i:u_i\ne 0}|\mu_i(u_i)|,
\]
where $u_i$ denotes the $i$th coordinate of $\vu$. Therefore, our task reduces to showing that
\begin{align}
\sum_{\vu\in C^\perp\setminus \{0\}}\prod_{i:u_i\ne 0}|\mu_i(u_i)| \label{eq:sum_prod_mu_overview}
\end{align}
is negligible.

\paragraph{Applying the Brascamp--Lieb inequality.}
The crux of our proof is to upper-bound \Cref{eq:sum_prod_mu_overview} using a variant of the Brascamp--Lieb inequality~\cite{BL76} adapted to MDS codes. 

Following the entropy-subadditivity approach to Brascamp--Lieb-type inequalities developed by  \cite{CarlenCordero2009}, we prove the following inequality using the fact that $C^\perp$ is an MDS code of rate $1-R$; see \Cref{thm:MDS_BL}. For any nonnegative functions $h_i:\FF_q\rightarrow \mathbb{R}_{\ge 0}$ for $i\in[n]$, it holds that
\[
\sum_{\vu\in C^\perp}\prod_{i\in[n]} h_i(u_i)
\le
\prod_{i\in[n]}
\left(
\sum_{a\in\mathbb F_q}h_i(a)^Q
\right)^{1/Q},
\]
where
$Q=\frac{1}{1-R}.$

The first naive attempt to use this inequality for bounding \Cref{eq:sum_prod_mu_overview} is to set $h_i(0)=1$ and $h_i(a)=|\mu_i(a)|$ for $a\ne 0$. Unfortunately, this is not strong enough. The reason is that this direct application does not sufficiently exploit the fact that every nonzero element of $C^\perp$ has Hamming weight at least $k+1$. To make use of this information, we introduce a tilting parameter $t\ge 1$ and define $h_i^{(t)}$ as follows:
\[
h_i^{(t)}(a)=
\begin{cases}
1 & \text{if } a=0,\\[1.2ex]
t|\mu_i(a)| & \text{if } a\ne 0.
\end{cases}
\]
Since $\wt(\vu)\ge k+1>k$ for all $\vu\in C^\perp\setminus\{0\}$, we have
\[
\prod_{i:u_i\ne 0}|\mu_i(u_i)|
\le
t^{-k}\prod_{i\in[n]}h_i^{(t)}(u_i).
\]
Therefore,
\[
\sum_{\vu\in C^\perp\setminus \{0\}}\prod_{i:u_i\ne 0}|\mu_i(u_i)|
\le
t^{-k}\sum_{\vu\in C^\perp}\prod_{i\in[n]}h_i^{(t)}(u_i),
\]
where we included the term $\vu=0$ in the sum, which only makes the quantity larger.

Applying the above Brascamp--Lieb-type inequality to $h_i^{(t)}$, we obtain
\[
t^{-k}\sum_{\vu\in C^\perp}\prod_{i\in[n]} h_i^{(t)}(u_i)
\le
t^{-k}
\prod_{i\in[n]}
\left(
\sum_{a\in\mathbb F_q}h_i^{(t)}(a)^Q
\right)^{1/Q}.
\]
By the definition of $h_i^{(t)}$, the RHS is equal to
\[
t^{-k}
\prod_{i\in[n]}
\left(
1+t^Q M_i
\right)^{1/Q},
\]
where
\[
M_i=\sum_{a\in\mathbb F_q\setminus\{0\}}|\mu_i(a)|^Q.
\]
Taking 
$M\ge \max_{i\in[n]} M_i,$ 
we get
\[
\sum_{\vu\in C^\perp\setminus \{0\}}\prod_{i:u_i\ne 0}|\mu_i(u_i)|
\le
t^{-k}(1+t^Q M)^{n/Q}.
\]

We now optimize over $t$. Setting 
$t^Q=\frac{R}{M(1-R)}$ 
gives
\[
\frac{1}{n}
\log
\left(
t^{-k}(1+t^Q M)^{n/Q}
\right)
=
\frac{1}{Q}\left(H(R)+R\log M\right),
\]
where $H(\cdot)$ is the binary entropy function. Thus, it suffices to show that $M$ is small enough so that
\[
H(R)+R\log M<0
\]
with a constant margin.

\paragraph{Bounding $M$ via Fourier bias bounds.}
It remains to give a good upper bound $M$ of $M_i$. Recall that $M_i$ is the sum of the $Q$th powers of the absolute values of the Fourier biases $\mu_i$. Note that $Q\ge 2$ when $R\ge 1/2$, which is the regime considered in this work. 

First, bounding $M_i$ is easy when $Q=2$. By a straightforward calculation, we have
\[
\sum_{a\in\mathbb F_q\setminus\{0\}}|\mu_i(a)|^2=\frac{1-\rho}{\rho}.
\]
See \Cref{lem:Fourier_quadratic_sum} where we note 
$|\mu_i(a)|=\frac{1}{\rho\sqrt{q}}|\hat{\one}_{S_i}(a)|$. 

Moreover, as shown in \cite{SunWootters26},\footnote{Strictly speaking, they prove this for prime fields, but the argument can be generalized to extension fields.} for any $a\in \FF_q\setminus\{0\}$, we have
\[
|\mu_i(a)|\le \frac{\sin(\rho \pi)}{\rho \pi}+O(p^{-2}),
\]
where $p$ is the characteristic of $\FF_q$. See \Cref{lem:Fourier_max}.

Combining the above, we obtain
\[
\begin{aligned}
M_i
&=\sum_{a\in\mathbb F_q\setminus\{0\}}|\mu_i(a)|^Q\\
&\le
\sum_{a\in\mathbb F_q\setminus\{0\}}|\mu_i(a)|^2\left(\max_{a\in\mathbb F_q\setminus\{0\}}|\mu_i(a)|\right)^{Q-2}
\\
&\le
\frac{1-\rho}{\rho}
\left(\frac{\sin(\rho \pi)}{\rho \pi}+O(p^{-2})\right)^{Q-2}.
\end{aligned}
\]
This gives a worst-case upper bound on $M_i$ depending only on $\rho$ and $R$, ignoring the subconstant $O(p^{-2})$ term. 

In the actual proof, we use a slightly sharper bound on $M_i$ by also controlling the fourth moment \[ \sum_{a\in\FF_q\setminus\{0\}}|\mu_i(a)|^4. \] This fourth-moment estimate is obtained through an additive-combinatorial bound on the additive energy of $S_i$. Interpolating this fourth-moment bound with the Fourier bias bound above yields the final expression for $M$ used in the theorem. We omit these details in this overview and refer the reader to \Cref{sec:Fourier}.

Substituting the bound on $M$ into the condition
\[
H(R)+R\log M<0
\]
yields the parameter regime in which the algorithm is correct in the worst case.

This completes the overview of the proof of worst-case correctness of the algorithm for the case $s=1$.

\paragraph{Generalizing to the case of $s<1$.}
Here, we briefly discuss how to generalize the above algorithm and analysis to the case of $s<1$. For this, we rely on the idea of \cite{chailloux2025opixsoftdecoders}. The idea is simple. Let $\tau<1$ be a constant slightly larger than $s$. We modify the function $W$ as $W(\ve)=\prod_{i\in[n]}W_i(e_i)$, where
\[
W_i(e_i)=
\begin{cases}
\sqrt{\dfrac{\tau}{|S_i|}} & \text{if } e_i\in S_i,\\[1.2ex]
\sqrt{\dfrac{1-\tau}{q-|S_i|}} & \text{if } e_i\notin S_i.
\end{cases}
\]

By repeating a similar analysis to the case of $s=1$, we can show that a quantum polynomial-time algorithm prepares a state that has inverse-polynomial overlap with the normalized variant of
\[
\sum_{\vy\in\FF_q^n}V(\vy)W(\vy)\ket{\vy}.
\]

Intuitively, measuring this state in the standard basis should yield a solution to the OPI problem with satisfaction rate at least $s$. Indeed, $V$ restricts the state to codewords in $C$, while each $|W_i|^2$ places total mass $\tau>s$ on $S_i$. Thus, the resulting distribution biases codewords toward those that satisfy many of the constraints.

Making this intuition rigorous is somewhat delicate: one has to show that, after restricting to $C$, the mass of codewords with satisfaction rate below $s$ is negligible. Roughly, we prove this by comparing the above weighted state with the analogous state defined using the parameter $s$. This completes the extension from the case $s=1$ to the case of general $s<1$.

\subsection{Related Work}
\paragraph{Regev's reduction and its variants.}
The general idea of performing decoding in the Fourier domain first appeared in the work of Regev~\cite{JACM:Regev09}, who showed a quantum reduction from the shortest integer solution (SIS) problem to the learning with errors (LWE) problem.
Stehl\'{e}, Steinfeld, Tanaka, and Xagawa~\cite{AC:SSTX09}, and Lyubashevsky, Peikert, and Regev~\cite{JACM:LyuPeiReg13}, extended Regev's reduction to the setting of ideal lattices.
These works form part of the foundation of lattice-based cryptography; see, e.g., \cite{EPRINT:Peikert15} for a survey on the subject.

Brakerski, Kirshanova, Stehl\'{e}, and Wen~\cite{PKC:BKSW18} used Regev's reduction to show an equivalence between the LWE problem and a problem called the extrapolated dihedral coset problem.

Chen, Liu, and Zhandry~\cite{EC:CheLiuZha22} applied Regev's idea to design a quantum algorithm for solving the SIS problem under the $\ell_\infty$ norm, in a parameter regime where no efficient classical algorithm was known.
However, Kothari, O'Donnell, and Wu~\cite{kothari2026exponentialquantumspeedupmathrmsisinfty}, building on earlier works~\cite{ISS12,II24}, showed that there is a classical algorithm for the SIS problem under the $\ell_\infty$ norm with better parameters than the quantum algorithm of~\cite{EC:CheLiuZha22}.
On the other hand, the authors of~\cite{kothari2026exponentialquantumspeedupmathrmsisinfty} argue that their dequantization result does not seem to extend to the OPI problem.

Yamakawa and Zhandry~\cite{JACM:YamZha24} used Regev's reduction in the context of query complexity, giving the first example of an efficiently verifiable problem that admits a super-polynomial quantum speedup relative to a random oracle.

Briaud, Dinur, Ghosal, Jain, Lou, and Sahai~\cite{SODA:BDGJLS26} adapted the result of~\cite{JACM:YamZha24} to the problem of finding solutions to systems of multivariate equations, giving a quantum algorithm for such problems that outperforms the best known classical algorithms.

Jordan, Shutty, Wootters, Zalcman, Schmidhuber, King,
Isakov, Khattar, and Babbush~\cite{DQI}, building on
Regev's idea, introduced a quantum algorithm called
decoded quantum interferometry (DQI), which provides
approximation guarantees for MaxLINSAT when suitable
decoding algorithms are available.
Their algorithm also applies to MaxXORSAT, which is
the special case of MaxLINSAT over $\mathbb{F}_2$
with singleton constraint sets.
In contrast, our results, similarly to
\cite{SunWootters26}, require the characteristic
of the underlying field to grow with the input size.
Whether similar improvements over DQI can be obtained
for more general MaxLINSAT instances, particularly over
fields of small characteristic, remains open even at
the existential level.

On the hardness side, Kramer, Schubert, and
Eisert~\cite{kramer2026tight} showed that, for any
fixed finite field $\mathbb{F}_q$ and any fixed
constraint-set size $r$, it is NP-hard to achieve
an approximation ratio exceeding the random-assignment
baseline $r/q$ by any positive constant on general
MaxLINSAT instances, even when the instances are
nearly satisfiable.
This highlights the importance of exploiting additional
structure, such as the MDS property considered in
our work.

Rosmanis~\cite{rosmanis2026nearly} improved the
efficiency of DQI for OPI, obtaining a nearly
linear-time algorithm in the random-access input model.

Chailloux and Tillich~\cite{TQC:ChaTil24} applied Regev's reduction in the context of linear codes, giving a quantum algorithm for finding a short codeword.

Chailloux and Tillich~\cite{STOC:ChaTil25} extended Regev's reduction using decoders that succeed in decoding only with probability $1/\poly(n)$.
They used this idea to improve the parameter regime for solving the OPI problem over DQI, but their analysis works only in the average-case setting over the random choice of the constraint subsets.
Chailloux and Hermouet~\cite{C:ChaHer26} further improved the result of~\cite{STOC:ChaTil25} to show an equivalence between certain variants of the SIS and LWE problems.
Chailloux~\cite{chailloux2025opixsoftdecoders} adapted the result of~\cite{C:ChaHer26} to the OPI setting and further improved the parameter regimes in which a quantum algorithm can solve OPI.
However, this result also provides only an average-case analysis.

\paragraph{Comparison with Sun--Wootters.}
The work of Sun and Wootters~\cite{SunWootters26} is closely related to ours.
Indeed, they were the first to observe that no known quantum algorithm for OPI improves upon DQI in the worst case. Their results are purely existential, whereas our main contribution is algorithmic. Thus, even in parameter regimes where their existential bound is stronger, it does not directly yield a quantum algorithm.

At a high level, the techniques are different: their approach is closer to the analysis of DQI~\cite{DQI}, whereas ours is closer to the algorithms of~\cite{JACM:YamZha24,chailloux2025opixsoftdecoders}. 
A common feature is that both works need to control Fourier sums over the dual code, although the exact quantities are different. The quantity considered by Sun and Wootters appears in \cite[Equation (1.19)]{SunWootters26}, whereas our corresponding quantities appear in \Cref{eq:sum_prod_mu_overview} in the overview and in \Cref{lem:Z_is_negl} in the general case. 

Sun and Wootters also observe that a similar quantity appears in the seemingly unrelated context of leakage-resilient secret sharing~\cite{JC:BDIR21,C:MPSW21,ISIT:MNPW22}.
However, later works~\cite{C:KleKom23,C:Nguyen24,TCC:Kasser24} improved the results on leakage-resilient secret sharing without bounding the above quantity.
Thus, improving this bound does not necessarily lead to improvements in the context of leakage-resilient secret sharing. It remains an open question whether our techniques can improve the state of the art in leakage-resilient secret sharing.\footnote{
In subsequent work, Batra, Kitagawa, Nishimaki, and Yamakawa~\cite{batra2026quantumleakageresilienceshamir}
built on our techniques to establish quantum leakage
resilience for Shamir secret sharing.
} 

\paragraph{Brascamp--Lieb inequality in coding theory.}
A recent work of Brakensiek, Chen, Dhar, and Zhang~\cite{brakensiek2025combinatorialboundslistrecovery} relies on a variant of the Brascamp--Lieb inequality to prove list-recoverability bounds for various families of (folded) linear codes. While their inequality is related to the entropic subadditivity underlying our MDS Brascamp--Lieb inequality, our application is quite different: they use the entropic bound directly to control the size of recovered lists, whereas we use the Gibbs variational principle to derive an inequality for sums of products over codewords, and then apply it to bound Fourier-bias sums over the dual code.

\subsection{Concurrent and Subsequent Work}

Concurrent and independent work by Jo
\cite{jo2026efficientexactquantumsampling} gives a different
quantum algorithm for OPI that achieves worst-case improvements
over DQI. Jo's approach efficiently samples the Fourier-defined
distribution underlying the existential bounds of Sun and Wootters
\cite{SunWootters26}, and both the algorithm and its analysis
differ substantially from ours.
Jo's results focus on the balanced case over prime fields,
whereas our results apply to arbitrary constant densities
over prime-power fields of growing characteristic.
Even in the balanced case, neither algorithmic bound dominates
the other: Jo's bound is stronger at lower rates, whereas ours
is stronger when $R$ approaches $3/4$ from below.
Moreover, our algorithm achieves exact satisfaction $s=1$
for every fixed $R>3/4$, whereas Jo's algorithm guarantees
$s=1-o(1)$.

Subsequent to our work and that of Jo,
Sun and Wootters
\cite{sun2026improvedlocalleakageresilience}
obtained substantially stronger worst-case quantum
algorithms for OPI, improving upon the algorithmic
bounds of both works.
In particular, their algorithm finds a perfect solution
whenever $R+\rho>1$. In the balanced case, this
improves the rate threshold for finding a perfect
solution from $R>3/4$ to $R>1/2$.
They also establish stronger existential bounds.
Their results leverage recent breakthroughs in
list decoding and list recovery of Reed--Solomon codes
\cite{brakensiek2026algorithmiclistdecodingreedsolomon,
jeronimo2026algorithmiclistdecodingcapacity,DKT26},
all of which appeared subsequent to our work.

They also establish an unconditional quantum--classical
query-complexity separation for the original (unfolded)
OPI problem in the membership-oracle model. 
This extends the scope of earlier oracle separations
for folded OPI by Yamakawa and Zhandry~\cite{JACM:YamZha24}
and Kramer, Gil-Fuster, Jones, Eisert, and
Schreiber~\cite{kramer2026provable}.

\subsection{Paper Organization} 
In \Cref{sec:preliminaries}, we introduce basic definitions and review relevant known results.
In \Cref{sec:BL}, we prove the MDS Brascamp--Lieb inequality, and in \Cref{sec:Fourier}, we establish several Fourier sum bounds. These results serve as key technical ingredients in our analysis.
In \Cref{sec:Regev}, we present and analyze a Regev-like quantum algorithm.
In \Cref{sec:main}, we apply this algorithm to the OPI problem and establish its performance guarantees using the results from the preceding sections.
Finally, in \Cref{sec:existential}, we derive existential bounds based on the analysis in \Cref{sec:main}.

\section{Preliminaries}\label{sec:preliminaries}

\paragraph{Basic notations.}
For a set $X$, $|X|$ is the cardinality of $X$.
For a non-empty finite set $X$, 
we denote by $x\gets X$ to mean that $x$ is uniformly taken from $X$.
For a distribution $D$ over a set $X$, we denote by $x\gets D$ to mean that $x\in X$ is taken according to the distribution $D$. 

For $x\in \mathbb{C}$, $\overline{x}$ means its complex conjugate. 


A function $p:\mathbb{N}\rightarrow \mathbb{R}_{\ge 0}$ is said to be polynomially bounded if there exists a constant $c>0$ such that $p(n)\le n^c$ for all sufficiently large $n$. 
We often write $\poly$ to denote an unspecified function that is polynomially bounded.  
A function $\mu:\mathbb{N}\rightarrow \mathbb{R}_{\ge 0}$ is said to be negligible if for every constant $c>0$, $\mu(n)<n^{-c}$ for all sufficiently large $n$.
We often write $\negl$ to denote an unspecified negligible function. 

For a positive integer $n$, $[n]$ means the set $\{1,...,n\}$.
For a random variable $X$, $\Ex[X]$ denotes its expected value.
For a quantum or classical algorithm $\A$, we denote $y\gets \A(x)$ to mean that $\A$ outputs $y$ on input $x$. 
For a finite set $S\subseteq X$, $\one_{S}:X\rightarrow \{0,1\}$ denotes the indicator function of $S$, i.e., 
\[
\one_S(x)= 
\begin{cases}
1& x\in S\\
0& \text{otherwise}
\end{cases}.
\]

Throughout, $\log$ denotes the logarithm to base $2$. 

A probability distribution on a finite set $X$ is a function $P:X\rightarrow \mathbb{R}_{\ge 0}$ such that $\sum_{x\in X}P(x)=1$. 
For a probability distribution $P$ on a finite set $X$, let
$H(P)=-\sum_{x\in X}P(x)\log P(x)$ denote its Shannon entropy, where
$0\log 0$ is interpreted as $0$. 
For $p\in[0,1]$, we also write
$H(p)=-p\log p-(1-p)\log(1-p)$
for the binary entropy function, i.e., the entropy of the Bernoulli
distribution with parameter $p$. The meaning of $H$ will be clear from the context.


\paragraph{Notations for quantum states.}
For a not necessarily normalized state $\ket{\psi}$, we denote by $\|\ket{\psi}\|$ to mean its Euclidean norm. 
For quantum states $\rho$ and $\sigma$, we denote by $\TD(\rho,\sigma)$ to mean the trace distance between them, and
$F(\rho,\sigma)$ to mean the (root) fidelity between them, i.e., $F(\rho,\sigma)=\Tr\sqrt{\sqrt{\rho}\sigma\sqrt{\rho}}$.

The following inequality is known as the Fuchs–van de Graaf inequality.
\begin{align}
\TD(\rho,\sigma)\le \sqrt{1-F(\rho,\sigma)^2}. \label{eq:Fuchs–van_de_Graaf}
\end{align}

\subsection{Finite Fields}\label{sec:finite_field}
We review basic notations and facts on finite fields. The following exposition is taken from  \cite{JACM:YamZha24}. 

For a prime power $q=p^e$, $\FF_q$ denotes a field of order $q$. 
We denote by $0$ to mean $(0,...,0)\in \FF_q^n$ where $n$ will be clear from the context. 

For $\vx\in \FF_q^n$ and $i\in[n]$, we use $x_i$ to mean the $i$th coordinate of $x$.  
For $\vx\in \FF_q^n$ and $\vy\in \FF_q^n$, we define $\vx\cdot \vy\defeq \sum_{i=1}^{n}x_iy_i$.


The trace function $\Tr:\FF_q\rightarrow \FF_p$ is defined by
\begin{align*}
    \Tr(x)\defeq \sum_{i=0}^{e-1}x^{p^{i}}. 
\end{align*}
The trace function is $\FF_p$-linear, i.e., for any $a,b\in \FF_p$ and $x,y\in \FF_q$, we have 
\begin{align*}
    \Tr(ax+by)=a\Tr(x)+b\Tr(y).
\end{align*}
We let $\omega_p\defeq e^{2\pi i/p}$. 
For any $\vx\in \FF_q^n \setminus \{\vzero\}$, we have 
\begin{align}\label{eq:sum_is_zero}
\sum_{\vy\in \FF_q^n}\omega_p^{\Tr(\vx\cdot \vy)}=0.     
\end{align}

For $\vx\in \FF_q^n$, we denote by $\wt(\vx)$ to mean the Hamming weight of $\vx$, i.e., 
$\wt(\vx)\defeq |\{i\in [n]: x_i \neq 0\}|$.
For $\vx\in \FF_q^n$ and a subset $S\subseteq [n]$, we denote by $\vx_S$ to mean $(x_i)_{i\in S}$. 
For an integer $0\le w \le n$, we define subsets $B_{w}\subseteq \FF_q^n$ and $T_{w}\subseteq \FF_q^n$ as follows:\footnote{While $B_w$ and $T_w$ also depend on $q$ and $n$, we omit this dependence from the notation as this will be clear from the context.}
\begin{align*}
    &B_w=\{\vx\in \FF_q^n: \wt(\vx)\le w\},
    &T_w=\{\vx\in \FF_q^n: \wt(\vx)> w\}
\end{align*}
Note that $B_w$ and $T_w$ are disjoint and that $B_w\cup T_w=\FF_q^n$. 
\subsection{Quantum Fourier Transform over Finite Fields}
We review known facts on quantum Fourier transform over finite fields. The following exposition is taken from  \cite{JACM:YamZha24}. 

On a quantum system over a finite field $\FF_q$ for a prime power $q=p^e$, a quantum Fourier transform is a unitary denoted by $\QFT_{\FF_q}$ such that for any $x\in \FF_q$, 
\begin{align*}
    \QFT_{\FF_q}\ket{x} = \frac{1}{\sqrt{q}}\sum_{z\in \FF_q}\omega_p^{\Tr(x \cdot z)}\ket{z}. 
\end{align*} 
A quantum Fourier transform over $\FF_q$ can be approximated to within error $\epsilon$ in time polynomial in $\log q$ and $\log 1/\epsilon$~\cite{BCW02,DHI06}. For ease of exposition, we ignore the approximation error in the rest of the paper since it can be made exponentially small by a polynomial-size quantum circuit.  

For any positive integer $n$ and $\vx\in \FF_q^n$, 
by the linearity of the QFT, we have 
\begin{align*}
    \QFT_{\FF_q}^{\otimes n}\ket{\vx} = \frac{1}{q^{n/2}}\sum_{\vz\in \FF_q^n}\omega_p^{\Tr(\vx \cdot \vz)}\ket{\vz}
\end{align*} 
by the definition of $\QFT_{\FF_q}$ and linearity of $\Tr$. 

For a function $f:\FF_q^n\ra \mathbb{C}$, we define 
\begin{align*}
    \hat{f}(\vz)\defeq \frac{1}{q^{n/2}}\sum_{\vx\in \FF_q^n}f(\vx)\omega_p^{\Tr(\vx \cdot \vz)}. 
\end{align*}
Then it is easy to see that we have 
\begin{align*}
    \QFT_{\FF_q}^{\otimes n}\sum_{\vx\in \FF_q^n} f(\vx)\ket{\vx} = \sum_{\vz\in \FF_q^n} \hat{f}(\vz)\ket{\vz}.
\end{align*}
For functions $f:\FF_q^n\ra \mathbb{C}$ and $g:\FF_q^n\ra \mathbb{C}$, $f\cdot g$ and $f\ast g$ denote the point-wise product and convolution of $f$ and $g$, respectively, i.e.,
\begin{align*}
&(f\cdot g)(\vx)\defeq f(\vx)\cdot g(\vx)\\
    &(f\ast g)(\vx)\defeq \sum_{\vy\in \FF_q^n}f(\vy)\cdot g(\vx-\vy).
\end{align*}

The following lemmas are standard. See e.g., \cite{JACM:YamZha24} for the proofs.

\begin{lemma}[Parseval's identity]\label{lem:Parseval}
For any $f:\FF_q^n\ra \mathbb{C}$
and $g:\FF_q^n\ra \mathbb{C}$, we have 
\begin{align*} 
    \sum_{\vx\in \FF_q^n}f(\vx)\overline{g(\vx)}=\sum_{\vz\in \FF_q^n}\hat{f}(\vz)\overline{\hat{g}(\vz)}.
\end{align*}
In particular, we have 
\begin{align*} 
    \sum_{\vx\in \FF_q^n}|f(\vx)|^2=\sum_{\vz\in \FF_q^n}|\hat{f}(\vz)|^2.
\end{align*}
\end{lemma}

\begin{lemma}\label{lem:QFT_prod}
Suppose that we have $f_i:\FF_q \rightarrow \mathbb{C}$ for $i\in [n]$ and $f:\FF_q^n\ra \mathbb{C}$ is defined by
\begin{align}\label{eq:f_product}
    f(\vx)\defeq \prod_{i\in [n]}f_i(\vx_i)
\end{align}
where $\vx=(\vx_1,\vx_2,...,\vx_n)$. 
Then, we have 
\begin{align*}
    \hat{f}(\vz)=\prod_{i\in [n]}\hat{f}_i(\vz_i)
\end{align*}
where $\vz=(\vz_1,\vz_2,...,\vz_n)$. 
\end{lemma}

\begin{lemma}[Convolution theorem]\label{lem:convolution}
For functions $f:\FF_q^n\ra \mathbb{C}$ and $g:\FF_q^n\ra \mathbb{C}$, the following  holds.
\begin{align}\label{eq:conv_one}
    \widehat{f\cdot g} = \frac{1}{q^{n/2}}(\hat{f} \ast \hat{g}). 
\end{align}
\end{lemma}


\subsection{Error Correcting Codes.}\label{sec:codes}
In this section, we review basic definitions on error correcting codes.

\paragraph{Codes.} 
A code of length $n\in \mathbb{N}$ over an alphabet $\Sigma$ (which is a finite set) is a subset $C\subseteq \Sigma^n$.

The minimum distance of a code $C$ is defined as the 
smallest Hamming distance between two distinct elements of $C$, that is, $\min_{\vx\ne \vy\in C}|\{i\in[n]:x_i\ne y_i\}|$ where $\vx=(x_1,\ldots,x_n)$ and $\vy=(y_1,\ldots,y_n)$.


\paragraph{Linear codes.}
A code $C$ is said to be a linear code if its alphabet is $\Sigma=\FF_q$ for some prime power $q$ and $C\subseteq \FF_q^n$ is a linear subspace of $\FF_q^n$.
When the dimension of $C$ is $k$, we call $R=k/n$ the rate of $C$. 

We say that a linear code $C$ of length $n$ and dimension $k$ is a maximum distance separable (MDS) code if its minimum distance is equal to $n-k+1$. Note that this is equivalent to $\min_{\vx\in C\setminus \{0\}}\wt(\vx)=n-k+1$. 

\paragraph{Dual codes.}
Let $C$ be a linear code of length $n$ and dimension $k$ over $\FF_q$. The \emph{dual code} $C^\perp$ of $C$ is defined as the orthogonal complement of $C$ as a linear space over $\FF_q$, i.e., 
$$C^{\perp}\defeq \{\vz\in \FF_q^n: \vx \cdot \vz = 0 \text{~for~all~}\vx \in C\}.$$
$C^{\perp}$ is a linear code of length $n$ and dimension $n-k$ over $\FF_q$. 

It is well-known that a dual of an MDS code is also an MDS code. 

The following lemma is also well-known. See e.g., \cite{JACM:YamZha24} for the proof. 

\begin{lemma}\label{lem:fourier_dual}
For a linear code $C\subseteq \FF_q^n$, if we define 
\[
f(\vx)\defeq 
\begin{cases}
\frac{1}{\sqrt{|C|}}& \vx\in C\\
0& \text{otherwise}
\end{cases},
\]
then we have 
\[
\hat{f}(\vz)= 
\begin{cases}
\frac{1}{\sqrt{|C^{\perp}|}}& \vz\in C^{\perp}\\
0& \text{otherwise}
\end{cases}.
\]
In other words, when the dimension of $C$ is $k$,  
\[
\hat{\one}_{C}=q^{(2k-n)/2}\one_{C^\perp}.
\]
\end{lemma}

\paragraph{List decoding.}
We say that a linear code $C\subseteq \FF_q^n$ is $(r,L)$-list-decodable if there is a classical deterministic polynomial-time algorithm $\mathsf{ListDecode}$ that satisfies the following:\footnote{Strictly speaking, the list-decodability is defined for a family of codes with increasing parameter sizes rather than one with a fixed parameter to make the efficiency requirement meaningful.}  
For any $\vx\in C$ and $\ve\in B_{rn}$, $\mathsf{ListDecode}(\vx+\ve)$  outputs $\mathcal{L}\subseteq \FF_q^n$ such that 
$|\mathcal{L}|\le L$ and $\vx \in \mathcal{L}$. 

Without loss of generality, we assume the following. 
 Let $Y\subseteq \FF_q^n$ be a set defined as 
 $$Y=\{\vx+\ve\mid \vx\in C, \ve\in  B_{rn}\}.$$ 
 For any $\vy \in Y$, $\mathsf{ListDecode}(\vy)$ outputs 
 $$\mathcal{L}_\vy=\{\vx\in C \mid 
 \vy-\vx \in B_{rn}\}.$$

This assumption is without loss of generality. Indeed, by the definition of list decodability, for every $\vx\in C$ such that $\vy-\vx\in B_{rn}$, the list output by $\mathsf{ListDecode}(\vy)$ contains $\vx$. Hence the output list contains $\mathcal L_{\vy}$. Moreover, membership in $\mathcal L_{\vy}$ can be efficiently checked, since it suffices to check whether $\vx\in C$ and whether $\wt(\vy-\vx)\le rn$. Therefore, by filtering out all elements that do not belong to $\mathcal L_{\vy}$, we may assume that the output of $\mathsf{ListDecode}(\vy)$ is exactly $\mathcal L_{\vy}$.


\paragraph{(Generalized) Reed--Solomon codes.}
 We review the definition and known facts on (generalized) Reed--Solomon codes. 
 See e.g., \cite[Section 6]{Lindell_Coding} for more details.

Let $q$ be a prime power, $n\le q$ be a positive integer, $0\le d< n$ be an integer,  
$\valpha=(\alpha_1,\alpha_2,\ldots,\alpha_n)\in \FF_q^n$ be a pairwise distinct tuple, and $\vecv=(v_1,...,v_n)\in ({\FF_q\setminus \{0\}})^n$.   
A generalized Reed--Solomon code $\GRS_{\FF_q,\valpha,d,\vecv}$ over $\FF_q$ w.r.t. $\valpha$, $d$, and $\vecv$ is defined as follows:
\[
\GRS_{\FF_q,\valpha,d,\vecv}\defeq \{(v_1f(\alpha_1),v_2f(\alpha_2)...v_nf(\alpha_n)):f\in \FF_q[x]_{deg\leq d}\}
\]
where $\FF_q[x]_{deg\leq d}$ denotes the set of polynomials over $\FF_q$ of degree at most $d$. 
We remark that $\GRS_{\FF_q,\valpha,d,\vecv}$ is a linear code over $\FF_q$ that has length $n$ and dimension $d+1$.  
It is easy to see that a generalized Reed--Solomon code is an MDS code. 

A Reed--Solomon code is a special case of a generalized Reed--Solomon code where  $\vecv=(1,1,\ldots,1)$. 
We denote it by $\RS_{\FF_q,\valpha,d}$ (which is equivalent to $\GRS_{\FF_q,\valpha,d,(1,1,\ldots,1)})$. 
The dual of $\RS_{\FF_q,\valpha,d}$ is $\GRS_{\FF_q,\valpha,n-d-2,\vecv}$ for some $\vecv\in \FF_q^n$~\cite[Claim 6.3]{Lindell_Coding}. 

It is known that generalized Reed--Solomon codes are list-decodable for relative noise ratio up to $1-\sqrt{R}$ where $R=(d+1)/n$ is the rate~\cite{GS99,CHK2026}. 
More precisely, for any constant $\epsilon>0$, there is a polynomial $\poly$ such that  
$\GRS_{\FF_q,\valpha,d,\vecv}$ is $(\ell,L)$-list-decodable where $\ell=1-\sqrt{R}-\epsilon$ and $L=\poly(n)$.

\subsection{Other Lemmas}
We rely on the following well-known lemmas.
\begin{lemma}[Chernoff Bound]\label{lem:Chernoff}
Let $X_1,...,X_n$ be independent random variables taking values in $\bit$, $X\defeq \sum_{i\in [n]}X_i$, and $\mu\defeq \Ex[X]$.
For any $\delta\geq 0$, it holds that 
\begin{align*}
    \Pr[X\geq (1+\delta)\mu]\leq e^{-\frac{\delta^2 \mu}{2+\delta}}.
\end{align*}
\end{lemma}

\begin{lemma}[H\"{o}lder's inequality]\label{lem:Holder}
Let $n$ be a positive integer, $(x_1,x_2,...,x_n)\in \mathbb{C}^n$, $(y_1,y_2,...,y_n)\in \mathbb{C}^n$, and $p,q>1$ be reals such that $1/p+1/q=1$. Then it holds that 
\[
\sum_{i=1}^{n}|x_i y_i|\le \left(\sum_{i=1}^{n}|x_i|^p\right)^{\frac{1}{p}}
\left(\sum_{i=1}^{n}|y_i|^q\right)^{\frac{1}{q}}.
\]
\end{lemma}

\begin{lemma}\label{lem:2norm_to_tracenorm}
Let $\ket{\psi_1}$ and $\ket{\psi_2}$ be not necessarily normalized states of the same size.
Suppose that we have  
\[
\|\ket{\psi_1}-\ket{\psi_2}\| \le \epsilon.
\]
Then we have 
\[
\TD\left(\frac{\ketbra{\psi_1}{\psi_1}}{\|\ket{\psi_1}\|^2}-\frac{\ketbra{\psi_2}{\psi_2}}{\|\ket{\psi_2}\|^2}\right) \le \frac{\epsilon}{\max\{\|\ket{\psi_1}\|,\|\ket{\psi_2}\|\}}.
\]
\end{lemma}
\begin{proof}
Let $\ket{\tilde{\psi}_1}=\frac{\ket{\psi_1}}{\|\psi_1\|}$ and $\ket{\tilde{\psi}_2}=\frac{\ket{\psi_2}}{\|\psi_2\|}$.

We have 
\begin{align*}
1-\left|\braket{\tilde{\psi}_1}{\tilde{\psi}_2}\right|^2\le 
\frac{\|\ket{\psi_1}-\ket{\psi_2}\|^2}{\|\ket{\psi_2}\|^2}.
\end{align*}
Indeed, 
\begin{align*}
&\|\ket{\psi_1}-\ket{\psi_2}\|^2-\|\ket{\psi_2}\|^2\left(1-\left|\braket{\tilde{\psi}_1}{\tilde{\psi}_2}\right|^2\right)\\
=~&\|\ket{\psi_1}\|^2-\braket{\psi_1}{\psi_2}-\braket{\psi_2}{\psi_1}+\|\ket{\psi_2}\|^2\left|\braket{\tilde{\psi}_1}{\tilde{\psi}_2}\right|^2\\
=~& \left|\|\ket{\psi_1}\|-\|\ket{\psi_2}\|\braket{\tilde{\psi}_1}{\tilde{\psi}_2}\right|^2\\
\ge~& 0. 
\end{align*}
Therefore, 
\begin{align*}
    \TD\left(\ketbra{\tilde{\psi}_1}{\tilde{\psi}_1},\ketbra{\tilde{\psi}_2}{\tilde{\psi}_2}\right)
    &=\sqrt{1-\left|\braket{\tilde{\psi}_1}{\tilde{\psi}_2}\right|^2}\\
    &\le \frac{\|\ket{\psi_1}-\ket{\psi_2}\|}{\|\ket{\psi_2}\|}\\
    &\le \frac{\epsilon}{\|\ket{\psi_2}\|}.
\end{align*}

Similarly, we can show
\[
\TD\left(\ketbra{\tilde{\psi}_1}{\tilde{\psi}_1},\ketbra{\tilde{\psi}_2}{\tilde{\psi}_2}\right)\le \frac{\epsilon}{\|\ket{\psi_1}\|}.
\]
Thus, the lemma follows. 
\end{proof}

\subsection{Maximum Linear Satisfaction and Optimal Polynomial Intersection Problems}

The Optimal Polynomial Intersection (OPI) Problem is defined in Problem \ref{prob:OPI}. 

Here, we define Maximum Linear Satisfaction (MaxLINSAT) Problem~\cite{DQI} that generalizes OPI.  

\begin{problem}[Maximum Linear Satisfaction Problem]\label{prob:MaxLINSAT}
Let $q$ be a prime power, let $k\le n$ be positive integers, let
$\mA \in \mathbb{F}_q^{n\times k}$, let $S_i \subseteq \mathbb{F}_q$ for each
$i\in[n]$, and let $0\le s \le 1$.

The maximum linear satisfaction problem
$\MaxLINSAT(q,\mA,\{S_i\}_{i\in[n]},s)$ is the following problem:
\begin{description}
    \item[Given:] $q$, $\mA$, $s$, and quantum oracle access to the membership
    oracles for the sets $S_i$, $i\in[n]$.
    \item[Find:] A vector $\vy\in \mathbb{F}_q^n$ in the image of $\mA$ such that\footnote{\cite{DQI} formalizes MaxLINSAT as the problem of outputting $\vx\in\FF_q^k$ such that $\vy=\mA\vx$ satisfies the condition below. This formulation is equivalent to ours, since given such a vector $\vy$, one can efficiently recover a corresponding $\vx$ by Gaussian elimination.}
    \[
        \left|\{i\in[n] : \vy_i \in S_i\}\right| \ge s n.
    \]
\end{description}
\end{problem}

As observed in \cite{DQI}, OPI is a special case of MaxLINSAT. 
Indeed, $\OPI(k,q,\valpha,\{S_i\}_{i\in[n]},s)$  is equivalent to 
$\MaxLINSAT(q,\mA(\RS_{\FF_q,\valpha,k-1}),\{S_i\}_{i\in[n]},s)$ where $\mA(\RS_{\FF_q,\valpha,k-1})$ is the generator matrix of the Reed--Solomon code $\RS_{\FF_q,\valpha,k-1}$.  
\section{MDS Brascamp--Lieb Inequality}\label{sec:BL}
We show a theorem that can be seen as a variant of the Brascamp--Lieb inequality~\cite{BL76} adapted to the setting of MDS codes.\footnote{
Subsequent works~\cite{sun2026improvedlocalleakageresilience,batra2026quantumleakageresilienceshamir}
provide simpler proofs of the MDS Brascamp--Lieb inequality.
} 
\begin{theorem}[MDS Brascamp--Lieb inequality]\label{thm:MDS_BL}
Let $q$ be a prime power and let $n,k$ be positive integers with $k\le n$. 
Let $C\subseteq \FF_q^n$ be an MDS code of dimension $k$, let $R=\frac{k}{n}$, and let $h_i:\FF_q\rightarrow \mathbb{R}_{\ge 0}$ for $i\in [n]$. 
Then it holds that 
\[
\sum_{\vu\in C}\prod_{i\in [n]} h_i(u_i)
\le
\prod_{i\in [n]}
\left(
\sum_{a\in\mathbb F_q}h_i(a)^{1/R}
\right)^{R}.
\]
\end{theorem}

Our proof follows the entropy-subadditivity approach to Brascamp-Lieb
inequalities developed by~\cite{CarlenCordero2009}.

To prove \Cref{thm:MDS_BL}, we first recall a lemma that can be viewed as the finite-alphabet form of the Gibbs variational principle; see, e.g., \cite[Lemma 4.10]{van2016probability}. 
\begin{lemma}[Gibbs variational principle]\label{lem:Gibbs}
Let $X$ be a finite set and let $f:X\rightarrow \mathbb{R}$. 
Then it holds that 
\[
\log\sum_{v\in X}2^{f(v)}
=
\sup_{P\in \mathcal{P}(X)}
\left\{\sum_{u\in X}P(u)f(u)+H(P)
\right\}
\]
where $\mathcal{P}(X)$ denotes the set of all probability distributions on $X$
and $H(\cdot)$ denotes the Shannon entropy. 
\end{lemma}
\Cref{lem:Gibbs} follows from  \cite[Lemma~4.10]{van2016probability} by restricting to finite sets and taking one of the two probability measures to be the uniform measure.
For the reader's convenience, we provide a proof of \Cref{lem:Gibbs}  in \Cref{sec:proof_Gibbs}. 

Below, we prove \Cref{thm:MDS_BL}.
\begin{proof}[Proof of \Cref{thm:MDS_BL}]
     We first reduce to the case where each $h_i$ is strictly positive on $\FF_q$.
Suppose the theorem has been proved under this assumption. For general
nonnegative $h_i$, fix $\epsilon>0$ and apply the strictly positive case to 
$h_{i,\epsilon}=h_i+\epsilon$. 
Letting $\epsilon\rightarrow 0$ then yields the desired inequality, since all
sets involved are finite and both sides are continuous in the values of the
functions $h_i$.

Thus, in the remainder of the proof, we assume that $h_i(u_i)>0$ for all
$i\in[n]$ and all $u_i\in\FF_q$.
     
     Since $C$ is an MDS code, any pair of distinct codewords cannot agree on more than $k-1$ coordinates. Thus,  for any $I\subseteq [n]$ with $|I|=k$, if we let $\pi_I:C\rightarrow \FF_q^k$ be the corresponding projection function, i.e., 
     $
     \pi_I(x_1,\ldots,x_n)=(x_{i_1},...,x_{i_k})
     $
    where $I=\{i_1,...,i_k\}$ for $i_1<i_2<\cdots<i_k$, then $\pi_I$ is injective. 

    Let $P$ be a probability distribution on $C$. 
    For $i\in [n]$, 
    let $P_i$ be the marginal distribution of $P$ on the $i$th coordinate, 
    and for $I=\{i_1,...,i_k\}$ where $i_1<i_2<\cdots<i_k$, 
    let $P_I$ be the marginal distribution of $P$ on coordinates in $I$. 
    That is, 
    $$P_i(u_i)=\sum_{(u_j)_{j\in [n]\setminus \{i\}}\in \FF_q^{n-1}}P(u_1,\ldots,u_n),$$
    and 
$$P_I(u_{i_1},\ldots,u_{i_k})=\sum_{(u_j)_{j\in [n]\setminus I}\in \FF_q^{n-k}}P(u_1,\ldots,u_n),$$
    where  we extend $P$ to a
probability distribution on $\mathbb F_q^n$ by setting $P(\vu)=0$ for
$\vu\notin C$. 

Then we have 
    \[
    H(P)=H(P_I)\le \sum_{i\in I}H(P_i).
    \]
where the equality follows from the injectivity of $\pi_I$ and the inequality follows from subadditivity of Shannon entropy.
    
 Averaging over all subsets $I\subseteq [n]$ with size $k$, 
 \begin{align}
 H(P)\le \frac{1}{\binom{n}{k}}\sum_{I:|I|=k}\sum_{i\in I}H(P_i)
 = \frac{\binom{n-1}{k-1}}{\binom{n}{k}}\sum_{i\in [n]}H(P_i)=R\sum_{i\in [n]}H(P_i). \label{eq:entropy_bound}
 \end{align}
 We apply \Cref{lem:Gibbs} with
 $X=C$ and 
 $f(\vu)=\sum_{i\in[n]} \log  h_i(u_i)$ where $\vu=(u_1,\ldots,u_n)$.
 Then we obtain
 \begin{align*}
 \log \sum_{\vu \in C}\prod_{i\in [n]}h_i(u_i)=\sup_{P\in \mathcal{P}(C)}
\left\{\sum_{\vu\in C}P(\vu)\sum_{i\in [n]}\log h_i(u_i)+H(P)
\right\}. 
\end{align*}
Moreover, the quantity in the supremum can be upper bounded as follows:
\begin{align*}
\sum_{\vu\in C}P(\vu)\sum_{i\in [n]}\log h_i(u_i)+H(P)
&=
\sum_{i\in [n]}\sum_{u_i\in \FF_q}P_i(u_i)\log h_i(u_i)+H(P)\\
&\le  \sum_{i\in [n]}\left(\sum_{u_i\in \FF_q}P_i(u_i)\log h_i(u_i)+RH(P_i)\right)\\
&=\sum_{i\in [n]}R\left( \sum_{u_i\in \FF_q}P_i(u_i)\log h_i(u_i)^{1/R}+H(P_i)\right)\\
&\le \sum_{i\in [n]}R\log \sum_{u_i\in \FF_q}h_i(u_i)^{1/R}
 \end{align*}
 where the first inequality follows from \Cref{eq:entropy_bound} and the second inequality follows from applying  \Cref{lem:Gibbs} with
 $X=\FF_q$ and 
 $f(u_i)=\log  h_i(u_i)^{1/R}$.

 Combining the above, we obtain
 \[
\log \sum_{\vu \in C}\prod_{i\in [n]}h_i(u_i)\le \sum_{i\in[n]}R\log \sum_{u_i\in \FF_q}h_i(u_i)^{1/R}.
 \]
 Exponentiating both sides gives the desired bound.
\end{proof}
\section{Fourier Bias Bounds}\label{sec:Fourier}
In this section, we show a theorem that controls sum of $Q$-th power of Fourier transform of characteristic functions.  
Below, $O$ hides constants that may depend on $\rho$ and $Q$. 

\begin{theorem}\label{thm:Fourier_Qth_power_sum}
    Let $q=p^e$ be a prime power and $S\subseteq \FF_q$ be a subset with $|S|=\rho q$, and $Q\ge 2$ be a real.  
It holds that  
\[
\frac{1}{q^{Q/2}}\sum_{a\in \FF_q\setminus \{0\}}|\hat{\one}_{S}(a)|^{Q}\le 
\begin{cases}
    \left(\rho(1-\rho)\right)^{\frac{4-Q}{2}}\left(\rho_0^3\left(\frac{2}{3}-\rho_0\right)\right)^{\frac{Q-2}{2}}+O(p^{-1}) 
    ~~~&Q\le 4\\
    \left(\frac{\sin(\rho \pi)}{\pi}\right)^{Q-4}\rho_0^3\left(\frac{2}{3}-\rho_0\right)+O(p^{-1}) &Q\ge 4.
\end{cases}
\]
where $\rho_0=\min\{\rho,1-\rho\}$.
\end{theorem}

To prove \Cref{thm:Fourier_Qth_power_sum}, we prove the following lemmas.

\begin{lemma}\label{lem:Fourier_quadratic_sum}
    Let $q$ be a prime power and $S\subseteq \FF_q$ be a subset with $|S|=\rho q$.  
It holds that  
\[
\frac{1}{q}\sum_{a\in \FF_q\setminus \{0\}}|\hat{\one}_{S}(a)|^2= \rho(1-\rho). 
\]
\end{lemma}
\begin{proof}
We have 
\begin{align*}
\sum_{a\in \FF_q}|\hat{\one}_{S}(a)|^2=\sum_{a\in \FF_q}|\one_{S}(a)|^2=|S|,
\end{align*}
and 
\begin{align*}
|\hat{\one}_{S}(0)|^2=\left|\frac{1}{\sqrt{q}}\sum_{x\in \FF_q}\one_S(x)\right|^2=\frac{|S|^2}{q}.
\end{align*}
The lemma follows from them. 
\end{proof}

\begin{lemma}\label{lem:Fourier_quartic_sum}
    Let $q$ be a prime power and $S\subseteq \FF_q$ be a subset with $|S|=\rho q$.  
It holds that  
\[
\frac{1}{q^2}\sum_{a\in \FF_q\setminus \{0\}}|\hat{\one}_{S}(a)|^4\le \rho_0^3\left(\frac{2}{3}-\rho_0\right)+O(p^{-1})
\]
where $\rho_0=\min\{\rho,1-\rho\}$.
\end{lemma}

We prove \Cref{lem:Fourier_quartic_sum} in \Cref{sec:proof_quartic_sum} using results from additive combinatorics. 

\begin{lemma}\label{lem:Fourier_max}
    Let $q=p^e$ be a prime power and $S\subseteq \FF_q$ be a subset with $|S|=\rho q$. 
  For any   $a\in \FF_q\setminus \{0\}$, 
it holds that  
\[
\frac{|\hat{\one}_{S}(a)|}{\sqrt{q}}\le\frac{\sin(\rho \pi)}{\pi} + O(p^{-2}).
\]
\end{lemma}
\Cref{lem:Fourier_max} was proven in 
\cite[Fact 4.12]{SunWootters26} for the special case where $q=p$ is a prime. 
 It can be easily generalized to the case of prime powers. We give a proof in \Cref{sec:proof_Fourier_max} for completeness.  

Below, we prove \Cref{thm:Fourier_Qth_power_sum}  using \Cref{lem:Fourier_quadratic_sum,lem:Fourier_quartic_sum,lem:Fourier_max}. 

\begin{proof}[Proof of \Cref{thm:Fourier_Qth_power_sum}]
First, we prove the case of $2\le Q\le 4$.

Define $\alpha=\frac{Q-2}{2}$.
Then we have $0\le \alpha\le 1$ and 
$Q=2(1-\alpha)+4\alpha$.

Then we have 
\begin{align*}
\frac{1}{q^{Q/2}}\sum_{a\in \FF_q\setminus \{0\}}|\hat{\one}_{S}(a)|^{Q}
&=\frac{1}{q^{Q/2}}\sum_{a\in \FF_q\setminus \{0\}}\left(|\hat{\one}_{S}(a)|^2\right)^{1-\alpha}\left(|\hat{\one}_{S}(a)|^4\right)^{\alpha}\\
&\le \left(\frac{1}{q}\sum_{a\in \FF_q\setminus \{0\}}|\hat{\one}_{S}(a)|^2\right)^{1-\alpha}\left(\frac{1}{q^2}\sum_{a\in \FF_q\setminus \{0\}}|\hat{\one}_{S}(a)|^4\right)^{\alpha}\\
&\le \left(\rho(1-\rho)\right)^{1-\alpha}\left(\rho_0^3\left(\frac{2}{3}-\rho_0\right)\right)^{\alpha}+O(p^{-1}) 
\end{align*}
where the first inequality follows from H\"{o}lder's inequality (\Cref{lem:Holder}) and the second inequality follows from \Cref{lem:Fourier_quadratic_sum,lem:Fourier_quartic_sum}. 
This completes the proof for the case of $2\le Q \le 4$ noting that 
$1-\alpha=\frac{4-Q}{2}$ and $\alpha=\frac{Q-2}{2}$.

 Next, we prove the case of $Q\ge 4$. We have 
 \begin{align*}
\frac{1}{q^{Q/2}}\sum_{a\in \FF_q\setminus \{0\}}|\hat{\one}_{S}(a)|^{Q}
&\le \left(\max_{a\in \FF_q\setminus\{0\}}\left\{
\frac{|\hat{\one}_{S}(a)|}{\sqrt{q}}
\right\}\right)^{Q-4}\left(\frac{1}{q^2}\sum_{a\in \FF_q\setminus \{0\}}|\hat{\one}_{S}(a)|^4\right)\\
&\le \left(\frac{\sin(\rho \pi)}{\pi}\right)^{Q-4}\rho_0^3\left(\frac{2}{3}-\rho_0\right)+O(p^{-1}), 
\end{align*}
where the second inequality follows from \Cref{lem:Fourier_max,lem:Fourier_quartic_sum}.

This completes the proof of \Cref{thm:Fourier_Qth_power_sum}. 
\end{proof}

\subsection{Proof of  \Cref{lem:Fourier_quartic_sum}}\label{sec:proof_quartic_sum}
In this section, we prove \Cref{lem:Fourier_quartic_sum}. 
We first review necessary background on additive combinatorics.  
\paragraph{Preliminaries from additive combinatorics.}

We review notations and results from additive combinatorics 

Let $(G,+)$ be a finite abelian group.
For  subsets 
$A,B\subseteq G$ and 
$x\in G$, define 
\[
r_{A,B}(x)=|\{(a,b)\in A\times B:a+b=x\}|.
\]

The following theorem was shown by Green and Ruzsa~\cite{GR05} as a generalization of Pollard's theorem~\cite{Pollard1974} and Kneser's theorem~\cite{Kneser1953}.
\begin{theorem}[{\cite[Proposition 6.1]{GR05}}]\label{thm:green_ruzsa}
Let $(G,+)$ be a finite abelian group, $A,B\subseteq G$ be subsets, and 
$t$ be a positive integer with $t\le \min\{|A|,|B|\}$.
 Then it holds that  
\[
\sum_{x\in G}\min\{t,r_{A,B}(x)\}\ge t\cdot \min\{|G|,|A|+|B|-D(G)-t\},
\]
where $D(G)$ is the maximum size of a proper subgroup of $G$. 
\end{theorem}

In particular, for a prime power $q=p^e$, $\FF_q$ is isomorphic to $(\mathbb{Z}/(p\mathbb{Z}))^e$ as an additive group. 
Clearly, the largest proper subgroup of $(\mathbb{Z}/(p\mathbb{Z}))^e$ is $(\mathbb{Z}/(p\mathbb{Z}))^{e-1}$, of which size is $p^{e-1}$.
Also, when $A=B=S$ for some $S\subseteq \FF_q$ with $|S|\le q/2$, we have $|A|+|B|-p^{e-1}-t=2|S|-p^{e-1}-t<q$ for any positive integer $t$. 
Thus, the above theorem immediately implies the following corollary. 

\begin{corollary}\label{cor:green_ruzsa}
Let $q=p^e$ be a prime power,  $S\subseteq \FF_q$ be a subset with $|S|\le q/2$, and 
$t$ be a positive integer with $t\le |S|$.
 Then it holds that  
\[
\sum_{x\in \FF_q}\min\{t,r_{S,S}(x)\}\ge t\cdot (2|S|-p^{e-1}-t).
\]
\end{corollary}

Let $q$ be a prime power.
For $S\subseteq \FF_q$, define its additive energy
\[
E(S)=|\{(x_1,x_2,x_3,x_4)\in S^4: x_1+x_2=x_3+x_4\}|.
\]

It is well-known that additive energy can be written as quartic sum of Fourier coefficients.  We provide a proof for completeness. 
\begin{lemma}\label{lem:additive_energy}
    For a prime power $q$ and $S\subseteq \FF_q$, 
    it holds that 
    \[E(S)= q\sum_{a\in \FF_q}|\hat{\one}_{S}(a)|^4\]
\end{lemma}
\begin{proof}
We have 
\begin{align*}
   \sum_{a\in \FF_q}|\hat{\one}_{S}(a)|^4
   &= \sum_{a\in \FF_q}\hat{\one}_{S}(a)^2\overline{\hat{\one}_{S}(a)^2}\\
   &=\sum_{a\in \FF_q}\left(\frac{1}{\sqrt{q}}\sum_{x_1\in \FF_q}\one_S(x_1)\omega_p^{\Tr(ax_1)}\right)
   \left(\frac{1}{\sqrt{q}}\sum_{x_2\in \FF_q}\one_S(x_2)\omega_p^{\Tr(ax_2)}\right)\\
   &\quad\left(\frac{1}{\sqrt{q}}\sum_{x_3\in \FF_q}\one_S(x_3)\omega_p^{-\Tr(ax_3)}\right)
   \left(\frac{1}{\sqrt{q}}\sum_{x_4\in \FF_q}\one_S(x_4)\omega_p^{-\Tr(ax_4)}\right)
   \\
   &=\frac{1}{q^2}\sum_{(x_1,x_2,x_3,x_4)\in S^4}\sum_{a\in \FF_q}
   \omega_p^{\Tr(a(x_1+x_2-x_3-x_4))}\\
   &=\frac{1}{q^2}\sum_{\substack{(x_1,x_2,x_3,x_4)\in S^4:\\
   x_1+x_2=x_3+x_4
   }}q\\
   &=\frac{E(S)}{q},
\end{align*}
where the fourth equality follows from \Cref{eq:sum_is_zero}.
\end{proof}

\paragraph{Bounding additive energy.}

By \Cref{lem:additive_energy}, 
proving \Cref{lem:Fourier_quartic_sum} is reduced to bounding additive energy. We do so by using \Cref{thm:green_ruzsa}. 

\begin{lemma}\label{lem:additive_energy_bound}
    For a prime power $q=p^e$ and $S\subseteq \FF_q$ with $|S|\le q/2$, 
    it holds that 
    \[E(S)\le  \frac{2|S|^3+3p^{e-1}|S|^2+(1-3p^{e-1})|S|}{3}.\]
\end{lemma}
\begin{proof}
For $x\in \FF_q$, 
define 
\[r(x)=r_{S,S}(x)=|\{(a,b)\in S^2:a+b=x\}|.\]
Then we have 
\begin{align}
\sum_{x\in \FF_q}r(x)=|S|^2 \label{eq:sum_r}
\end{align}
and 
\begin{align}
\sum_{x\in \FF_q} r(x)^2=E(S). \label{eq:sum_r_square}
\end{align}

For $x\in \FF_q$, we have the following  identity 
\begin{align*}
r(x)^2&=r(x)+r(x)(r(x)-1)\\
&=r(x)+2\sum_{t=1}^{r(x)-1}(r(x)-t)\\
&=r(x)+2\sum_{t=1}^{|S|-1}\max\{r(x)-t,0\}.
\end{align*} 
Summing over $x\in \FF_q$, by \Cref{eq:sum_r,eq:sum_r_square}, 
we have
\begin{align}
E(S)=|S|^2+2\sum_{t=1}^{|S|-1}\sum_{x\in \FF_q}\max\{r(x)-t,0\}. \label{eq:energy_bound}
\end{align}

Below, we bound $\sum_{x\in \FF_q}\max\{r(x)-t,0\}$.

For any $t\in \{1,2,...,|S|-1\}$, 
we have 
\begin{align*}
\sum_{x\in \FF_q}\max\{r(x)-t,0\}
&=\sum_{x\in \FF_q}r(x)-\sum_{x\in \FF_q}\min(t,r(x))\\
&\le |S|^2-t(2|S|-p^{e-1}-t)\\
&=(|S|-t)^2+t p^{e-1}.
\end{align*}
where the inequality in the second line follows from \Cref{eq:sum_r}, \Cref{cor:green_ruzsa}, and $|S|\le q/2$. 

Substituting this into \Cref{eq:energy_bound}, we have
\begin{align*}
E(S)
&\le |S|^2+2\sum_{t=1}^{|S|-1}\left((|S|-t)^2+tp^{e-1}\right)\\
&=|S|^2+\frac{(|S|-1)\cdot|S|\cdot (2|S|-1)}{3}+(|S|-1)|S|p^{e-1}\\
&=\frac{2|S|^3+3p^{e-1}|S|^2+(1-3p^{e-1})|S|}{3}.
\end{align*}
This completes the proof of \Cref{lem:additive_energy_bound}.
\end{proof}

Given \Cref{lem:additive_energy,lem:additive_energy_bound}, the proof of \Cref{lem:Fourier_quartic_sum} is straightforward.
\begin{proof}[Proof of \Cref{lem:Fourier_quartic_sum}]
We first show the case of $\rho\le 1/2$. 
By a straightforward combination of \Cref{lem:additive_energy,lem:additive_energy_bound}, we have 
\[
\frac{1}{q^2}\sum_{a\in \FF_q}|\hat{\one}_{S}(a)|^4\le \frac{2|S|^3+3p^{e-1}|S|^2+(1-3p^{e-1})|S|}{3q^3}.
\]
We also have 
\[
|\hat{\one}_{S}(0)|=\frac{1}{\sqrt{q}}\sum_{x\in \FF_q}\one_S(x)=\frac{|S|}{\sqrt{q}}
\]
Combining them, 
\[
\frac{1}{q^2}\sum_{a\in \FF_q\setminus \{0\}}|\hat{\one}_{S}(a)|^4\le \frac{2|S|^3+3p^{e-1}|S|^2+(1-3p^{e-1})|S|}{3q^3}-\frac{|S|^4}{q^4}=\rho^3\left(\frac{2}{3}-\rho\right)+O(p^{-1})
\]
as desired.

When $\rho>1/2$, let $S^c$ be the complement of $S$.
Then $|S^c|/q=1-\rho< 1/2$. 
Moreover, 
    since $\one_{S}+\one_{S^c}=1$,  
    by \Cref{eq:sum_is_zero}, for any $a\in \FF_q\setminus \{0\}$, we have 
    $\hat{\one}_S(a)+\hat{\one}_{S^c}(a)=0$ and thus $|\hat{\one}_S(a)|=|\hat{\one}_{S^c}(a)|$. Thus, applying the same argument to $S^c$ gives the desired bound. 
\end{proof}

\subsection{Proof of \Cref{lem:Fourier_max}}\label{sec:proof_Fourier_max}
\begin{proof}[Proof of \Cref{lem:Fourier_max}]
 The case where $q=p$ is a prime is proven in \cite[Fact 4.12]{SunWootters26}.\footnote{
 The normalization factor in the  definition of Fourier transform in \cite{SunWootters26} is different from ours; in their definition $\hat{\one}_S(z)=\frac{1}{p}\sum_{x\in S}\omega_p^{-xz}$ when $q=p$ is a prime. 
 } 
 
 We show the general case by reducing it to the case of prime $q$. 
Assume that $e\ge 2$.  
 For $t\in \FF_p$, let 
 \[
 c_t=|\{x\in S:\Tr(xa)=t\}|.
 \]
Note that $c_t\le |\{x\in \FF_q:\Tr(xa)=t\}|=p^{e-1}$. 
 For $j\in [p^{e-1}]$, define 
 \[
 T_j=\{t\in \FF_p: c_t\ge j\}.
 \]
Then we have
\begin{align}
\sum_{j\in [p^{e-1}]}|T_j|=\sum_{t\in \FF_p}c_t=|S|. \label{eq:sum_T_j}
\end{align}

We have  
\begin{align}
    \sum_{x\in S}\omega_p^{-\Tr(xa)}
    =\sum_{t\in \FF_p}c_t\omega_p^{-t}
    =\sum_{j\in [p^{e-1}]}\sum_{t\in T_j}\omega_p^{-t}. \label{eq:sum_expand}
\end{align}
For each $j\in [p^{e-1}]$, 
let $\rho_j=\frac{|T_j|}{p}$. 
By the lemma for the case of $q=p$, 
we have 
\begin{align}
\frac{1}{p}\left|\sum_{t\in T_j}\omega_p^{-t}\right|\le \frac{\sin(\rho_j \pi)}{\pi}+O(p^{-2}). \label{eq:base_case_q=p}
\end{align} 
Then we have 
\begin{align}
 \begin{split}
    \frac{|\hat{\one}_S(a)|}{\sqrt{q}}
    &=
    \frac{1}{q}\left|\sum_{x\in S}\omega_p^{-\Tr(xa)}\right|\\
    &\le \frac{1}{p^{e-1}}\cdot\sum_{j\in [p^{e-1}]}\frac{1}{p}\left|\sum_{t\in T_j}\omega_p^{-t}\right|\\
    &\le \frac{1}{p^{e-1}}\cdot\sum_{j\in [p^{e-1}]}\frac{\sin(\rho_j \pi)}{\pi}+O(p^{-2}),
    \end{split}\label{eq:upper_bound_sum}
\end{align}
where the first inequality follows from \Cref{eq:sum_expand} and the triangle inequality, and 
the second inequality follows from 
\Cref{eq:base_case_q=p}. 
We note that 
\[
\frac{1}{p^{e-1}}\sum_{j\in [p^{e-1}]}\rho_j=\frac{1}{q}\sum_{j\in [p^{e-1}]}|T_j|=\frac{|S|}{q}=\rho
\]
where the second last equality follows from \Cref{eq:sum_T_j}.

Since the function $u\mapsto \sin(\pi u)$ is concave on $[0,1]$, Jensen's inequality gives 
\begin{align}
\frac{1}{p^{e-1}}\cdot\sum_{j\in [p^{e-1}]}\sin(\rho_j \pi)\le 
\sin\left(\frac{1}{p^{e-1}}\cdot\sum_{j\in [p^{e-1}]}\rho_j \pi\right)=\sin(\rho \pi). \label{eq:Jensen}
\end{align}
\Cref{lem:Fourier_max} follows from \Cref{eq:upper_bound_sum,eq:Jensen}. 
\end{proof}
\section{Regev-like Algorithm with List Decoder}\label{sec:Regev}
In this section, we show a theorem that characterizes a condition of success of Regev-like algorithm with list decoders. We note that a similar algorithm is already considered in \cite{STOC:ChaTil25,chailloux2025opixsoftdecoders}, but they only give average-case analysis, and do not give conditions for the algorithm's success in the worst-case. The theorem can also be seen as a variant of \cite[Lemma 5.1]{JACM:YamZha24}. 
\begin{theorem}\label{thm:Regev_like_new}
Let $q$ be a prime power and let $n$ be a positive integer such that
$\log q\le \poly(n)$. Let $C\subseteq \FF_q^n$ be a linear code of dimension
$k$, and suppose that $C^\perp$ is $(r,L)$-list-decodable for some $L\le \poly(n)$.  
Let $W:\FF_q^n\to\mathbb C$ be a normalized function, i.e.,
$
\sum_{\ve\in \FF_q^n}|W(\ve)|^2=1.
$

Assume that
\begin{align}
&\sum_{\ve\in T_{rn}}|\hat W(\ve)|^2\le \negl(n), \label{eq:cond_one}\\
&q^{-(n-k)}
\sum_{\vy\in \FF_q^n}
\left|
\sum_{\vx\in C^\perp}
\one_{T_{rn}}(\vy-\vx)\hat W(\vy-\vx)
\right|^2
\le \negl(n), \label{eq:cond_two}\\
&\left|
q^{n-k}\sum_{\vx\in C}|W(\vx)|^2-1
\right|
\le \negl(n) \label{eq:cond_three}
\end{align}
where we recall that $T_{rn}=\{\vx\in \FF_q^n: \wt(\vx)> rn\}$.

Then there exists a quantum polynomial-time algorithm that, given a
description of $C$ and the quantum state
$
\sum_{\ve\in \FF_q^n}W(\ve)\ket{\ve},
$
outputs a state $\rho_{\rm out}$ satisfying
$$
\TD\left(
\rho_{\rm out},
\ketbra{\tilde\psi_{\rm target}}{\tilde\psi_{\rm target}}
\right)
\le 
1-\frac{1}{\poly(n)}, 
$$
where 
$$
\ket{\tilde\psi_{\rm target}}
=
\frac{\ket{\psi_{\rm target}}}{\|\ket{\psi_{\rm target}}\|},
\qquad
\ket{\psi_{\rm target}}
=
\sum_{\vx\in C}W(\vx)\ket{\vx}.
$$ 
\end{theorem}

In particular, the following corollary will be used in our worst-case quantum algorithm for OPI in \Cref{sec:main}.

\begin{corollary}\label{cor:Regev_like_new}
Let $q$ be a prime power and let $n$ be a positive integer such that
$\log q\le \poly(n)$. Let $C\subseteq \FF_q^n$ be a linear code of dimension
$k$, and suppose that $C^\perp$ is $(r,L)$-list-decodable for
some $L\le \poly(n)$. 

For each $i\in[n]$, let $S_i\subseteq \FF_q$ be a subset satisfying $|S_i|=\rho q$ where $\rho\ge 1/\poly(n)$. 

Define $W:\FF_q^n\to\mathbb C$ by $$ W(e_1,\ldots,e_n)=\prod_{i=1}^n W_i(e_i), $$ where $$ W_i(e_i)= \begin{cases} 
\sqrt{\dfrac{\tau}{|S_i|}} & \text{if } e_i\in S_i,\\[1.2ex] \sqrt{\dfrac{1-\tau}{q-|S_i|}} & \text{if } e_i\notin S_i \end{cases} $$ and $\rho\le \tau\le 1$.

Assume that
\begin{align}
&r\ge 1-\left(\sqrt{\tau\rho}+\sqrt{(1-\tau)(1-\rho)}\right)^2+\Omega(1), \label{eq:cond_one_cor}\\
&q^{-(n-k)}
\sum_{\vy\in \FF_q^n}
\left|
\sum_{\vx\in C^\perp}
\one_{T_{rn}}(\vy-\vx)\hat W(\vy-\vx)
\right|^2
\le \negl(n), \label{eq:cond_two_cor}\\
&\left|
q^{n-k}\sum_{\vx\in C}|W(\vx)|^2-1
\right|
\le \negl(n), \label{eq:cond_three_cor}\\
&\left|
q^{n-k}\sum_{\vx\in C_{sn}}|W(\vx)|^2
\right|
\le \negl(n). \label{eq:cond_four_cor}
\end{align}
where 
$0\le s\le 1$, 
$C_{sn}=\{\vx=(x_1,...,x_n) \in C: \left|\{i\in[n] : x_i \in S_i\}\right| < s n\}$,  and we recall that $T_{rn}=\{\vx\in \FF_q^n: \wt(\vx)> rn\}$.

Then there exists a quantum polynomial-time algorithm that solves $\MaxLINSAT(q,\mA,\{S_i\}_{i\in[n]},s)$ with probability $1/\poly(n)$ where $\mA$ is the generator matrix of $C$.  
\end{corollary}

In the rest of this section, we prove \Cref{thm:Regev_like_new} and \Cref{cor:Regev_like_new}.
\subsection{Proof of \Cref{thm:Regev_like_new}}
\begin{proof}[Proof of \Cref{thm:Regev_like_new}]
Let $V$ be the normalized indicator function of $C$, that is, 
\begin{align*}
    V(\vx)=
    \begin{cases}
    \frac{1}{\sqrt{|C|}}& \vx\in C\\
    0& \text{otherwise}
    \end{cases}.
\end{align*}
By \Cref{lem:fourier_dual}, we have 
\begin{align*}
    \hat{V}(\vx)=
    \begin{cases}
    \frac{1}{\sqrt{|C^\perp|}}& \vx\in C^\perp\\
    0& \text{otherwise}
    \end{cases}.
\end{align*}

Let $\mathsf{ListDecode}$ be the list-decoder of $C^\perp$. 
For $\rand \in [2^n]$, let $\mathsf{Decode}_\rand$ be an algorithm that works as follows:
\begin{description}
    \item[$\mathsf{Decode}_\rand(\vy)$:]
    On input $\vy \in \FF_q^n$, 
run $\mathsf{ListDecode}(\vy)$ to obtain a list $\mathcal{L}\subseteq \FF_q^n$. 
If $\mathcal{L}=\emptyset$, output $0$.
Otherwise, 
let $\mathcal{L}=\{\vx_1,\vx_2,...,\vx_\ell\}$, and  
output $\vx_i$ for $i=1+(\rand \mod \ell)$.
\end{description}
For any $\vx\in C^\perp$ and $\ve\in B_{rn}$, 
if we let $\mathcal{L}=\mathsf{ListDecode}(\vx+\ve)$, then for any $\vx'\in \mathcal{L}$, it holds that
\begin{align} \label{eq:rand_decode}
    \left|\Pr_{\rand \gets [2^n]}[\mathsf{Decode}_{\rand}(\vx+\ve)=\vx']-\frac{1}{|\mathcal{L}|}\right|\le 2^{-n}.
\end{align}

We construct the quantum algorithm as follows.
\begin{enumerate}
\item 
The algorithm generates a uniform superposition over $C$. 
Together with the input state $\sum_{\ve\in \FF_q^n}W(\ve)\ket{\ve}$, the state of the joint system can be written as 
\[
\ket{\psi_1}=\sum_{(\vx,\ve)\in \FF_q^n\times \FF_q^n }V(\vx)W(\ve)\ket{\vx}\ket{\ve}.
\]
\item Apply $\QFT_{\FF_q^n}$ on both registers, which results in 
\[
\ket{\psi_2}=\sum_{(\vx,\ve)\in \FF_q^n\times \FF_q^n }\hat{V}(\vx)\hat{W}(\ve)\ket{\vx}\ket{\ve}.
\]
\item Add the value of the second register into the first register. 
This results in the following state.
\begin{align*}
\ket{\psi_3}
=\sum_{(\vx,\ve)\in \FF_q^n\times \FF_q^n  }\hat{V}(\vx)\hat{W}(\ve)\ket{\vx}\ket{\vx+\ve}
\end{align*}
\item 
Choose $\rand\gets [2^n]$ and subtract the value of $\mathsf{Decode}_\rand$ on the second register from the first register, which results in 
\[
\ket{\psi_{4,\rand}}=\sum_{(\vx,\ve)\in \FF_q^n\times \FF_q^n }\hat{V}(\vx)\hat{W}(\ve)\ket{\vx-\mathsf{Decode}_\rand(\vx+\ve)}\ket{\vx+\ve}.
\]
\item  Apply $\QFT_{\FF_q^n}^{-1}$ on both registers, and let $\ket{\psi_{5,\rand}}$ be the resulting state.
\item Output the second register of  $\ket{\psi_{5,\rand}}$ discarding the first register. 
\end{enumerate}

Define the following (not necessarily normalized) states. 
\begin{align*}
&\ket{\psi'_2}=\sum_{(\vx,\ve)\in \FF_q^n\times B_{rn}  }\hat{V}(\vx)\hat{W}(\ve)\ket{\vx}\ket{\ve},\\
    &\ket{\psi'_{4,\rand}}=\sum_{(\vx,\ve)\in \FF_q^n\times B_{rn}  }\hat{V}(\vx)\hat{W}(\ve)\ket{\vx-\mathsf{Decode}_\rand(\vx+\ve)}\ket{\vx+\ve},\\
     &\ket{\psi_{\rm mid}}=\sum_{(\vx,\ve)\in \FF_q^n\times B_{rn}  }\hat{V}(\vx)\hat{W}(\ve)\ket{0}\ket{\vx+\ve}, \\
      &\ket{\psi_{\rm ideal}}=\sum_{(\vx,\ve)\in \FF_q^n\times \FF_q^n  }\hat{V}(\vx)\hat{W}(\ve)\ket{0}\ket{\vx+\ve}. \\
\end{align*} 


For proving that the output of the algorithm satisfies the desired property, 
we prove the following claims.
\begin{claim}\label{cla:four_to_four_prime}
  For any $\rand \in [2^n]$, 
 we have
  \[
 \|\ket{\psi_{4,\rand}}-\ket{\psi'_{4,\rand}}\|\le \negl(n).
 \]
\end{claim}
\begin{proof}
We have    
 \begin{align*}
  \|\ket{\psi_{2}}-\ket{\psi'_{2}}\|
  &=\left\|\sum_{(\vx,\ve)\in \FF_q^n\times T_{rn}  }\hat{V}(\vx)\hat{W}(\ve)\ket{\vx}\ket{\ve}\right\|\\
  &=\left(\sum_{(\vx,\ve)\in \FF_q^n\times T_{rn}}|\hat{V}(\vx)|^2|\hat{W}(\ve)|^2\right)^{1/2}\\  
  &=\left(\sum_{\ve\in T_{rn}}|\hat{W}(\ve)|^2\right)^{1/2}\le \negl(n)
 \end{align*}
 where the final inequality follows from \Cref{eq:cond_one}. 
 Since $\ket{\psi_{4,\rand}}$ and $\ket{\psi'_{4,\rand}}$ are obtained by applying the same unitary to $\ket{\psi_2}$ and $\ket{\psi'_2}$, respectively, \Cref{cla:four_to_four_prime}  follows from the above inequality. 
\end{proof}

\begin{claim}\label{cla:mid_to_ideal}
 We have
\[
\|\ket{\psi_{\rm mid}}-\ket{\psi_{\rm ideal}}\|\le \negl(n).
\]
\end{claim}
\begin{proof}
 We have    
 \begin{align*}
  \|\ket{\psi_{\rm mid}}-\ket{\psi_{\rm ideal}}\|
  &=\left\|\sum_{(\vx,\ve)\in \FF_q^n\times T_{rn}  }\hat{V}(\vx)\hat{W}(\ve)\ket{0}\ket{\vx+\ve}\right\|\\
  &=\left(
  \sum_{\vy \in \FF_q^n}
  \left|\sum_{\substack{\vx\in \FF_q^n:\\\vy-\vx\in T_{rn}}}\hat{V}(\vx)\hat{W}(\vy-\vx)\right|^2\right)^{1/2}\\  
   &=\left(
 q^{-(n-k)}\sum_{\vy\in \FF_q^n}\left|\sum_{\vx\in C^\perp}\one_{T_{rn}}(\vy-\vx)\hat{W}(\vy-\vx)\right|^2\right)^{1/2}\le \negl(n)
 \end{align*}
 where the final inequality follows from \Cref{eq:cond_two}. 
\end{proof}

\begin{claim}\label{cla:ideal_norm}
We have 
\[
\left|\|\ket{\psi_{\rm ideal}}\|-1\right|\le \negl(n).
\]
\end{claim}
\begin{proof}
  By \cref{lem:convolution}, we have 
  \begin{align*}
  (I\times \QFT_{\FF_q^n}^{-1})\ket{\psi_{\rm ideal}}
  &=
  (I\times \QFT_{\FF_q^n}^{-1})\sum_{\vy\in \FF_q^n}(\hat{V}*\hat{W})(\vy)\ket{0}\ket{\vy}\\
  &=q^{n/2}\sum_{\vy\in \FF_q^n}V(\vy)W(\vy)\ket{0}\ket{\vy}
  \end{align*}
  where the second equality follows from \Cref{lem:convolution}. 
    Thus, we have 
    \begin{align*}
    \|\ket{\psi_{\rm ideal}}\|^2
    &=q^n\sum_{\vy\in \FF_q^n}|V(\vy)|^2|W(\vy)|^2\\
    &=q^{n-k}\sum_{\vy\in C}|W(\vy)|^2.
    \end{align*}
\Cref{cla:ideal_norm} follows from the above and  \Cref{eq:cond_three}.
\end{proof}

\begin{claim}\label{cla:fourprime_to_mid}
 We have
\begin{align*}
    \left|\Ex_{\rand\gets [2^n]}\braket{\psi'_{4,\rand}}{\psi_{\rm mid}}\right|\ge 
    L^{-1}-\negl(n).
\end{align*}
\end{claim}
\begin{proof}
 Let $Y\subseteq \FF_q^n$ be a set defined as 
 $$Y=\{\vx+\ve\mid \vx\in C^\perp, \ve\in  B_{rn}\}.$$ 
  By $(r,L)$-list-decodability of $C^\perp$, 
 for any $\vy \in Y$, $\mathsf{ListDecode}(\vy)$ outputs\footnote{As explained in \Cref{sec:codes}, this is assumed without loss of generality.} 
 $$\mathcal{L}_\vy=\{\vx\in C^\perp \mid 
 \vy-\vx \in B_{rn}\}.$$

For $\vy\in Y$ and $\vx\in \mathcal{L}_\vy$, define $a_{\vy,\vx}=\hat{V}(\vx)\hat{W}(\vy-\vx)$. 
For $\rand\in [2^n]$ and  $\vy\in Y$, 
define $\vx_{\rand,\vy}=\mathsf{Decode}_\rand(\vy)$.  
Since $\hat V$ is supported on $C^\perp$, by setting
$\vy=\vx+\ve$,
we can rewrite $\ket{\psi'_{4,\rand}}$ and  $\ket{\psi_{\rm mid}}$ as follows.
\begin{align*}
&\ket{\psi'_{4,\rand}}=\sum_{\vy\in Y}\sum_{\vx\in \mathcal{L}_\vy}a_{\vy,\vx}\ket{\vx-\vx_{\rand,\vy}}\ket{\vy},\\
    &\ket{\psi_{\rm mid}}=\sum_{\vy\in Y}\sum_{\vx\in \mathcal{L}_\vy}a_{\vy,\vx}\ket{0}\ket{\vy}.
\end{align*}
Then we have 
\begin{align*}
    \braket{\psi'_{4,\rand}}{\psi_{\rm mid}}=\sum_{\vy \in Y}
\overline{a_{\vy,\vx_{\rand,\vy}}}
    \left(\sum_{\vx\in \mathcal{L}_\vy}a_{\vy,\vx}\right).
\end{align*}
Let $p_{\vy}(\vx)$ be the probability that $\vx_{\rand,\vy}=\vx$ where $\rand\gets [2^n]$.
By \Cref{eq:rand_decode}, for any $\vy\in Y$ and $\vx\in \mathcal{L}_\vy$, 
if we define $\Delta_{\vy}(\vx)=p_{\vy}(\vx)-|\mathcal{L}_\vy|^{-1}$
we have $|\Delta_{\vy}(\vx)|\le 2^{-n}$. 

Then we have 
\begin{align*}
    \Ex_{\rand\gets [2^n]}\braket{\psi'_{4,\rand}}{\psi_{\rm mid}}
    &=\sum_{\vy \in Y}
    \left(\sum_{\vx'\in\mathcal{L}_\vy}p_{\vy}(\vx')
    \overline{a_{\vy,\vx'}}\right)
    \left(\sum_{\vx\in \mathcal{L}_\vy}a_{\vy,\vx}\right)\\
    &=\sum_{\vy \in Y}
    \left(\sum_{\vx'\in\mathcal{L}_\vy}|\mathcal{L}_\vy|^{-1}
    \overline{a_{\vy,\vx'}}\right)
    \left(\sum_{\vx\in \mathcal{L}_\vy}a_{\vy,\vx}\right)\\
    &~~~+\sum_{\vy \in Y}
    \left(\sum_{\vx'\in\mathcal{L}_\vy}\Delta_{\vy}(\vx')
    \overline{a_{\vy,\vx'}}\right)
    \left(\sum_{\vx\in \mathcal{L}_\vy}a_{\vy,\vx}\right).
\end{align*}
First, we analyze the first term as follows.
\begin{align*}
    &\sum_{\vy \in Y}
    \left(\sum_{\vx'\in\mathcal{L}_\vy}|\mathcal{L}_\vy|^{-1}
    \overline{a_{\vy,\vx'}}\right)
    \left(\sum_{\vx\in \mathcal{L}_\vy}a_{\vy,\vx}\right)\\
    &=\sum_{\vy\in Y}|\mathcal{L}_\vy|^{-1}\left|\sum_{\vx \in \mathcal{L}_\vy}a_{\vy,\vx}\right|^2\\
    &\ge L^{-1}\|\ket{\psi_{\rm mid}}\|^2\\
    &\ge L^{-1}-\negl(n)
\end{align*}
where the final inequality follows since $\|\ket{\psi_{\rm mid}}\|\ge 1-\negl(n)$ by \Cref{cla:mid_to_ideal,cla:ideal_norm}.

We analyze the absolute value of the second term as follows.
\begin{align*}
&\left|
    \sum_{\vy \in Y}
    \left(\sum_{\vx'\in\mathcal{L}_\vy}\Delta_{\vy}(\vx')
    \overline{a_{\vy,\vx'}}\right)
    \left(\sum_{\vx\in \mathcal{L}_\vy}a_{\vy,\vx}\right)
\right| \\
&\le
\sum_{\vy \in Y}
    \left(\sum_{\vx'\in\mathcal{L}_\vy}|\Delta_{\vy}(\vx')|
    |a_{\vy,\vx'}|\right)
    \left|\sum_{\vx\in \mathcal{L}_\vy}a_{\vy,\vx}\right|\\
&\le
\sum_{\vy \in Y}    \left(\sum_{\vx'\in\mathcal{L}_\vy}|\Delta_{\vy}(\vx')|^2\right)^{1/2}\left(\sum_{\vx'\in\mathcal{L}_\vy}|a_{\vy,\vx'}|^2\right)^{1/2}
    \left|\sum_{\vx\in \mathcal{L}_\vy}a_{\vy,\vx}\right|\\
&\le
2^{-n}\sqrt{L}
\sum_{\vy \in Y}
    \left(\sum_{\vx'\in\mathcal{L}_\vy}|a_{\vy,\vx'}|^2\right)^{1/2}
    \left|\sum_{\vx\in \mathcal{L}_\vy}a_{\vy,\vx}\right|\\
&\le
2^{-n}\sqrt{L}
\left(
\sum_{\vy \in Y}
\sum_{\vx'\in\mathcal{L}_\vy}|a_{\vy,\vx'}|^2
\right)^{1/2}
\left(
\sum_{\vy \in Y}
\left|\sum_{\vx\in \mathcal{L}_\vy}a_{\vy,\vx}\right|^2
\right)^{1/2}\\
&=2^{-n}\sqrt{L}\left\|\ket{\psi'_{4,\rand}}\right\|\left\|\ket{\psi_{\rm mid}}\right\|\\
&\le \negl(n)
\end{align*}
where the first inequality follows from the triangle inequality, 
the second inequality follows from Cauchy-Schwarz inequality, 
the third inequality follows from $|\Delta_{\vy}(\vx')|\le 2^{-n}$ and $|\mathcal{L}_\vy|\le L$, 
the fourth inequality follows from Cauchy-Schwarz inequality, and 
the final inequality follows from $L=\poly(n)$, $\left\|\ket{\psi'_{4,\rand}}\right\|\le 1$ by definition,  
and 
$\left\|\ket{\psi_{\rm mid}}\right\|\le 1+\negl(n)$ by \Cref{cla:mid_to_ideal,cla:ideal_norm}. 

Combining above, 
\Cref{cla:fourprime_to_mid} follows.
\end{proof}

For unnormalized states, 
we put tildes on states to mean the normalized variants, e.g.,  $\ket{\tilde{\psi}'_{4,\rand}}=\frac{\ket{\psi'_{4,\rand}}}{\|\ket{\psi'_{4,\rand}}\|}$, 
$\ket{\tilde{\psi}_{\rm mid}}=\frac{\ket{\psi_{\rm mid}}}{\|\ket{\psi_{\rm mid}}\|}$, and 
$\ket{\tilde{\psi}_{\rm ideal}}=\frac{\ket{\psi_{\rm ideal}}}{\|\ket{\psi_{\rm ideal}}\|}$.  
\begin{claim}\label{cla:fourprime_to_ideal_trace}
We have
\[
\TD\left(\Ex_{\rand\gets[2^n]}
\ketbra{\psi_{4,\rand}}{\psi_{4,\rand}},
\ketbra{\tilde{\psi}_{\rm ideal}}{\tilde{\psi}_{\rm ideal}}\right)
\le
\sqrt{1-L^{-2}}+\negl(n).
\]
\end{claim}

\begin{proof}
Let
\[
\rho=
\Ex_{\rand\gets[2^n]}
\ketbra{\psi_{4,\rand}}{\psi_{4,\rand}}, ~~~
\rho'=
\Ex_{\rand\gets[2^n]}
\ketbra{\tilde{\psi}'_{4,\rand}}{\tilde{\psi}'_{4,\rand}},
~~~
\sigma_{\rm mid}=
\ketbra{\tilde{\psi}_{\rm mid}}{\tilde{\psi}_{\rm mid}},
~~~
\sigma_{\rm ideal}=
\ketbra{\tilde{\psi}_{\rm ideal}}{\tilde{\psi}_{\rm ideal}}.
\]

First, by \Cref{lem:2norm_to_tracenorm,cla:four_to_four_prime} and 
$\|\psi_{4,\rand}\|=1$, we have 
\[
\TD(\rho,\rho')\le \negl(n).
\]

Next, we show that \(\rho'\) is close to \(\sigma_{\rm mid}\). 

By \Cref{cla:four_to_four_prime}, we have $\left|\left\|\ket{\psi'_{4,\rand}}\right\|-1\right|\le \negl(n)$ for any $\rand$, and 
by \Cref{cla:mid_to_ideal,cla:ideal_norm}, 
we have $\left|\left\|\ket{\psi_{\rm mid}}\right\|-1\right|\le \negl(n)$.
Thus, \Cref{cla:fourprime_to_mid} implies 
\[
\left|
\Ex_{\rand\gets[2^n]}
\braket{\tilde{\psi}'_{4,\rand}}{\tilde{\psi}_{\rm mid}}
\right|\ge
L^{-1}-\negl(n).
\]

Since \(\sigma_{\rm mid}\) is pure, we have
\begin{align*}
F(\rho',\sigma_{\rm mid})^2
&=
\bra{\tilde{\psi}_{\rm mid}}\rho'\ket{\tilde{\psi}_{\rm mid}}\\
&=
\Ex_{\rand\gets[2^n]}
\left|
\braket{\tilde{\psi}'_{4,\rand}}{\tilde{\psi}_{\rm mid}}
\right|^2\\
&\ge
\left|
\Ex_{\rand\gets[2^n]}
\braket{\tilde{\psi}'_{4,\rand}}{\tilde{\psi}_{\rm mid}}
\right|^2\\
&\ge
\left(L^{-1}-\negl(n)\right)^2.
\end{align*}
Therefore, by the Fuchs--van de Graaf inequality (\Cref{eq:Fuchs–van_de_Graaf}), 
\begin{align*}
\TD(\rho',\sigma_{\rm mid})
&\le
\sqrt{1-F(\rho',\sigma_{\rm mid})^2}\\
&\le
\sqrt{1-\left(L^{-1}-\negl(n)\right)^2}\\
&\le
\sqrt{1-L^{-2}}+\negl(n).
\end{align*}

Next, we show closeness between $\sigma_{\rm mid}$ and $\sigma_{\rm ideal}$.
By \Cref{lem:2norm_to_tracenorm,cla:mid_to_ideal,cla:ideal_norm},  we have 
\[
\TD(\sigma_{\rm mid},\sigma_{\rm ideal})
\le
\negl(n).
\]
By the triangle inequality,
\begin{align*}
\TD(\rho,\sigma_{\rm ideal})
&\le
\TD(\rho,\rho')+
\TD(\rho',\sigma_{\rm mid})
+
\TD(\sigma_{\rm mid},\sigma_{\rm ideal})\\
&\le
\sqrt{1-L^{-2}}+\negl(n).
\end{align*} 
This completes the proof of \Cref{cla:fourprime_to_ideal_trace}.
\end{proof}

We are now ready to complete the proof of \Cref{thm:Regev_like_new}. 
Suppose that the state after Step 4 is replaced with $\ket{\psi_{\rm ideal}}$. 
As already seen in the proof of \Cref{cla:ideal_norm}, we have 
  \begin{align*}
  (I\times \QFT_{\FF_q^n}^{-1})\ket{\psi_{\rm ideal}}
 =q^{n/2}\sum_{\vy\in \FF_q^n}V(\vy)W(\vy)\ket{0}\ket{\vy}.
  \end{align*}
Noting that $V$ is the normalized indicator function of $C$, the above state is proportional to the target state 
\begin{align*}
\ket{\psi_{\rm target}}=\sum_{\vx\in C}W(\vx)\ket{\vx}.
\end{align*}
By \Cref{cla:fourprime_to_ideal_trace}, the actual algorithm's output is also within trace distance at most $$\sqrt{1-L^{-2}}+\negl(n)\le 1-\frac{1}{2L^2}+\negl(n)$$ of $\ket{\tilde{\psi}_{\rm target}}$. Since $L\le \poly(n)$, 
this completes the proof of \Cref{thm:Regev_like_new}. 
\end{proof}

\subsection{Proof of \Cref{cor:Regev_like_new}}
\begin{proof}[Proof of \Cref{cor:Regev_like_new}.]
We first show that the conditions for \Cref{thm:Regev_like_new} follow from those for \Cref{cor:Regev_like_new}. 

For every $i\in[n]$, we have 
\begin{align*}
|\hat{W}_i(0)|^2
&=\left|q^{-1/2}\sum_{x\in \FF_q}W_i(x)\right|^2\\
&=q^{-1}\left(|S_i|\cdot \sqrt{\frac{\tau}{|S_i|}}+(q-|S_i|)\cdot  \sqrt{\frac{1-\tau}{q-|S_i|}}\right)^2\\
&=\left(\sqrt{\tau\rho}+\sqrt{(1-\tau)(1-\rho)}\right)^2.
\end{align*}
We also note that 
\[
\sum_{e_i\in \FF_q}|\hat{W}_i(e_i)|^2=\sum_{e_i\in \FF_q}|W_i(e_i)|^2=1
\]
by Parseval's identity (\Cref{lem:Parseval})
and 
\[
\hat{W}(\ve)=\prod_{i\in [n]}\hat{W}_i(e_i)
\]
by \Cref{lem:QFT_prod}.

Therefore, 
the LHS of \Cref{eq:cond_one}, 
$\sum_{\ve\in T_{rn}}|\hat W(\ve)|^2$, can be interpreted as the probability that, when flipping $n$ independent
coins, each of which lands tails with probability 
$$1-(\sqrt{\tau\rho}+\sqrt{(1-\tau)(1-\rho)})^2,$$ we observe more than $rn$ tails.

By 
\Cref{eq:cond_one_cor} and
the Chernoff bound (\Cref{lem:Chernoff}), 
this probability is exponentially small in $n$. 
Thus, \Cref{eq:cond_one} holds.

Moreover, \Cref{eq:cond_two,eq:cond_three} are identical to \Cref{eq:cond_two_cor,eq:cond_three_cor}.

Thus, by \Cref{thm:Regev_like_new}, there is a quantum polynomial-time algorithm that, given $
\sum_{\ve\in \FF_q^n}W(\ve)\ket{\ve}
$, generates a state $\rho_{\rm out}$ that is within trace distance $1-1/\poly(n)$ of
$\ket{\tilde{\psi}_{\rm target}}=\frac{\ket{\psi_{\rm target}}}{\|\ket{\psi_{\rm target}}\|}$ where 
$\ket{\psi_{\rm target}}=
\sum_{\vx\in C}W(\vx)\ket{\vx}$. 

Here, we observe that $
\sum_{\ve\in \FF_q^n}W(\ve)\ket{\ve}
$ can be generated in quantum polynomial-time given $q,\rho,\tau$, and quantum access to the membership oracles for the sets $S_i$ except with negligible failure
probability.
First, we remark that 
\[
\sum_{\ve\in\FF_q^n} W(\ve)\ket{\ve}
=
\bigotimes_{i=1}^n
\left(
\sum_{e_i\in\FF_q} W_i(e_i)\ket{e_i}
\right),
\]
so it suffices to generate 
\[
\sum_{e_i\in\FF_q} W_i(e_i)\ket{e_i}=\sum_{e_i\in S_i}\sqrt{\frac{\tau}{|S_i|}}\ket{e_i}+\sum_{e_i\notin S_i}\sqrt{\frac{1-\tau}{q-|S_i|}}\ket{e_i}
\]
for each $i\in [n]$.

This can be done as follows. 
First, prepare the uniform superposition over $\FF_q$ along with an additional ancillary qubit initialized to $\ket{0}$:
\[
\frac{1}{\sqrt{q}}\sum_{e_i\in \FF_q}\ket{e_i}\ket{0}.
\]
Then controlled by the first register's value $e_i$, apply the following unitary on the second qubit.
\begin{align*}
  \ket{0}\mapsto \begin{cases}
      \ket{1}~~~& e_i\in S_i\\
      \alpha \ket{1}+\sqrt{1-\alpha^2}\ket{0}~~~& e_i\notin S_i
  \end{cases} 
\end{align*}
where 
\[
\alpha
=
\sqrt{\frac{\rho}{\tau}}\sqrt{\frac{1-\tau}{1-\rho}}.
\]
Note that $\alpha\le 1$ since we assume $\tau\ge \rho=|S_i|/q$. 
Also note that this unitary can be implemented by making a quantum query to the membership oracle of $S_i$.  

This results in the state 
\[
\frac{1}{\sqrt{q}}\left(
\sum_{e_i\in S_i}\ket{e_i}\ket{1}+\sum_{e_i\notin S_i}\ket{e_i}(\alpha \ket{1}+\sqrt{1-\alpha^2}\ket{0})
\right).
\]
Then measure the second register. 
The measurement outcome is $1$ with probability
\[
\frac{|S_i|+\alpha^2(q-|S_i|)}{q}=\frac{\rho}{\tau}\ge \rho.
\]
If the measurement outcome is $0$, repeat the same procedure until one obtains $1$. Since $\rho\ge 1/\poly(n)$, this can be done with $\poly(n)$ times repetitions except with negligible probability. 

Suppose that the measurement outcome is $1$. 
Then, the first register collapses to
\[
\sqrt{\frac{\tau}{q\rho}}\left(
\sum_{e_i\in S_i}\ket{e_i}+\sum_{e_i\notin S_i}\alpha\ket{e_i}
\right)=\sum_{e_i\in S_i}\sqrt{\frac{\tau}{|S_i|}}\ket{e_i}+\sum_{e_i\notin S_i}\sqrt{\frac{1-\tau}{q-|S_i|}}\ket{e_i}
\]
as desired. 

Thus, $\rho_{\rm out}$ can be generated in quantum polynomial-time given $q,\rho,\tau$, and quantum access to the membership oracles for the sets $S_i$ except with negligible failure
probability.

Suppose that we measure $\ket{\tilde{\psi}_{\rm target}}$ in the computational basis. 
The probability that the measurement outcome is not a solution to $\MaxLINSAT(q,\mA,\{S_i\}_{i\in[n]},s)$ is
\[
\frac{\sum_{\vx\in C_{sn}}|W(\vx)|^2}{\sum_{\vx\in C}|W(\vx)|^2}.
\]
This is negligible in $n$ by \Cref{eq:cond_three_cor,eq:cond_four_cor}. 
Thus, if we measure $\ket{\tilde{\psi}_{\rm target}}$ in the computational basis, we obtain a solution to $\MaxLINSAT(q,\mA,\{S_i\}_{i\in[n]},s)$ with probability $1-\negl(n)$. 
Since $\rho_{\rm out}$  is within trace distance $1-1/\poly(n)$ of $\ket{\tilde{\psi}_{\rm target}}$, if we measure $\rho_{\rm out}$ in the computational basis, we obtain a solution to $\MaxLINSAT(q,\mA,\{S_i\}_{i\in[n]},s)$ with probability $1/\poly(n)$. 
\end{proof}
\begin{remark}
In the proof above, the quantum algorithm appears to require the value of $\rho$, even though $\rho$ is not explicitly given as part of the MaxLINSAT instance. This is not an issue: using oracle access to the sets $S_i$, one can approximate $\rho$ to arbitrarily good inverse-polynomial precision. Such an approximation is sufficient to prepare a state sufficiently close to $\rho_{\rm out}$ and hence to solve the MaxLINSAT problem.
\end{remark}

\section{Worst-Case Quantum Algorithm for Optimal Polynomial Intersection}\label{sec:main}

In this section, we show the following theorem. 
\begin{theorem}\label{thm:main}
Let $q=p^e$ be a prime power. 
Let $n,k$ be positive integers such that 
$n/2\le k\le n\le q$ and let $R=k/n$. 
Let $\rho,\tau\in[0,1]$ be reals with $0<\rho<\tau\le 1$.
The parameters may depend on $n$ as long as they satisfy 
$p=\omega(1)$, $\log q\le \poly(n)$, 
$1-R=\Omega(1)$, 
$\rho=\Omega(1)$, $1-\rho=\Omega(1)$,  
and $\tau-\rho=\Omega(1)$. 

Suppose that the following hold.
\begin{align}
     R\ge 1-\left(\sqrt{\tau \rho}+\sqrt{(1-\tau)(1-\rho)}\right)^4+\Omega(1) \label{eq:theorem_cond_one}
\end{align}
\begin{align}
-(H(R)+R \log M_{\tau,\rho})= \Omega(1) \label{eq:theorem_cond_two}
\end{align}
where 
$$H(R)=-R\log R-(1-R)\log (1-R)$$
is the binary entropy function, and 
\[
M_{\tau,\rho}=\lambda_{\tau,\rho}^Q\left(\frac{\tau-\rho}{\rho(1-\rho)}\right)^Q N_{\rho,Q},
\]
for 
\[
Q=1/(1-R), 
\]
\begin{align*}
N_{\rho,Q}
=
\begin{cases}
    \left(\rho(1-\rho)\right)^{\frac{4-Q}{2}}\left(\rho_0^3\left(\frac{2}{3}-\rho_0\right)\right)^{\frac{Q-2}{2}} 
    ~~~&Q\le 4\\
    \left(\frac{\sin(\rho \pi)}{\pi}\right)^{Q-4}\rho_0^3\left(\frac{2}{3}-\rho_0\right) &Q\ge 4,
\end{cases} 
\end{align*}
\[
\rho_0=\min\{\rho,1-\rho\}, 
\]
and 
\[
\lambda_{\tau,\rho}=\max\left\{1,\frac{4\left(\sqrt{\tau\rho}+\sqrt{(1-\tau)(1-\rho)}\right)}{\sqrt{\tau \rho^{-1}}+\sqrt{(1-\tau)(1-\rho)^{-1}}}-1\right\}.
\]

Then there exists a quantum polynomial-time algorithm that solves $\OPI(k,q,\valpha,\{S_i\}_{i\in[n]},s)$ for 
any $s=\tau-\Omega(1)$, 
any distinct $\alpha_1,\alpha_2,...,\alpha_n\in \FF_q$ and any $S_i\subseteq \FF_q$ such that $|S_i|=\rho q$ for all $i\in [n]$.

Moreover, if the conditions hold for $\tau=1$, then the algorithm works even when $s=1$. 
\end{theorem}
\begin{remark}[The case of $\rho=1/2$]\label{rem:case_of_rho_half}
We explain how \Cref{thm:main_intro} is derived from \Cref{thm:main}. 
First, when $\rho=1/2$ and $\tau=1$, \Cref{eq:theorem_cond_one} is equivalent to
\[
R\ge 0.75+\Omega(1).
\]
Moreover, for $R>0.75$ and $\tau=1$, one can check that \Cref{eq:theorem_cond_two} is always satisfied. 
Thus, there is quantum polynomial-time algorithm that solves OPI with satisfaction rate $s=1$ in the worst-case whenever $R>0.75$. This gives the first part of \Cref{thm:main_intro}.

For the second part, assume $R\le 0.75$, in which case $Q\le 4$. 
When $Q\le 4$ and $\rho=1/2$, one can check that 
\[\lambda_{\tau,\rho}=1\]
and 
\[
N_{\rho,Q}=\left(\frac{1}{4}\right)^{\frac{4-Q}{2}}\left(\frac{1}{48}\right)^{\frac{Q-2}{2}}.
\]
Substituting them into the definition of $M_{\tau,\rho}$, we have 
\[
M_{\tau,\rho}=3\left(\frac{2\tau-1}{\sqrt{3}}\right)^Q. 
\]

Thus, 
\Cref{eq:theorem_cond_one,eq:theorem_cond_two} reduce to 
\Cref{eq:theorem_cond_one_intro,eq:theorem_cond_two_intro} in \Cref{thm:main_intro}, with $\tau$ in place of $s$.
However, in order to obtain an algorithm for satisfaction rate $s$, it suffices to satisfy the conditions for some $\tau\ge s+\Omega(1)$.
Since the conditions in \Cref{thm:main_intro} are assumed to hold for $s$ with a constant margin, by continuity we can choose $\tau\ge s+\Omega(1)$ sufficiently close to $s$ so that the same conditions also hold for $\tau$.

By the above argument  \Cref{thm:main_intro} follows. 
\end{remark}

In  the next subsection, we prove \Cref{thm:main} using \Cref{cor:Regev_like_new}. 
\subsection{Proof of \Cref{thm:main}}\label{sec:proof_main}
In this subsection, we prove \Cref{thm:main}. 

Note that 
$\OPI(k,q,\valpha,\{S_i\}_{i\in[n]},s)$  is equivalent to 
$\MaxLINSAT(q,\mA(\RS_{\FF_q,\valpha,k-1}),\{S_i\}_{i\in[n]},s)$ where $\mA(\RS_{\FF_q,\valpha,k-1})$ is the generator matrix of the Reed-Solomon code $\RS_{\FF_q,\valpha,k-1}$.  
Thus, it suffices to show that the conditions of \Cref{cor:Regev_like_new} are satisfied for $C=\RS_{\FF_q,\valpha,k-1}$. 

When $C=\RS_{\FF_q,\valpha,k-1}$, its dual is a generalized Reed-Solomon code of dimension $n-k$, and thus 
$(r,L)$-list-decodable for 
$r=1-\sqrt{1-R}-\Omega(1)$ and $L\le \poly(n)$. 
Thus, \Cref{eq:theorem_cond_one} in \Cref{thm:main} implies \Cref{eq:cond_one_cor} in \Cref{cor:Regev_like_new}. 

It remains to prove \Cref{eq:cond_two_cor,eq:cond_three_cor,eq:cond_four_cor} in \Cref{cor:Regev_like_new} also hold. 

Before proceeding with the proofs, we make an observation that simplifies the argument. 
Let $M_{s,\rho}$ be the same as $M_{\tau,\rho}$ except that $\tau$ is replaced with $s$. 
Then we can assume 
\begin{align}
-(H(R)+R \log M_{s,\rho})= \Omega(1) \label{eq:theorem_cond_two_s}
\end{align}
without loss of generality. This can be seen as follows. Since \Cref{eq:theorem_cond_two} holds, we can choose $s'=\tau-\Omega(1)$ sufficiently close to $\tau$ so that
\begin{align*}
-(H(R)+R \log M_{s',\rho})= \Omega(1). 
\end{align*}
Moreover, if $s'> s$, solving $\OPI(k,q,\valpha,\{S_i\}_{i\in[n]},s)$ is easier than $\OPI(k,q,\valpha,\{S_i\}_{i\in[n]},s')$. Thus, we can replace $s$ with $s'$ so that  we can assume \Cref{eq:theorem_cond_two_s} without loss of generality. Looking ahead, this is used in the proof of \Cref{eq:cond_four_cor} in \Cref{sec:bound_four}.

Below, we use the notations from \Cref{cor:Regev_like_new} and prove \Cref{eq:cond_two_cor,eq:cond_three_cor,eq:cond_four_cor}. 
Below, let $W$ and $W_i$ be the functions defined in \Cref{cor:Regev_like_new}.

\subsubsection{Proving \Cref{eq:cond_two_cor}}\label{sec:bound_two}
Let  $\epsilon$ be the LHS of \Cref{eq:cond_two_cor}. 
That is, 
\begin{align*}\epsilon=
q^{-(n-k)}
\sum_{\vy\in \FF_q^n}
\left|
\sum_{\vx\in C^\perp}
\one_{T_{rn}}(\vy-\vx)\hat W(\vy-\vx)
\right|^2.
\end{align*}

\paragraph{Decomposition into shift correlations.}
We first rewrite $\epsilon$ as a sum of correlation terms indexed by shifts in the dual code $C^\perp$.
\begin{lemma}\label{lem:split_to_sum_of_gamma}
    It holds that 
\begin{align*}
    \epsilon=\sum_{\vu\in C^\perp}\Gamma(\vu), 
\end{align*}
where 
\[
\Gamma(\vu)=\sum_{\ve\in \FF_q^n}
\one_{T_{rn}}(\ve)\one_{T_{rn}}(\ve+\vu)\hat W(\ve)\overline{\hat W(\ve+\vu)}.
\]
\end{lemma}
\begin{proof}[Proof of \Cref{lem:split_to_sum_of_gamma}]
Define 
\begin{align*}
g(\ve)=\one_{T_{rn}}(\ve)\hat W(\ve).
\end{align*}
Then we have 
\begin{align*}
  \epsilon=  q^{-(n-k)}
\sum_{\vy\in \FF_q^n}
\sum_{\vx,\vx'\in C^\perp}
g(\vy-\vx)
\overline{g(\vy-\vx')}.
\end{align*}
Note that for any $\vx\in C^\perp$, when $\vx'$ ranges over $C^\perp$, $\vx'-\vx$ also ranges over $C^\perp$.
Thus, by changing variables as $\vu=\vx'-\vx$ and $\ve=\vy-\vx$, we have 
\begin{align*}
  \epsilon
  &=  q^{-(n-k)}
\sum_{\ve\in \FF_q^n}
\sum_{\vx,\vu\in C^\perp}
g(\ve)
\overline{g(\ve+\vu)}\\
&=\sum_{\ve\in \FF_q^n}
\sum_{\vu\in C^\perp}
g(\ve)
\overline{g(\ve+\vu)}.
\end{align*}
\Cref{lem:split_to_sum_of_gamma} directly follows from this.
\end{proof}

\begin{lemma}\label{lem:gamma_zero}
    It holds that
    \[
    \Gamma(0)\le \negl(n).
    \]
\end{lemma}
\begin{proof}[Proof of \Cref{lem:gamma_zero}]
We have 
\begin{align*}
\Gamma(0)
=\sum_{\ve\in T_{rn}}|\hat{W}(\ve)|^2.
\end{align*}
As seen in the proof of \Cref{cor:Regev_like_new}, \Cref{eq:cond_one_cor}, which follows from \Cref{eq:theorem_cond_one},  implies that this is $\negl(n)$. 
\end{proof}

By \Cref{lem:split_to_sum_of_gamma,lem:gamma_zero},  for proving that $\epsilon\le \negl(n)$, it suffices to show 
\begin{align}
    \sum_{\vu\in C^\perp \setminus \{0\}}\left|\Gamma(\vu)\right|\le \negl(n). \label{eq:sum_of_gamma_is_negl}
\end{align}

\paragraph{Bounding $\Gamma$ via generating polynomials.} 
Next, we bound $\Gamma(\vu)$ by introducing a bivariate generating polynomial whose coefficients keep track of the two weights $\wt(\ve)$ and $\wt(\ve+\vu)$.

For 
$i\in [n]$ and 
$b\in \FF_q$, define a polynomial
\begin{align*}
  K_i^{(b)}(X,Y)=\sum_{a\in \FF_q}\hat{W}_i(a)\overline{\hat{W}_i(a+b)}X^{\one_{a\neq 0}}Y^{\one_{a+b\ne 0}},   
\end{align*}
where $\one_{a\neq 0}$ is a variable that takes $1$ when $a\ne 0$ and otherwise takes $0$, and $\one_{a+b\ne 0}$ is defined similarly. 

Below, for a polynomial $F(X,Y)\in \mathbb{C}[X,Y]$, $\|F(X,Y)\|_1$ denotes its $\ell_1$ norm, that is,  the sum of the absolute values of all the coefficients. 

\begin{lemma}\label{lem:upper_bound_gamma_by_K}
It holds that 
\begin{align*}
   |\Gamma(\vu)|\le \prod_{i=1}^{n}\|K_i^{(u_i)}(X,Y)\|_1.
\end{align*} 
\end{lemma}
\begin{proof}[Proof of \Cref{lem:upper_bound_gamma_by_K}]
For $\vu=(u_1,u_2,...,u_n)\in \FF_q^n$, we have 
\begin{align}
   \Gamma(\vu)=\sum_{a,b> rn}[X^a Y^b]\prod_{i=1}^{n}K_i^{(u_i)}(X,Y). \label{eq:Gamma_to_K}
\end{align}
where $[X^aY^b]F(X,Y)$ denotes the coefficient of $X^aY^b$ in $F(X,Y)$.

Indeed, expanding the product, we get
\begin{align*}
\prod_{i=1}^{n}K_i^{(u_i)}(X,Y)
&=
\prod_{i=1}^n
\sum_{e_i\in \FF_q}
\hat{W}_i(e_i)\overline{\hat{W}_i(e_i+u_i)}
X^{\one_{e_i\ne 0}}
Y^{\one_{e_i+u_i\ne 0}}\\
&=
\sum_{\ve\in \FF_q^n}
\prod_{i=1}^n
\hat{W}_i(e_i)\overline{\hat{W}_i(e_i+u_i)}
X^{\one_{e_i\ne 0}}
Y^{\one_{e_i+u_i\ne 0}}\\
&=\sum_{\ve\in \FF_q^n}
\hat{W}(\ve)\overline{\hat{W}(\ve+\vu)}
X^{\wt(\ve)}
Y^{\wt(\ve+\vu)}.
\end{align*}
\Cref{eq:Gamma_to_K} immediately follows from the above and the definition of $\Gamma$.  

\Cref{lem:upper_bound_gamma_by_K} immediately follows from \Cref{eq:Gamma_to_K}. 
\end{proof}

\begin{lemma}\label{lem:K_zero}
It holds that 
$$\|K_i^{(0)}(X,Y)\|_1=1.$$
\end{lemma}
\begin{proof}[Proof of \Cref{lem:K_zero}]
By definition, 
\begin{align*}
  K_i^{(0)}(X,Y)=\sum_{a\in \FF_q}|\hat{W}_i(a)|^2(XY)^{\one_{a\neq 0}}.   
\end{align*}
As seen in the proof of \Cref{cor:Regev_like_new}, we have
\begin{align*}
|\hat{W}_i(0)|^2
=\left(\sqrt{\tau\rho}+\sqrt{(1-\tau)(1-\rho)}\right)^2.
\end{align*}
We also have
\[
\sum_{a\in \FF_q}|\hat{W}_i(a)|^2=1
\]
by \Cref{lem:Parseval}.

Thus, we have 
\begin{align*}
K_i^{(0)}(X,Y)=\left(\sqrt{\tau\rho}+\sqrt{(1-\tau)(1-\rho)}\right)^2+\left(1-\left(\sqrt{\tau\rho}+\sqrt{(1-\tau)(1-\rho)}\right)^2\right)XY.    
\end{align*}

Noting that $0\le \left(\sqrt{\tau\rho}+\sqrt{(1-\tau)(1-\rho)}\right)^2\le 1$, 
\Cref{lem:K_zero} follows. 
\end{proof}

We define $\mu_i:\FF_q\rightarrow \mathbb{C}$ and $\nu_i:\FF_q\rightarrow \mathbb{C}$ as follows. 

\begin{align}
\mu_i(b)=\sum_{x\in \FF_q}|W_i(x)|^2\omega_p^{-\Tr(bx)}, \label{eq:definition_mu}
\end{align}
\begin{align}
    \nu_i(b)=
   \begin{cases} 
\mu_i(b) & \text{if } \kappa_{\tau,\rho}\le 1/2,\\[1.2ex] (4\kappa_{\tau,\rho}-1)\mu_i(b)& \text{if } \kappa_{\tau,\rho}>1/2 \end{cases}, \label{eq:definition_nu}
\end{align}
where 
\[
\kappa_{\tau,\rho}=
\frac{\sqrt{\tau\rho}+\sqrt{(1-\tau)(1-\rho)}}{\sqrt{\tau \rho^{-1}}+\sqrt{(1-\tau)(1-\rho)^{-1}}}.
\]

\begin{lemma}\label{lem:K_non_zero}
For $b\ne 0$, it holds that 
    \begin{align*}
    \|K_i^{(b)}(X,Y)\|_1= |\nu_i(b)|.
\end{align*}
\end{lemma}
\begin{proof}[Proof of \Cref{lem:K_non_zero}]
For $b\ne 0$, 
\begin{align*}
\sum_{a\in \FF_q}\hat{W}_i(a)\overline{\hat{W}_i(a+b)}
&=\sum_{a\in \FF_q}\left(\frac{1}{\sqrt{q}}\sum_{x \in \FF_q }W_i(x)\omega_p^{\Tr(ax)}\right)
\left(\frac{1}{\sqrt{q}}\sum_{y \in \FF_q }\overline{W_i(y)\omega_p^{\Tr((a+b)y)}}\right)\\
&=\frac{1}{q}\sum_{x,y\in \FF_q}W_i(x)W_i(y)\omega_p^{-\Tr(by)}\left(\sum_{a\in \FF_q}\omega_p^{\Tr(a(x-y))}\right)\\
&=\sum_{x\in \FF_q}|W_i(x)|^2\omega_p^{-\Tr(bx)}\\
&=\mu_i(b)
\end{align*}
where the third equality follows from \Cref{eq:sum_is_zero}. 
For $b\ne 0$, 
\begin{align}
\mu_i(b)=\left(\sqrt{\tau \rho^{-1}}+\sqrt{(1-\tau)(1-\rho)^{-1}}\right)\hat{W}_i(-b). \label{eq:mu_and_hat_W}
\end{align}
Indeed, this can be seen from the following two equations, where we use \Cref{eq:sum_is_zero} in the third and seventh equations:
\begin{align*}
 \mu_i(b)&= \sum_{x\in \FF_q}|W_i(x)|^2\omega_p^{-\Tr(bx)} \\
 &=\sum_{x\in S_i}\frac{\tau}{|S_i|}\omega_p^{-\Tr(bx)}+ \sum_{x\notin S_i}\frac{1-\tau}{q-|S_i|}\omega_p^{-\Tr(bx)}\\
 &=\left(\frac{\tau}{|S_i|}-\frac{1-\tau}{q-|S_i|}\right)\sum_{x\in S_i}\omega_p^{-\Tr(bx)}\\
  &=\frac{1}{q}\left(\tau\rho^{-1}-(1-\tau)(1-\rho)^{-1}\right)\sum_{x\in S_i}\omega_p^{-\Tr(bx)}\\
  &=\frac{1}{q}\left(\sqrt{\tau \rho^{-1}}+\sqrt{(1-\tau)(1-\rho)^{-1}}\right)
  \left(\sqrt{\tau \rho^{-1}}-\sqrt{(1-\tau)(1-\rho)^{-1}}\right)
  \sum_{x\in S_i}\omega_p^{-\Tr(bx)}\\
  &=\frac{1}{\sqrt{q}}\left(\sqrt{\tau \rho^{-1}}+\sqrt{(1-\tau)(1-\rho)^{-1}}\right)\left(\sqrt{\frac{\tau}{|S_i|}}-\sqrt{\frac{1-\tau}{q-|S_i|}}\right)\sum_{x\in S_i}\omega_p^{-\Tr(bx)}\\
  &=\left(\sqrt{\tau \rho^{-1}}+\sqrt{(1-\tau)(1-\rho)^{-1}}\right)\cdot \frac{1}{\sqrt{q}}\left(\sum_{x\in \FF_q}W_i(x)\omega_p^{-\Tr(bx)}\right)\\
  &=\left(\sqrt{\tau \rho^{-1}}+\sqrt{(1-\tau)(1-\rho)^{-1}}\right)\hat{W}_i(-b). 
\end{align*}

The constant term of $K_i^{(b)}(X,Y)$ is $0$ since $a=0$ and $a+b=0$ cannot hold simultaneously when $b\ne 0$. 

The coefficient of $Y$ of $K_i^{(b)}(X,Y)$ is 
\[
\hat{W}_i(0)\overline{\hat{W}_i(b)}=\hat{W}_i(0)\hat{W}_i(-b)=\kappa_{\tau,\rho}\mu_i(b)
\]
since we have
$\hat{W}_i(0)=\sqrt{\tau\rho}+\sqrt{(1-\tau)(1-\rho)}$ as seen in the proof of \Cref{cor:Regev_like_new} and
\Cref{eq:mu_and_hat_W}.  

Similarly, 
the coefficient of $X$ is also
\[
\hat{W}_i(-b)\overline{\hat{W}_i(0)}=\kappa_{\tau,\rho}\mu_i(b)
\]

Thus, for $b\ne 0$, 
\begin{align*}
 K_i^{(b)}(X,Y)=\mu_i(b)(\kappa_{\tau,\rho} X+\kappa_{\tau,\rho} Y+(1-2\kappa_{\tau,\rho})XY).   
\end{align*}
Since $\kappa_{\tau,\rho}\ge 0$, \Cref{lem:K_non_zero} follows. 
\end{proof}

By \Cref{lem:upper_bound_gamma_by_K,lem:K_zero,lem:K_non_zero}, we have 
\[
|\Gamma(\vu)|\le 
\prod_{i:u_i\ne 0} |\nu_i(u_i)|. 
\]
Thus, for proving \Cref{eq:sum_of_gamma_is_negl}, and thus for completing the proof of \Cref{eq:cond_two_cor}, we are left to prove the following lemma.
\begin{lemma}\label{lem:Z_is_negl}
It holds that 
\begin{align*}
   \sum_{\vu\in C^\perp \setminus \{0\}}\prod_{i:u_i\ne 0} |\nu_i(u_i)|\le \negl(n).
\end{align*}
\end{lemma}
We note that \Cref{lem:Z_is_negl} is reused in the proof of \Cref{eq:cond_three_cor} in \Cref{sec:bound_three}.
We prove this below. 

\paragraph{Proving \Cref{lem:Z_is_negl} via MDS Brascamp–Lieb.}
Let $Z=\sum_{\vu\in C^\perp \setminus \{0\}}\prod_{i:u_i\ne 0} |\nu_i(u_i)|$.

For $i\in [n]$,
and $t\ge 1$, 
define $h_i^{(t)}:\FF_q\rightarrow \mathbb{R}_{\ge 0}$ as 
\begin{align*}
    h_i^{(t)}(a)=
   \begin{cases} 
1 & \text{if } a=0,\\[1.2ex] t|\nu_i(a)|& \text{if } a\ne 0\end{cases}.
\end{align*}

Since $C^\perp$ is an MDS code of dimension $n-k$, for any $\vu\in C^\perp \setminus \{0\}$, 
$\wt(\vu)\ge k+1$, and thus we have 
\[
 \prod_{i:u_i\ne 0} |\nu_i(u_i)|
 \le 
t^{-(k+1)}\prod_{i\in [n]}h_i^{(t)}(u_i)
\le t^{-k}\prod_{i\in [n]}h_i^{(t)}(u_i).
\]
Thus,\footnote{We include the term corresponding to $\vu=0$, which can only increase the sum.} 
\[
Z\leq t^{-k}\sum_{\vu\in C^\perp}\prod_{i\in [n]}h_i^{(t)}(u_i).
\]

Here, we apply the MDS Brascamp–Lieb inequality (\Cref{thm:MDS_BL}), noting that the rate of $C^\perp$ is $1-R$. 
Letting $Q=1/(1-R)$ and
\[
M_i=\sum_{a\in \FF_q\setminus \{0\}}|\nu_i(a)|^Q, 
\]
we have 
\begin{align}
Z\leq t^{-k}\prod_{i\in [n]}\left(1+t^{Q}M_i\right)^{1/Q}. \label{eq:upper_bound_Z}
\end{align} 

In the following lemma, we give an upper bound for $M_i$.
\begin{lemma}\label{lem:upper_bound_M}
    For $i\in [n]$, it holds that 
    \[
    M_i\le M_{\tau,\rho}+o(1),
    \]
    where  $M_{\tau,\rho}$ is as in \Cref{thm:main}. 
    That is, 
\[
M_{\tau,\rho}=\lambda_{\tau,\rho}^Q\left(\frac{\tau-\rho}{\rho(1-\rho)}\right)^Q N_{\rho,Q},
\]
for 
\begin{align*}
N_{\rho,Q}
=
\begin{cases}
    \left(\rho(1-\rho)\right)^{\frac{4-Q}{2}}\left(\rho_0^3\left(\frac{2}{3}-\rho_0\right)\right)^{\frac{Q-2}{2}} 
    ~~~&Q\le 4\\
    \left(\frac{\sin(\rho \pi)}{\pi}\right)^{Q-4}\rho_0^3\left(\frac{2}{3}-\rho_0\right) &Q\ge 4,
\end{cases} 
\end{align*}
\[
\rho_0=\min\{\rho,1-\rho\}, 
\]
and 
\[
\lambda_{\tau,\rho}=\max\left\{1,4\kappa_{\tau,\rho}-1\right\}.
\]
\end{lemma}
\begin{proof}[Proof of \Cref{lem:upper_bound_M}]
Recall that $\mu_i$ is defined as 
\[
\mu_i(a)=\sum_{x\in \FF_q}|W_i(x)|^2\omega_p^{-\Tr(ax)}.
\]
By the calculation below \Cref{eq:mu_and_hat_W},\footnote{In particular, the fourth line in the calculation can be easily seen to be equal to $\frac{\tau-\rho}{\rho(1-\rho)\sqrt{q}}\overline{\hat{\one}_{S_i}(a)}$.} we have
\[
\mu_i(a)=\frac{\tau-\rho}{\rho(1-\rho)\sqrt{q}}\overline{\hat{\one}_{S_i}(a)}
\]

By \Cref{thm:Fourier_Qth_power_sum}, 
we have 
\begin{align}
    \sum_{a\in \FF_q\setminus \{0\}}|\mu_i(a)|^Q
    =\left(\frac{\tau-\rho}{\rho(1-\rho)}\right)^Q q^{-Q/2}\sum_{a\in \FF_q\setminus \{0\}}|\hat{1}_{S_i}(a)|^Q\le \left(\frac{\tau-\rho}{\rho(1-\rho)}\right)^Q N_{\rho,Q}+O(p^{-1}). \label{eq:bound_q}
\end{align}

Therefore,
\begin{align*}
M_i&=\sum_{a\in \FF_q\setminus \{0\}}|\nu_i(a)|^Q\\
&= \lambda_{\tau,\rho}^Q\sum_{a\in \FF_q\setminus \{0\}}|\mu_i(a)|^Q\\
&\le \lambda_{\tau,\rho}^Q\left(\frac{\tau-\rho}{\rho(1-\rho)}\right)^Q N_{\rho,Q}
  +o(1).
\end{align*} 
Note that the error term is $o(1)$ since we assume $p=\omega(1)$ and $\lambda_{\tau,\rho}^Q=O(1)$ since $\lambda_{\tau,\rho}\le 3$ and $Q=1/(1-R)=O(1)$. 

This completes the proof of \Cref{lem:upper_bound_M}. 
\end{proof}


Substituting \Cref{lem:upper_bound_M} into \Cref{eq:upper_bound_Z},
we have 
\[
Z\le t^{-k}\left(1+t^{Q}M_{\tau,\rho}+o(1)\right)^{n/Q}
\]
as long as $t=O(1)$. 


Taking logarithm, and dividing both sides by $n$, we have 
\[
\frac{1}{n}\log Z\le -R\log t+Q^{-1}\log(1+t^QM_{\tau,\rho}+o(1)).
\]
Set  
\[
t=\left(\frac{R}{M_{\tau,\rho}(1-R)}\right)^{1/Q}. 
\]
This indeed satisfies $t\ge 1$ since $M_{\tau,\rho}<1$ follows from \Cref{eq:theorem_cond_two} in \Cref{thm:main} and $R\ge 1/2$. 
We also note that $t=O(1)$ since
$1/Q\le 1$, 
$1-R=\Omega(1)$, and 
$M_{\tau,\rho}=\Omega(1)$, which follows from 
$1-R=\Omega(1)$, 
$\rho=\Omega(1)$, and
$\tau-\rho=\Omega(1)$. 

Then we have 
\begin{align*}
\frac{1}{n}\log Z
&\le \frac{1}{Q}\left(-R(\log R-\log(1-R)-\log M_{\tau,\rho})+\log (1+\frac{R}{1-R}+o(1))\right)\\
&
= (1-R)(H(R)+R \log M_{\tau,\rho})+o(1),
\end{align*}
where 
$$H(R)=-R\log R-(1-R)\log (1-R)$$
is the binary entropy function.\footnote{The error term can be put outside since $1/Q=O(1)$ and 
$\log(x+\epsilon)-\log x=O(\epsilon)$ for any $x\ge 1$. 
} 

By 
$1-R=\Omega(1)$ and 
\Cref{eq:theorem_cond_two} in \Cref{thm:main},  $H(R)+R \log M_{\tau,\rho}$ is bounded above by a negative constant. 

This implies $Z\le \negl(n)$, finishing the proof of   \Cref{lem:Z_is_negl}.

This completes the proof of
\Cref{eq:cond_two_cor}.

\subsubsection{Proving \Cref{eq:cond_three_cor}}\label{sec:bound_three}
Next we prove that \Cref{eq:cond_three_cor} holds.

We have 
\begin{align*}
q^{n-k}\sum_{\vx\in C}|W(\vx)|^2
&=q^{n-k}\sum_{\vx\in \FF_q^n}\one_{C}(\vx)\overline{|W(\vx)|^2}\\
&=q^{n-k}\sum_{\vx\in \FF_q^n}\hat{\one}_{C}(\vx)\overline{\widehat{|W|^2}(\vx)}\\
&=q^{n/2}\sum_{\vx\in \FF_q^n}\one_{C^\perp}(\vx)\overline{\widehat{|W|^2}(\vx)}
\end{align*}
where the second equality follows from Parseval identity (\Cref{lem:Parseval})
and the third equality follow from \Cref{lem:fourier_dual}.  
Here, $|W|^2$ means the function that takes $\vx\in \FF_q^n$ and outputs $|W(\vx)|^2$

Let $\mu_i$ and $\nu_i$ be as defined in 
\Cref{sec:bound_two} (\Cref{eq:definition_mu,eq:definition_nu}).

Then we have 
\begin{align*}
    \overline{\widehat{|W|^2}(\vx)}
    &=q^{-n/2}\sum_{\vz \in \FF_q^n} |W(\vz)|^2\omega_p^{-\Tr(\vx\cdot \vz)}\\
    &=q^{-n/2}\sum_{\vz \in \FF_q^n}\prod_{i\in [n]} |W_i(z_i)|^2\omega_p^{-\Tr(x_i\cdot z_i)}\\
    &=q^{-n/2}\prod_{i\in [n]}\sum_{z_i\in \FF_q^n} |W_i(z_i)|^2\omega_p^{-\Tr(x_i\cdot z_i)}\\
    &=q^{-n/2}\prod_{i\in [n]} \mu_i(x_i).
\end{align*}

Thus, we have 
\begin{align*}
q^{n-k}\sum_{\vx\in C}|W(\vx)|^2
&=\sum_{\vx \in C^\perp}\prod_{i\in [n]} \mu_i(x_i)\\
&= 1+\sum_{\vx \in C^\perp\setminus \{0\}}\prod_{i\in [n]} \mu_i(x_i).\\
\end{align*}

Since $|\mu_i(u_i)|\le|\nu_i(u_i)|$
for all $u_i\in \FF_q$ 
by definition, \Cref{lem:Z_is_negl} implies 
\begin{align*}
   \sum_{\vu\in C^\perp \setminus \{0\}}\prod_{i:u_i\ne 0} |\mu_i(u_i)|\le \negl(n). 
\end{align*}

Combining above,   
\Cref{eq:cond_three_cor} 
follows.

\subsubsection{Proving \Cref{eq:cond_four_cor}}\label{sec:bound_four}
If $\tau=s=1$, \Cref{eq:cond_four_cor} trivially holds since the LHS is $0$. 

Below, we focus on the case of $\tau<1$. 

Let $W'$ be the same as $W$ except that $\tau$ is replaced with $s$, that is, 
$W'(\ve)=\prod_{i\in [n]}W'_i(e_i)$
where 
$$ W'_i(e_i)= \begin{cases} 
\sqrt{\dfrac{s}{|S_i|}} & \text{if } e_i\in S_i,\\[1.2ex] \sqrt{\dfrac{1-s}{q-|S_i|}} & \text{if } e_i\notin S_i \end{cases}.$$


For any $\ve$, 
if we let $\ell=|i\in[n]: e_i\in S_i|$, then
we have 
\[
|W(\ve)|^2=\left(\frac{\tau}{s}\right)^\ell\left(\frac{1-\tau}{1-s}\right)^{n-\ell} |W'(\ve)|^2.
\]
Since we assume $\tau> s$, we have 
\begin{align*}
\left|
q^{n-k}\sum_{\vx\in C_{sn}}|W(\vx)|^2
\right|
&\le 
\left(\frac{\tau}{s}\right)^{s n}\left(\frac{1-\tau}{1-s}\right)^{(1-s)n}
\left|
q^{n-k}\sum_{\vx\in C_{sn}}|W'(\vx)|^2
\right|\\
&\le 
\left(\frac{\tau}{s}\right)^{s n}\left(\frac{1-\tau}{1-s}\right)^{(1-s)n}
\left|
q^{n-k}\sum_{\vx\in C}|W'(\vx)|^2
\right|.
\end{align*} 
\begin{lemma}\label{lem:upper_bound_KL}
If $\tau-s=\Omega(1)$, 
we have 
\[
\left(\frac{\tau}{s}\right)^{s}\left(\frac{1-\tau}{1-s}\right)^{(1-s)}\le 1-\Omega(1).
\]
\end{lemma}
\begin{proof}[Proof of \Cref{lem:upper_bound_KL}]
By assumption, there is a constant $\delta>0$ such that $\tau-s\ge \delta$ for sufficiently large $n$. 
Let $D=\{(s,\tau): 0\le s\le \tau\le 1,\tau-s\ge \delta\}$. 
We define $f:D\rightarrow \mathbb{R}$ by 
\begin{align*}
f(s,\tau)=
\begin{cases}
\left(\frac{\tau}{s}\right)^{s}\left(\frac{1-\tau}{1-s}\right)^{(1-s)} & \text{if }0<s<1,\\
1-\tau & \text{if }s=0
\end{cases}.
\end{align*}
Note that we never have $s=1$ for $(s,\tau)\in D$ since $s\le \tau-\delta\le 1-\delta <1$. 
It is easy to see that $f$ is continuous on $D$. 

By the weighted AM-GM inequality, we have 
\[
f(s,\tau)\le s\frac{\tau}{s}+(1-s)\frac{1-\tau}{1-s}=1
\]
and the equality holds if and only if $\frac{\tau}{s}=\frac{1-\tau}{1-s}$, that is, $\tau=s$. But we never have $\tau=s$ for $(s,\tau)\in D$, thus we have $f(s,\tau)<1$ for $(s,\tau)\in D$. 

Since $f$ is continuous on compact set $D$, $f$ has a maximum value on $D$. 
Since $f(s,\tau)<1$ for $(s,\tau)\in D$, 
there is $C<1$ such that $f(s,\tau)\le C$ for all $(s,\tau)\in D$. 

Since we have $(s,\tau)\in D$ for sufficiently large $n$, \Cref{lem:upper_bound_KL} holds.  
\end{proof}
By \Cref{lem:upper_bound_KL}, 
we have 
\[
\left(\frac{\tau}{s}\right)^{s n}\left(\frac{1-\tau}{1-s}\right)^{(1-s)n}\le \negl(n). 
\]

Since we assume \Cref{eq:theorem_cond_two_s}, by exactly the same proof as that of \Cref{eq:cond_three_cor}, where we replace $\tau$ with $s$, we have 
\[
\left|
q^{n-k}\sum_{\vx\in C}|W'(\vx)|^2
-1\right|\le \negl(n).
\]

Thus, we have 
\[
\left|
q^{n-k}\sum_{\vx\in C_{sn}}|W(\vx)|^2
\right|\le \negl(n). 
\]
This completes the proof of \Cref{eq:cond_four_cor}, and thus completes the whole proof of \Cref{thm:main}.

\section{Existential Bound for MDS MaxLINSAT}\label{sec:existential}
In this section, we show an existential bound for MaxLINSAT, which is a generalization of OPI.

\begin{theorem}\label{thm:existence}
Let $q=p^{e}$ be a prime power. 
Let $n,k$ be positive integers such that 
$n/2\le k\le n\le q$ and let $R=k/n$. 
Let $\rho,\tau\in[0,1]$ be reals with $0<\rho<\tau\le 1$.
The parameters may depend on $n$ as long as they satisfy 
$p=\omega(1)$, $\log q\le \poly(n)$, 
$1-R=\Omega(1)$, 
$\rho=\Omega(1)$, $1-\rho=\Omega(1)$,  
and $\tau-\rho=\Omega(1)$. 

Suppose that the following hold.
\begin{align}
-(H(R)+R \log M'_{\tau,\rho})= \Omega(1) \label{eq:existence_theorem_cond}
\end{align}
where 
$$H(R)=-R\log R-(1-R)\log (1-R)$$
is the binary entropy function, and 
\[
M'_{\tau,\rho}=\left(\frac{\tau-\rho}{\rho(1-\rho)}\right)^Q N_{\rho,Q},
\]
for 
\[
Q=1/(1-R), 
\]
\begin{align*}
N_{\rho,Q}
=
\begin{cases}
    \left(\rho(1-\rho)\right)^{\frac{4-Q}{2}}\left(\rho_0^3\left(\frac{2}{3}-\rho_0\right)\right)^{\frac{Q-2}{2}} 
    ~~~&Q\le 4\\
    \left(\frac{\sin(\rho \pi)}{\pi}\right)^{Q-4}\rho_0^3\left(\frac{2}{3}-\rho_0\right) &Q\ge 4,
\end{cases} 
\end{align*}
and 
\[
\rho_0=\min\{\rho,1-\rho\}. 
\]

Then for any matrix $\mA\in \FF_q^{n\times k}$ that generates an MDS code, 
for 
any $s=\tau-\Omega(1)$ 
and any $S_i\subseteq \FF_q$ such that $|S_i|=\rho q$ for all $i\in [n]$,  
$\MaxLINSAT(q,\mA,\{S_i\}_{i\in[n]},s)$ has a solution for sufficiently large $n$.

Moreover, if the condition holds for $\tau=1$, then the above holds even when $s=1$. 
\end{theorem}
\begin{remark}[Comparison with \Cref{thm:main}]\label{rem:comparison_alg_existence}
 The condition for the existence of a solution in \Cref{thm:existence} is similar to the condition for algorithmic success in \Cref{thm:main}, but it is weaker in the following two aspects. First, the first condition (\Cref{eq:theorem_cond_one}) of \Cref{thm:main} is not required in \Cref{thm:existence}. Second, the second condition (\Cref{eq:theorem_cond_two}) of \Cref{thm:main} is weakened by replacing $M_{\tau,\rho}$ with  $M'_{\tau,\rho}$, which is defined similarly to $M_{\tau,\rho}$ except that  $\lambda_{\tau,\rho}$ is always set to $1$. Thus, $M'_{\tau,\rho}\le M_{\tau,\rho}$, making the condition easier to satisfy. 

 We also note that \Cref{thm:existence} holds for any instance of MDS MaxLINSAT, whereas \Cref{thm:main} is stated only for the special case of OPI. 
\end{remark}

To prove the theorem, we state an obvious yet useful lemma inspired by \Cref{cor:Regev_like_new}. 
\begin{lemma}\label{lem:existence}
Let $q$ be a prime power and let $n,k$ be positive integers such that
$\log q\le \poly(n)$. Let $\mA\in \FF_q^{n\times k}$ be a matrix that generates an MDS code $C\subseteq \FF_q^n$.

For each $i\in[n]$, let $S_i\subseteq \FF_q$ be a subset satisfying $1\le |S_i|\le q-1$. 

Define $W:\FF_q^n\to\mathbb C$ by $$ W(e_1,\ldots,e_n)=\prod_{i=1}^n W_i(e_i), $$ where $$ W_i(e_i)= \begin{cases} 
\sqrt{\dfrac{\tau}{|S_i|}} & \text{if } e_i\in S_i,\\[1.2ex] \sqrt{\dfrac{1-\tau}{q-|S_i|}} & \text{if } e_i\notin S_i \end{cases} $$ and $0\le \tau\le 1$.

Assume that
\begin{align}
&
\sum_{\vx\in C}|W(\vx)|^2
\ne 
\sum_{\vx\in C_{sn}}|W(\vx)|^2, 
 \label{eq:existence_lemma_cond}
\end{align}
where 
$0\le s\le 1$ and 
$C_{sn}=\{\vx=(x_1,...,x_n) \in C: \left|\{i\in[n] : x_i \in S_i\}\right| < s n\}$.

Then $\MaxLINSAT(q,\mA,\{S_i\}_{i\in[n]},s)$ has a solution.
\end{lemma}
\begin{proof}
    If $\MaxLINSAT(q,\mA,\{S_i\}_{i\in[n]},s)$ has no solution, then $C\setminus C_{sn}=\emptyset$. 
    Then \Cref{eq:existence_lemma_cond} cannot hold. 
\end{proof}

Using \Cref{lem:existence}, we prove \Cref{thm:existence}.
\begin{proof}[Proof of \Cref{thm:existence}.]
Define $W$, $C$, and $C_{sn}$ as in \Cref{lem:existence}. 

Note that the definition of $W$ is the same as that in \Cref{cor:Regev_like_new}. 
Thus, by the calculation in \Cref{sec:bound_three}, we have 
\begin{align}
q^{n-k}\sum_{\vx\in C}|W(\vx)|^2
= 1+\sum_{\vx \in C^\perp\setminus \{0\}}\prod_{i\in [n]} \mu_i(x_i), \label{eq:existential_sum_W_square}
\end{align}
where $\mu_i$ is defined as in \Cref{eq:definition_mu}.

Moreover, by repeating the same argument as in the proof of \Cref{lem:Z_is_negl}, where $\nu_i$ is replaced with $\mu_i$, we obtain 
\begin{align}
   \sum_{\vu\in C^\perp \setminus \{0\}}\prod_{i:u_i\ne 0} |\mu_i(u_i)|\le \negl(n) \label{eq:bound_Z_prime}
\end{align}
assuming \Cref{eq:existence_theorem_cond}.

In a bit more detail, let $Z'$ be the LHS of \Cref{eq:bound_Z_prime}. 
By applying the MDS Brascamp–Lieb inequality in a similar manner to that in the proof of \Cref{lem:Z_is_negl}, we have 
\begin{align*}
   Z'\le t^{-k}\prod_{i\in [n]}\left(1+t^{Q}M'_i\right)^{1/Q}
\end{align*}
where 
\[
M'_i=\sum_{a\in \FF_q\setminus \{0\}}|\mu_i(a)|^Q. 
\]
As shown in \Cref{eq:bound_q}, 
\begin{align*}
 M'_i\le \left(\frac{\tau-\rho}{\rho(1-\rho)}\right)^Q N_{\rho,Q}+O(p^{-1})=M'_{\tau,\rho}+O(p^{-1}),
\end{align*}
where $M'_{\tau,\rho}$ is defined as in \Cref{thm:existence}.  

Thus, by a similar calculation to that in the proof of \Cref{lem:Z_is_negl} with setting
\[
t=\left(\frac{R}{M'_{\tau,\rho}(1-R)}\right)^{1/Q}, 
\]
we obtain 
\begin{align*}
\frac{1}{n}\log Z'
\le (1-R)(H(R)+R \log M'_{\tau,\rho})+o(1), 
\end{align*} 
This implies that $Z'\le \negl(n)$, i.e., \Cref{eq:bound_Z_prime} holds, when \Cref{eq:existence_theorem_cond} holds. 

Combining \Cref{{eq:existential_sum_W_square,eq:bound_Z_prime}}, we have 
\begin{align}
q^{n-k}\sum_{\vx\in C}|W(\vx)|^2
\ge  1-\negl(n). \label{eq:sum_over_C}
\end{align}

When $\tau=s=1$, then 
we obviously have 
\begin{align*}
\sum_{\vx\in C_{sn}}|W(\vx)|^2=0.
\end{align*}
Thus, \Cref{eq:existence_lemma_cond} holds for sufficiently large $n$. 

Below, we assume $\tau<1$ and $s=\tau-\Omega(1)$. 
By the same argument as in the beginning of \Cref{sec:proof_main}, we may assume, without loss of generality, that 
\begin{align}
-(H(R)+R \log M'_{s,\rho})= \Omega(1), \label{eq:existence_theorem_cond_s}
\end{align}
where $M'_{s,\rho}$ is the same as $M'_{\tau,\rho}$ except that $\tau$ is replaced with $s$.

Then by repeating a similar analysis to that in \Cref{sec:bound_four}, 
\Cref{eq:existence_theorem_cond_s} and $s=\tau-\Omega(1)$
give 
\begin{align}
\left|
q^{n-k}\sum_{\vx\in C_{sn}}|W(\vx)|^2
\right|\le \negl(n). \label{eq:sum_over_C_sn}
\end{align}
Comparing \Cref{eq:sum_over_C,eq:sum_over_C_sn}, \Cref{eq:existence_lemma_cond} holds for sufficiently large $n$. 

Thus, by \Cref{lem:existence}, $\MaxLINSAT(q,\mA,\{S_i\}_{i\in[n]},s)$ has a solution for sufficiently large $n$. 
\end{proof}

\section*{Acknowledgments} 
We thank Hayato Arai, Reimi Irokawa, and Tomohiro Yamazaki for helpful discussions, and anonymous reviewers of SODA 2027 for their helpful feedback. 

\section*{Statement on AI Usage} 

While the core ideas underlying the variant of the Regev-like
algorithm presented in \Cref{sec:Regev} and the reduction
of its analysis to Fourier-analytic bounds
(\Cref{thm:Regev_like_new,cor:Regev_like_new})
were developed by the authors, we made extensive use of
interactive discussions with ChatGPT to derive parameter
regimes in which the required conditions can be shown
to hold.  
In particular, ChatGPT helped us identify many components of the proof of \Cref{thm:main}, including the use of a generating polynomial, the MDS Brascamp--Lieb inequality,   and Pollard's theorem.\footnote{In the current version, we use a generalization of Pollard's theorem by Green and Ruzsa~(\Cref{thm:green_ruzsa}).} ChatGPT also assisted in obtaining proofs for the Fourier bias bounds in \Cref{sec:Fourier}, the MDS Brascamp--Lieb inequality, and the Gibbs variational principle. The authors verified, extended, and reorganized these arguments to obtain the final results. We also used ChatGPT to help draft Python code for generating the plots in \Cref{fig:balanced-case-comparison,fig:saturation-bounds,fig:existential_balanced-case-comparison,fig:existential_saturation_bounds}. The resulting codes were verified and modified by the authors. Finally, we used ChatGPT to assist with writing and basic editing. 
The ChatGPT versions used include GPT-6 Astra, 
GPT-5.6 Sol Pro, GPT-5.5 Pro, GPT-5.5 Thinking, GPT-5.4 Pro, and GPT-5.4 Thinking. 
The authors retain full responsibility for all content of the paper.

\newpage
\appendix 
\ifnum\arxiv=0
\section{Thinning Argument}\label{sec:thinning}

We give a simple thinning argument showing that an algorithmic guarantee
for a smaller density can be lifted to any larger density without losing the
target satisfaction rate.

\begin{lemma}[Thinning argument]\label{lem:thinning}
Let $0<\rho'<\rho<1$ be constants. Suppose that there is a quantum
polynomial-time algorithm $\A$ that solves 
$\OPI(k,q,\valpha,\{S'_i\}_{i\in[n]},s)$
with probability at least $1/\poly(n)$ for any distinct
$\alpha_1,\alpha_2,\ldots,\alpha_n\in \FF_q$ and any subsets
$S'_i\subseteq \FF_q$ satisfying 
$|S'_i|/q=\rho'\pm o(1)$
for all $i\in[n]$.

Then there is a quantum polynomial-time algorithm $\B$ that solves
    $\OPI(k,q,\valpha,\{S_i\}_{i\in[n]},s)$ 
with probability at least $1/\poly(n)$ for any distinct
$\alpha_1,\alpha_2,\ldots,\alpha_n\in \FF_q$ and any subsets
$S_i\subseteq \FF_q$ satisfying
    $|S_i|/q=\rho\pm o(1)$
for all $i\in[n]$.
\end{lemma}

\begin{proof}
Let $S_i\subseteq \FF_q$ be subsets satisfying
$|S_i|/q=\rho\pm o(1)$ for all $i\in[n]$. Set
    $$\theta=\frac{\rho'}{\rho}.$$ 
The algorithm $\B$ samples, for each $i\in[n]$, an independent 
random subset $T_i\subseteq \FF_q$ in which each element of $\FF_q$ is
included with probability $\theta$. Then $\B$ implements
quantumly accessible membership oracles for the random sets $T_i$. This can
be done in quantum polynomial time, for example by using Zhandry's
compressed oracle technique~\cite{C:Zhandry19}. 

Define the thinned sets
\[
    S'_i := S_i\cap T_i .
\]
The algorithm $\B$ provides $\A$ with quantum oracle access to the
membership oracle of $S'_i$. 
The algorithm $\B$ runs $\A$ on the instance
    $\OPI(k,q,\valpha,\{S'_i\}_{i\in[n]},s)$, where $\B$ simulates the membership oracle for $S'_i$ using its own membership oracle for $S_i$ and the internal membership oracle for $T_i$,   
and outputs whatever $\A$ outputs.

By the Chernoff bound, we have
\[
    |S'_i|/q=\rho'\pm o(1)
\]
simultaneously for all $i\in[n]$, except with negligible probability.
When this occurs, the worst-case guarantee of $\A$ applies to the sets
$S'_i$. Hence $\A$ outputs, with probability at least
$1/\poly(n)$, a polynomial $P\in \FF_q[X]$ with $\deg(P)<k$ such that
\[
    \bigl|\{i\in[n]: P(\alpha_i)\in S'_i\}\bigr|
    \ge sn .
\]
Since $S'_i\subseteq S_i$ for every $i\in[n]$, the same
polynomial also satisfies
\[
    \bigl|\{i\in[n]: P(\alpha_i)\in S_i\}\bigr|
    \ge sn .
\]
Therefore $P$ is also a valid solution to
\[
    \OPI(k,q,\valpha,\{S_i\}_{i\in[n]},s).
\]
The negligible failure probability from the thinning step only changes the
overall success probability by a negligible amount, so $\B$ succeeds with
probability at least $1/\poly(n)$.
\end{proof}

\section{Comparisons of Existential Bounds}\label{sec:comparison_existential}
Here, we give comparisons between our existential bound and those of \cite{SunWootters26} and DQI~\cite{DQI}. \Cref{fig:existential_balanced-case-comparison} shows the comparison for the case $\rho= 1/2$, and \Cref{fig:existential_saturation_bounds} gives a comparison of the saturation bound for various values of $\rho$.

\begin{figure}[!htbp]
    \centering
    \includegraphics[width=0.8\linewidth]{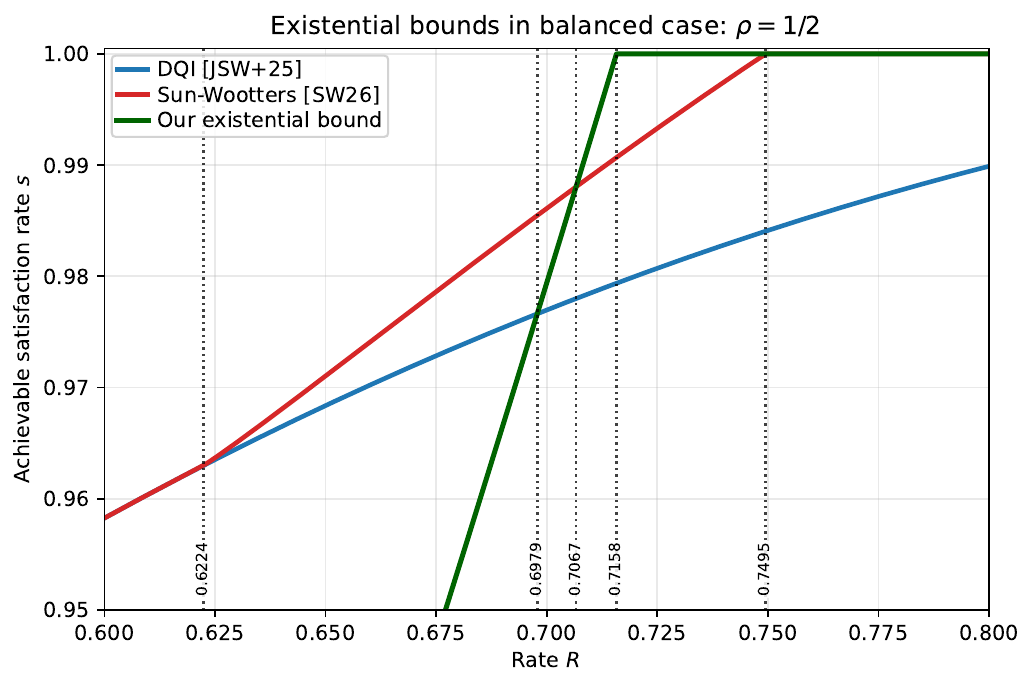}
    \caption{Comparison of existential satisfaction rates in the balanced case $\rho=1/2$.}
    \label{fig:existential_balanced-case-comparison}
\end{figure}
\begin{figure}[!htbp]
    \centering
    \includegraphics[width=0.8\linewidth]{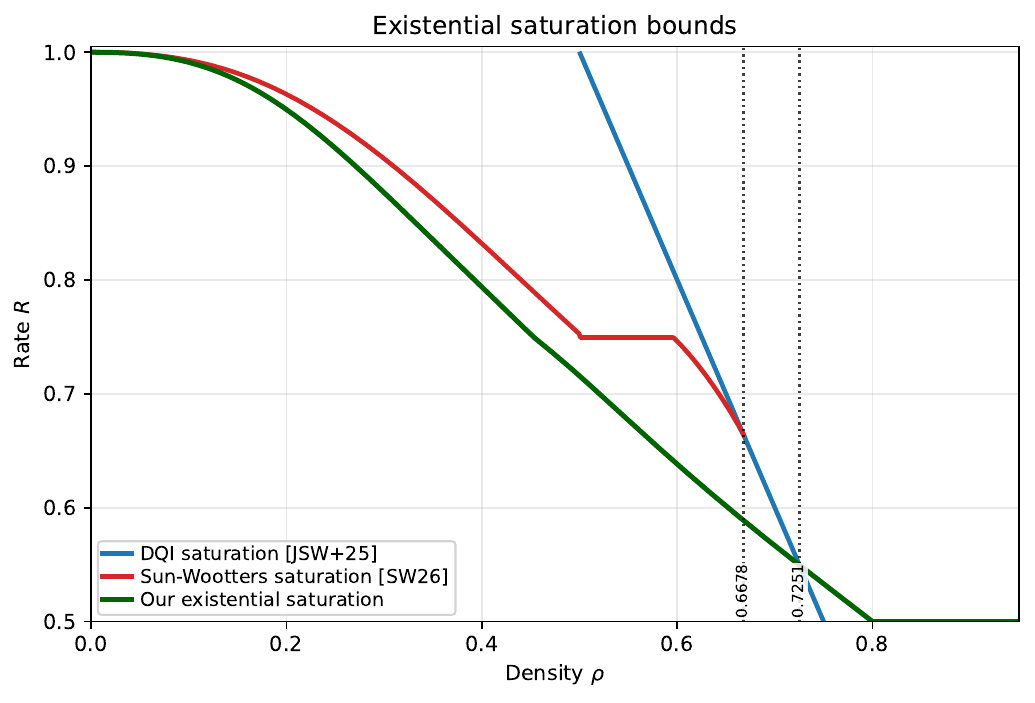}
    \caption{Comparison of existential saturation bounds as functions of $\rho$. The Sun--Wootters curve is based on \cite[Theorems 1.6 and 1.9]{SunWootters26}. While \cite[Theorem 1.9]{SunWootters26} gives a bound for general $\rho$, it has a bump for $\rho>1/2$. On the other hand, \cite[Theorem 1.6]{SunWootters26} gives a better bound in the balanced case $\rho\approx 1/2$. By a thinning argument similar to the one discussed in \Cref{sec:thinning}, one can increase the density $\rho$ without sacrificing the satisfaction rate. Thus, in the regime $\rho\ge 1/2$, whenever the bound from \cite[Theorem 1.9]{SunWootters26} is worse, we plot the thinned balanced-case bound, which gives $R\approx 0.7495$.
    }
    \label{fig:existential_saturation_bounds}
\end{figure}
\section{Proof of \Cref{lem:Gibbs}}\label{sec:proof_Gibbs}
\begin{proof}[Proof of \Cref{lem:Gibbs}] 
Let $P$ be a probability distribution over $X$. 
Let $\tilde{P}$ be a probability distribution over $X$ defined as
\[
\tilde{P}(u)=\frac{2^{f(u)}}{\sum_{v\in X}2^{f(v)}}.
\]
for $u\in X$. 
Multiplying both sides by $\sum_{v\in X}2^{f(v)}$ and taking logarithms, we obtain
\[
f(u)=\log \sum_{v\in X}2^{f(v)}+\log \tilde{P}(u).
\]
for $u\in X$. 
Thus, 
\begin{align*}
\sum_{u\in X}P(u)f(u)
&=\sum_{u\in X}P(u)\left(\log \sum_{v\in X}2^{f(v)}+\log \tilde{P}(u)\right)\\
&=\log \sum_{v\in X}2^{f(v)}+ \sum_{u\in X}P(u)\log\tilde{P}(u)\\
&=\log \sum_{v\in X}2^{f(v)}+\sum_{u:P(u)>0}P(u)\log P(u)+ \sum_{u:P(u)>0}P(u)\left(\log \tilde{P}(u)-\log P(u)\right)\\
&=\log \sum_{v\in X}2^{f(v)}-H(P)+ \sum_{u:P(u)>0}P(u)\log\frac{\tilde{P}(u)}{P(u)}.
\end{align*}

Thus, we have 
\[
\sum_{u\in X}P(u)f(u)+H(P)=\log \sum_{v\in X}2^{f(v)}+ \sum_{u:P(u)>0}P(u)\log\frac{\tilde{P}(u)}{P(u)}.
\]
When $P=\tilde{P}$, $\log\frac{\tilde{P}(u)}{P(u)}=0$ for all $u$ such that $P(u)>0$, and thus the second term in the RHS is equal to $0$.  
This implies
\[
\sup_{P\in \mathcal{P}(X)}
\left\{P(u)f(u)+H(P)
\right\}
\ge \log\sum_{v\in X}2^{f(v)}.  
\]

On the other hand, by concavity of $\log$, Jensen's inequality gives
\begin{align*}
\sum_{u:P(u)>0}P(u)\log\frac{\tilde{P}(u)}{P(u)}
&\le \log\left(\sum_{u:P(u)>0}P(u)\frac{\tilde{P}(u)}{P(u)}\right)\\
&=\log\left(\sum_{u:P(u)>0}\tilde{P}(u)\right)\\
&\le \log\left(\sum_{u\in X}\tilde{P}(u)\right)\\
&=0.
\end{align*}
Thus, we have 
\[
P(u)f(u)+H(P)\le \log \sum_{v\in X}2^{f(v)}
\]
for every $P$. 

This implies 
\[
\sup_{P\in \mathcal{P}(X)}
\left\{P(u)f(u)+H(P)
\right\}
\le \log\sum_{v\in X}2^{f(v)}.  
\]

This completes the proof of \Cref{lem:Gibbs}. 
\end{proof}
\bibliographystyle{alpha}
\bibliography{abbrev3,crypto,reference}
\else
\bibliographystyle{alpha}
\bibliography{abbrev3,crypto,reference}

\section{Proof of \Cref{lem:Gibbs}}\label{sec:proof_Gibbs}
\begin{proof}[Proof of \Cref{lem:Gibbs}] 
Let $P$ be a probability distribution over $X$. 
Let $\tilde{P}$ be a probability distribution over $X$ defined as
\[
\tilde{P}(u)=\frac{2^{f(u)}}{\sum_{v\in X}2^{f(v)}}.
\]
for $u\in X$. 
Multiplying both sides by $\sum_{v\in X}2^{f(v)}$ and taking logarithms, we obtain
\[
f(u)=\log \sum_{v\in X}2^{f(v)}+\log \tilde{P}(u).
\]
for $u\in X$. 
Thus, 
\begin{align*}
\sum_{u\in X}P(u)f(u)
&=\sum_{u\in X}P(u)\left(\log \sum_{v\in X}2^{f(v)}+\log \tilde{P}(u)\right)\\
&=\log \sum_{v\in X}2^{f(v)}+ \sum_{u\in X}P(u)\log\tilde{P}(u)\\
&=\log \sum_{v\in X}2^{f(v)}+\sum_{u:P(u)>0}P(u)\log P(u)+ \sum_{u:P(u)>0}P(u)\left(\log \tilde{P}(u)-\log P(u)\right)\\
&=\log \sum_{v\in X}2^{f(v)}-H(P)+ \sum_{u:P(u)>0}P(u)\log\frac{\tilde{P}(u)}{P(u)}.
\end{align*}

Thus, we have 
\[
\sum_{u\in X}P(u)f(u)+H(P)=\log \sum_{v\in X}2^{f(v)}+ \sum_{u:P(u)>0}P(u)\log\frac{\tilde{P}(u)}{P(u)}.
\]
When $P=\tilde{P}$, $\log\frac{\tilde{P}(u)}{P(u)}=0$ for all $u$ such that $P(u)>0$, and thus the second term in the RHS is equal to $0$.  
This implies
\[
\sup_{P\in \mathcal{P}(X)}
\left\{P(u)f(u)+H(P)
\right\}
\ge \log\sum_{v\in X}2^{f(v)}.  
\]

On the other hand, by concavity of $\log$, Jensen's inequality gives
\begin{align*}
\sum_{u:P(u)>0}P(u)\log\frac{\tilde{P}(u)}{P(u)}
&\le \log\left(\sum_{u:P(u)>0}P(u)\frac{\tilde{P}(u)}{P(u)}\right)\\
&=\log\left(\sum_{u:P(u)>0}\tilde{P}(u)\right)\\
&\le \log\left(\sum_{u\in X}\tilde{P}(u)\right)\\
&=0.
\end{align*}
Thus, we have 
\[
P(u)f(u)+H(P)\le \log \sum_{v\in X}2^{f(v)}
\]
for every $P$. 

This implies 
\[
\sup_{P\in \mathcal{P}(X)}
\left\{P(u)f(u)+H(P)
\right\}
\le \log\sum_{v\in X}2^{f(v)}.  
\]

This completes the proof of \Cref{lem:Gibbs}. 
\end{proof}
\fi
\end{document}